\documentclass[11pt]{article}

\usepackage{amsmath}
\usepackage{amssymb}
\usepackage{amsthm}
\usepackage{mathtools,nccmath}
\usepackage[english]{babel}
\usepackage{amsthm}

\usepackage[round]{natbib}

\usepackage{bibunits}
\defaultbibliographystyle{apalike} 
\defaultbibliography{AbbrvRef}

\theoremstyle{definition}
\newtheorem{theorem}{Theorem}
\newtheorem{definition}{Definition}
\newtheorem{proposition}{Proposition}
\newtheorem{corollary}{Corollary}
\newtheorem{lemma}{Lemma}

\usepackage{graphicx}
\usepackage{subcaption}

\usepackage{booktabs}

\usepackage[colorlinks,citecolor=blue,urlcolor=blue]{hyperref}

\usepackage{fullpage}

\usepackage{algorithm}
\usepackage{algpseudocode}  
\algrenewcommand\algorithmicrequire{\textbf{Input:}}
\algrenewcommand\algorithmicensure{\textbf{Output:}}
\newcommand{\Input}{\Require}
\newcommand{\Output}{\Ensure}

\newcommand{\Bcal}{\mathcal{B}}

\newcommand{\Ecal}{\mathcal{E}}
\newcommand{\Fcal}{\mathcal{F}}

\newcommand{\Kcal}{\mathcal{K}}
\newcommand{\Lcal}{\mathcal{L}}

\newcommand{\Ncal}{\mathcal{N}}

\newcommand{\Pcal}{\mathcal{P}}

\newcommand{\Scal}{\mathcal{S}}
\newcommand{\Tcal}{\mathcal{T}}
\newcommand{\Ucal}{\mathcal{U}}

\def\R{\mathbb{R}} 
\def\E{\mathbb{E}} 
\def\P{\mathbb{P}} 

\def\bq{\boldsymbol{q}}

\def\bR{\boldsymbol{R}}

\usepackage{bm}

\def\bgamma{\boldsymbol{\gamma}}

\def\bmu{\boldsymbol{\mu}}
\def\bnu{\boldsymbol{\nu}}

\def\bomega{\boldsymbol{\omega}}

\def\diag{\textup{diag}}
\DeclareMathOperator*{\cov}{Cov}

\DeclareMathOperator{\Tr}{Tr}

\DeclareMathOperator*{\argmin}{argmin}

\DeclareMathOperator{\supp}{Supp}

\theoremstyle{definition}

\newtheorem{assumption}{Assumption}

\def\s{\Sigma}

\usepackage{enumitem,kantlipsum}

\usepackage[capitalize,noabbrev]{cleveref}

\usepackage{xcolor}

\title{
Fr\'echet regression of multivariate distributions with nonparanormal transport}
\author{Junyoung Park and Irina Gaynanova\thanks{Corresponding author: Irina Gaynanova (irinagn@umich.edu)}}

\date{\normalsize Department of Biostatistics, University of Michigan, Ann Arbor, Michigan, USA
}

\begin{document}

\maketitle
\begin{abstract}
Regression with distribution-valued responses and Euclidean predictors has gained increasing scientific relevance. While methodology for univariate distributional data has advanced rapidly in recent years, multivariate distributions, which additionally encode dependence across univariate marginals, have received less attention and pose computational and statistical challenges. In this work, we address these challenges with a new regression approach for multivariate distributional responses, in which distributions are modeled within the semiparametric nonparanormal family. By incorporating the nonparanormal transport (NPT) metric---an efficient closed-form surrogate for the Wasserstein distance---into the Fr\'echet regression framework, our approach decomposes the problem into separate regressions of marginal distributions and their dependence structure, facilitating both efficient estimation and granular interpretation of predictor effects. We provide theoretical justification for NPT, establishing its bi-Lipschitz equivalence to the Wasserstein distance and proving that it mitigates the curse of dimensionality. We further prove uniform convergence guarantees for regression estimators, both when distributional responses are fully observed and when they are estimated from empirical samples, attaining fast convergence rates comparable to the univariate case. The utility of our method is demonstrated via simulations and an application to continuous glucose monitoring data.
\end{abstract}


\begin{bibunit}

\def\spacingset#1{\renewcommand{\baselinestretch}%
{#1}\small\normalsize}

\section{Introduction}\label{sec:introduction}

Modern studies increasingly collect distributional data, where each data object consists of a sample drawn from an underlying probability distribution. Examples include mortality distributions across countries \citep{ghodratiDistributionondistributionRegressionOptimal2022}, joint distributions of financial asset returns \citep{chenSlicedWassersteinRegression2025}, and distributions of wearable device measurements \citep{matabuenaGlucodensitiesNewRepresentation2021,coulterFastVariableSelection2025}. Building on this growing availability, many scientifically relevant questions are formulated as regression with distribution-valued responses and scalar or Euclidean vector predictors. In particular, our motivating example comes from wearable device studies, where the goal is to investigate associations between multivariate distributional responses constructed from continuous glucose monitoring (CGM) measurements and blood-test biomarkers as predictors.

Multiple methods have been developed for univariate distributional responses. Early approaches employed bijective transformations of univariate distributions into Hilbert spaces, followed by the application of functional data analysis techniques \citep{petersenFunctionalDataAnalysis2016,talskaCompositionalRegressionFunctional2018}. More recently, metric-based approaches on the space of probability measures have gained substantial attention. The Wasserstein distance is a canonical choice because it respects the geometry of the distributional domain 
by measuring the optimal cost of transporting probability mass
\citep{villaniTopicsOptimalTransportation2003}; moreover, in the univariate setting, it admits a quantile-based closed form that yields efficient computation and a convenient $L^2$-space structure for theoretical analysis. Building on this tractability, Fr\'echet regression, a general regression framework for metric space-valued responses, has been successfully applied to the univariate Wasserstein setting \citep{petersenFrechetRegressionRandom2019}, with subsequent generalizations to variable selection \citep{tuckerVariableSelectionGlobal2023}, single-index model \citep{bhattacharjeeSingleIndexFrechet2023}, and non-Euclidean predictor \citep{bhattacharjeeNonlinearGlobalFrechet2025} settings.

However, significantly less progress has been made in the regression of multivariate distributions on $\R^d$ ($d \ge 2$). Multivariate generalization of the aforementioned transformations of univariate distributions is not straightforward \citep{petersenModelingProbabilityDensity2022}. Within metric-based approaches, the multivariate Wasserstein distance lacks a closed-form expression and poses substantial computational and statistical challenges. For two empirical distributions with $N$ samples, the Wasserstein distance computation is $O(N^3)$ \citep{peyreComputationalOptimalTransport2019}, while the convergence rate to its population counterpart is $O(N^{-1/\max\{4,d\}})$ \citep{fournierRateConvergenceWasserstein2015, niles-weedMinimaxEstimationSmooth2022}, reflecting the curse of dimensionality. Furthermore, Fr\'echet regression with the univariate Wasserstein distance \citep{petersenFrechetRegressionRandom2019} cannot be directly generalized to the multivariate setting due to violations of underlying theoretical assumptions, as demonstrated in \citet{fanConditionalWassersteinBarycenters2024}. While alternative distances, such as the Hellinger or Fisher--Rao metrics \citep{zhuGeodesicOptimalTransport2025}, can be used for metric-based analyses, they lack the mass-transport property of the Wasserstein distance and thus do not capture the geometry of the underlying distributional domain.

To retain geometric properties of the multivariate Wasserstein distance while alleviating its obstacles, recent regression methods follow two main lines. One line incorporates computationally cheaper surrogates for the Wasserstein distance---such as the Sinkhorn distance \citep{cuturiSinkhornDistancesLightspeed2013} or the sliced Wasserstein distance \citep{bonnotteUnidimensionalEvolutionMethods2013}---into Fr\'echet regression \citep{fanConditionalWassersteinBarycenters2024, chenSlicedWassersteinRegression2025}. 
However, these methods are sensitive to surrogate-specific hyperparameters, can remain computationally demanding to attain high numerical accuracy, and rely on restrictive assumptions on the support of the random response distribution for theoretical guarantees.
The other line significantly restricts the model space to multivariate Gaussian responses \citep{xuWassersteinFtestsFrechet2025, okanoDistributionondistributionRegressionWasserstein2024}, where the Wasserstein distance admits a closed-form expression, also called the Bures--Wasserstein (BW) metric. This yields substantial computational gains and interpretation via Gaussian parameters, but the model is often too restrictive for real-world data. Consequently, neither line is fully adequate. There remains a need for a methodology that moves beyond restrictive Gaussian assumption while maintaining computational efficiency and theoretical rigor.

In this work, we propose a new regression approach that \textit{decouples} multivariate distributional responses into tractable subcomponents. Copula models naturally support this strategy by separating a multivariate distribution into its univariate marginals and dependence structure. Specifically, we model distributions within the nonparanormal (Gaussian copula) family \citep{liuNonparanormalSemiparametricEstimation2009}, which while still restrictive, is substantially broader than Gaussian family by allowing skewed or heavy-tailed marginals, and it admits the nonparanormal transport (NPT) metric \citep{shaoFastDistanceComputation2026}, a closed-form, tuning-free surrogate of the Wasserstein distance that is faster to compute than other Wasserstein surrogates.
By combining NPT with the Fr\'echet regression framework, we decompose the regression problem into separate Fr\'echet regressions for univariate marginals and for latent Gaussian correlation matrices, equipped with the univariate Wasserstein and BW metrics, respectively. The resulting method, termed \textit{nonparanormal Fr\'echet regression}, is computationally efficient, yields granular interpretability of predictor effects on marginals and dependence, supports component-wise inference, and comes with rigorous theoretical guarantees while alleviating the curse of dimensionality associated with multivariate Wasserstein distance. %

From the computational standpoint, we take advantage of decoupling. For regression of $d$ univariate marginal components, fast computation is based on existing methods with the univariate Wasserstein distance \citep{petersenFrechetRegressionRandom2019}. For regression of the latent Gaussian component, we address a new computational challenge: while the BW metric is closed form, the parameter space must be restricted to correlation matrices, rather than general covariance matrices, preventing application of existing techniques \citep{xuWassersteinFtestsFrechet2025}. We develop a new algorithm based on projected Riemannian gradient descent, which combines the efficient Riemannian method for covariance matrices \citep{chewiGradientDescentAlgorithms2020, fanConvergenceProjectedBuresWasserstein2024} with closed-form projection onto the space of correlation matrices. %

From the theoretical standpoint, our main contributions are three-fold:
\begin{enumerate}[wide, labelindent=0pt, topsep=2pt]
    \item We provide the first rigorous justification of NPT as a surrogate for the multivariate Wasserstein distance, {moving beyond the empirical observations of \citet{shaoFastDistanceComputation2026}.} Our main result is an explicit two-sided bound establishing bi-Lipschitz equivalence between the two distances; the upper-bound on the Wasserstein distance holds under minimal regularity conditions and relies on a novel Hermite-polynomial argument. We also show that NPT shares the geometric properties of the Wasserstein distance and coincides with it exactly under identical dependence structures.
    \item We circumvent the curse of dimensionality of multivariate Wasserstein distance under semiparametric model. Given only finite samples of size $N$, the general Wasserstein error bound has order $N^{-1/d}$, and even under higher-order smoothness the exponent retains its dependence on $d$ \citep{niles-weedMinimaxEstimationSmooth2022}. In contrast, by expanding the nonparanormal family to accommodate a semiparametric estimator, we establish a finite-sample NPT error bound of order $N^{-\alpha}$, where the dimension $d$ enters the rate only through multiplicative constants, while the exponent $\alpha$ is independent of $d$, with $\alpha=1/2$ under additional regularity conditions. Crucially, we are able to transfer this rate to multivariate Wasserstein distance via our theoretical upper-bound on the Wasserstein distance through the NPT.

    \item We establish fast uniform convergence rates for the nonparanormal Fr\'echet regression estimator, covering both fully observed and empirically sampled response distributions. This fills a critical gap in the literature, as existing multivariate methods either assume fully observed distributions or rely on density estimation that may suffer from the curse of dimensionality  \citep{fanConditionalWassersteinBarycenters2024, chenSlicedWassersteinRegression2025, xuWassersteinFtestsFrechet2025}. Our primary theoretical innovation lies in the analysis of the latent correlation regression estimator under the BW metric. In the fully observed case, we obtain a sharp parametric convergence rate under weaker assumptions than both the general Fr\'echet regression theory \citep{petersenFrechetRegressionRandom2019} and the recent covariance matrix regression theory \citep{xuWassersteinFtestsFrechet2025}. We accomplish this through new empirical process and compactness arguments that utilize the correlation matrix constraints.
    In the more challenging case where distributions are estimated from samples, we develop a Hessian-perturbation argument that controls the sampling error and yields a sharp convergence rate. 
    Combining these new developments with a sharper analysis of univariate marginals, we show that our regression estimators achieve fast convergence rates in NPT, comparable to those in univariate distributional regression \citep{petersenFrechetRegressionRandom2019, zhouWassersteinRegressionEmpirical2024}. Finally, through our established bounds for NPT, these fast rates transfer directly to the multivariate Wasserstein distance.
\end{enumerate}

The rest of the paper is organized as follows. Section~\ref{sec:prelim} introduces notation and reviews preliminaries of the Wasserstein distance. Section~\ref{sec:nonpara-npt-metric} reviews the nonparanormal family and the NPT metric, and establishes their theoretical properties, including empirical estimation theory and extension of the nonparanormal family. Section~\ref{sec:frechet-regression} proposes nonparanormal Fr\'echet regression, derives an estimation algorithm, and introduces component-wise model assessment. Section~\ref{sec:theory} establishes asymptotic convergence results for nonparanormal Fr\'echet regression. Section~\ref{sec:simulations} evaluates the performance of the proposed method on synthetic data. Section~\ref{sec:real-data} analyzes continuous glucose monitoring (CGM) data through its multivariate distributional representations and illustrates the component-wise interpretability. Section~\ref{sec:discussion} concludes with a discussion. All proofs are deferred to the supplementary material.

\section{Preliminaries}\label{sec:prelim}
\paragraph{Notation.}
For two numbers $a, b \in \R\cup \{\infty\}$, we write $a\wedge b = \min\{a, b\}$ and $a\vee b = \max\{a,b\}$.
For a vector $x\in\R^d$, we use $\|x\|$ to denote its $\ell^2$-Euclidean norm. We use $1_d\in \R^d$ to denote the vector of ones and $0_d\in\R^d$ the vector of zeros. 
For a matrix $M\in\R^{d\times d}$, let $\diag(M)\in\R^d$ be the vector of diagonal elements of $M$. When $M$ is symmetric, its smallest eigenvalue is denoted by $\lambda_{\min}(M)$. We denote by $\Scal_d = \{M \in \R^{d \times d}: M = M^T\}$ the space of symmetric $d \times d$ matrices, and write $\Scal_d^+ = \{M \in \Scal_d: \lambda_{\min}(M)\ge 0\}$ and $\Scal_d^{++} = \{M \in \Scal_d: \lambda_{\min}(M) > 0\}$ as the cones of positive semidefinite and positive definite matrices, respectively.
Let $\Pcal_2(\R^d)$ denote the space of Borel probability measures $\mu$ on $\R^d$ such that
$\int \|x\|^2 \, d\mu(x) < \infty$. For a measurable map $T:\R^d\to\R^d$, the associated push-forward of a measure $\mu$ on $\R^d$ is denoted by $T_\#\mu$, which satisfies $T_\#\mu(A) = \mu(T^{-1}(A))$ for every measurable set $A\subseteq\R^d$. We use the symbol $\gamma$ to denote the standard Gaussian measure of $\Ncal(0,1)$. Given functions $f$ and $g$, we write $f\circ g$ for their composition, defined by $(f\circ g)(x) = f(g(x))$ whenever the right-hand side is well-defined.

\paragraph{Background on optimal transport.}%

For $\mu, \nu \in \Pcal_2(\R^d)$, the squared 2-Wasserstein distance is defined as
\begin{equation}\label{eq:Wasserstein}
    d_W^2(\mu, \nu) = \inf_{\eta\in\Gamma(\mu, \nu)}\int_{\R^d\times\R^d} \|x-y\|^2 d\eta(x, y),
\end{equation}
where $\Gamma(\mu, \nu)$ is the set of joint distributions on $\R^d\times \R^d$ with marginals $\mu$ and $\nu$. %
We write $d_W(X, Y) = d_W(\mu, \nu)$ for random vectors $X\sim \mu$ and $Y\sim \nu$. If $\mu$ is absolutely continuous (or atomless when $d=1$), $d_W$ is equivalent to:
\begin{equation}\label{eq:Wasserstein-Monge}
    d_W^2(\mu, \nu) = \inf_{T:T_\#\mu = \nu} \int_{\R^d} \|x - T(x)\|^2d\mu(x), %
\end{equation}
where the infimum is over the measurable maps $T:\R^d\to\R^d$ such that $T_\#\mu = \nu$. The problem \eqref{eq:Wasserstein-Monge} attains the infimum at the map $T_\mu^\nu$, which is unique $\mu$-almost surely, called the optimal transport map from $\mu$ to $\nu$. See, e.g., \citet{villaniTopicsOptimalTransportation2003} for additional details.

In general, the Wasserstein distance is computed using algorithmic approaches \citep{peyreComputationalOptimalTransport2019}. Nonetheless, two special cases admit closed-form expressions. First, in the univariate case $d=1$, $d_W(\mu, \nu)$ is expressed using the corresponding quantile functions:
\begin{equation}\label{eq:univariate-wasserstein}
    d_W^2(\mu, \nu) = \|q_\mu - q_\nu \|_{L^2}^2 =  \int_0^1 (q_\mu(p) - q_\nu(p))^2 dp,
\end{equation}
where $q_\mu$ and $q_\nu$ denote the quantile functions of $\mu$ and $\nu$, respectively. This representation yields an isometric $L^2(0,1)$-embedding of univariate distributions as
\begin{equation}\label{eq:quantile-embedding}
    \Pcal(\R)\to \Kcal \subset L^2(0, 1);\qquad \mu\mapsto q_\mu,
\end{equation}
where $\Kcal$ is the cone of quantile functions.
The corresponding optimal transport map $T_\mu^\nu$ is
\begin{equation}\label{eq:univariate-transport}
    T_\mu^\nu = q_\nu \circ F_\mu,
\end{equation}
where $F_\mu$ is the cdf of $\mu$. 
Second, $d_W$ admits a closed-form for multivariate Gaussian measures on $\R^d$; in particular, for centered Gaussian measures $\mu=\Ncal(0, \s)$ and $\nu=\Ncal(0, Q)$, $d_W(\mu, \nu)$ coincides with the Bures--Wasserstein (BW) metric $\Bcal(\s, Q)$ of positive semidefinite matrices: 
\begin{equation}\label{eq:BW-distance}
    d_W^2(\mu, \nu) = \Tr[\s + Q - 2 (\s^{1/2} Q\s^{1/2})^{1/2}] =: \Bcal^2(\s, Q).
\end{equation}
In this case, the optimal transport map $T_\mu^\nu$ becomes a linear map represented by the matrix
\begin{equation}\label{eq:BW-transport}
    T_{\s}^{Q} = \s^{-1/2}(\s^{1/2} Q \s^{1/2})^{1/2}\s^{-1/2},
\end{equation}
defined whenever $\lambda_{\min}(\s) > 0$. The set $\Scal_d^{++}$ admits a Riemannian manifold structure with associated metric $\Bcal$, called the Bures--Wasserstein (BW) manifold \citep{bhatiaBuresWassersteinDistance2019}. %

\section{Nonparanormal transport metric and its properties}\label{sec:nonpara-npt-metric}

\subsection{Nonparanormal distributions and extended domain}\label{sec:nonpara-def-extension}

We begin by reviewing the family of nonparanormal distributions \citep{liuNonparanormalSemiparametricEstimation2009}. 

\begin{definition}\label{def:nonparanormal-npn-d}
    A continuous random vector $X = (X_1,\ldots, X_d)^\top\in\R^d$ with finite second moment follows the \textit{nonparanormal} distribution (Gaussian copula) if there is a collection of monotonically increasing functions $f = (f_1,\ldots,f_d)$ and a positive definite correlation matrix $\Sigma\in\Scal_d^{++}$ such that
    $f(X):=(f_1(X_1),\ldots,f_d(X_d))^\top \sim \Ncal(0_d, \Sigma).$
We write $X\sim {NPN}(f, \Sigma)$ and denote by $NPN(d)$ the set of all nonparanormal distributions in $\Pcal_2(\R^d)$. 
\end{definition}

The marginal transformations in Definition~\ref{def:nonparanormal-npn-d} admit a natural optimal transport representation. Writing $\mu_j$ for the continuous marginal law of $X_j$, the map $f_j$ is uniquely determined as $f_j = T_{\mu_j}^\gamma$, the optimal transport map from $\mu_j$ to the standard Gaussian measure $\gamma$ defined in \eqref{eq:univariate-transport}. Accordingly, the law of $X$ is completely characterized by its flexible continuous marginals $\bmu= (\mu_j)_{j=1}^d$, which can be skewed or heavy-tailed, and the latent Gaussian correlation matrix $\s$, which captures the dependence structure across marginals.

While the continuity of each marginal $\mu_j$ is essential for the transport map $f_j =T_{\mu_j}^\gamma$ to be well-defined, in practice one typically observes only finite samples from underlying continuous distributions, making observed empirical distributions discrete. To accommodate this, we extend the nonparanormal domain to allow discrete marginals. We rely on the key observation: the inverse transport map $T_\gamma^{\mu_j}$ of $T_{\mu_j}^\gamma$ is well-defined even when $\mu_j$ is discrete, because the source measure $\gamma$ is continuous. This motivates our reversed-transport characterization of nonparanormal distributions from latent Gaussian distributions as follows.
\begin{definition}\label{def:extension map}
    Let $\Pcal:=\Pcal_2(\R)$ and define $\Ecal_d= \{M\in\Scal_d^+: \diag(M) = 1_d\}$ as the set of correlation matrices. On the product space $\Pcal^d \times \Ecal_d$, where $\Pcal^d$ denotes the $d$-fold Cartesian product of $\Pcal$, we define the map $\Lambda:\Pcal^d \times \Ecal_d\to \Pcal_2(\R^d)$ as
    \begin{equation*}%
        \Lambda(\bmu, \s) = (T_{\bgamma}^{\bmu})_\# \Ncal(0_d, \s), 
    \end{equation*}
    where $\bmu = (\mu_j)_{j=1}^d \in\Pcal^d$, $\s \in\Ecal_d$, and $T_{\bgamma}^{\bmu}(x) := (T_{\gamma}^{\mu_1}(x),\ldots, T_{\gamma}^{\mu_d}(x))^\top$. %
\end{definition}

When each $\mu_j$ is continuous, $\Lambda(\bmu, \s)$ recovers the nonparanormal distribution $NPN(f, \s)$ with $f_j= T_{\mu_j}^\gamma$. Using the map $\Lambda$, we define the following extended nonparanormal domain.

\begin{definition}\label{def:extended-nonparanormal}
    The set of extended nonparanormal distributions is defined by $$\Lambda(d):= \Lambda(\Pcal_*^d\times \Ecal_d) \subset \Pcal_2(\R^d),$$
    where $\Pcal_*\subset \Pcal$ is the subset of distributions with nonzero variance.
\end{definition}
The domain $\Lambda(d)$ accommodates discrete marginals while excluding only point masses, whose zero variance is incompatible with the unit diagonals of latent correlation matrices. It contains the continuous nonparanormal family $NPN(d) = \Lambda(\Pcal_c^d\times \Ecal_d^{++})$ as a subset, where $\Pcal_c\subset\Pcal_*$ denotes the set of continuous distributions and $\Ecal_d^{++} = \Ecal_d \cap \Scal_d^{++}$. %

\subsection{Nonparanormal transport metric}\label{sec:npt-properties}

The nonparanormal transport (NPT) metric, introduced by \citet{shaoFastDistanceComputation2026}, is a metric on $NPN(d)$ that combines closed-form Wasserstein distances for the marginal and latent Gaussian components. For empirical analysis, we extend it to our enlarged domain $\Lambda(d)$:

\begin{definition}\label{def:npt}
    Let $\mu = \Lambda(\bmu, \s) \in\Lambda(d)$ be a measure with marginals $\bmu = (\mu_j)_{j=1}^d\in\Pcal_*^d$, defined as in Definition~\ref{def:extension map}, 
    and let $\nu = \Lambda(\bnu, Q)\in \Lambda(d)$ defined similarly. The squared NPT metric between $\mu$ and $\nu$ is defined by
    \begin{equation*}%
        d_{NPT}^2(\mu, \nu):= \sum_{j=1}^{d} d_W^2(\mu_{j}, \nu_{j}) + \Bcal^2(\s, Q).
    \end{equation*}
    We also write $d_{NPT}(X, Y) = d_{NPT}(\mu, \nu)$ for random vectors $X\sim\mu$ and $Y\sim\nu$.
\end{definition}

Verifying that NPT remains a metric on the extended domain $\Lambda(d)$ is more technical, since copula uniqueness, which establishes $d_{NPT}(\mu, \nu) = 0 \iff \mu = \nu$ on $NPN(d)$, does not extend to discrete marginals. Nonetheless, the latent Gaussian structure enables the following injectivity result, which shows that NPT is a valid metric on $\Lambda(d)$, isometric to the product metric on $\Pcal_*^d\times \Ecal_d$ ($\Pcal_*$ endowed with $d_W$ and $\Ecal_d$ with $\Bcal$):

\begin{lemma}\label{lemma:injectivity}
    The map $\Lambda:\Pcal^d \times \Ecal_d\to \Pcal_2(\R^d)$ is injective on $\Pcal_*^d\times \Ecal_d$.
\end{lemma}

The closed-form NPT metric is substantially faster to compute than other Wasserstein surrogates while closely approximating the Wasserstein distance, as \citet{shaoFastDistanceComputation2026} empirically observed. In this work, we provide theoretical justifications for using NPT instead of the Wasserstein distance. Specifically, we establish: (i) similar geometric properties of the two metrics, (ii) metric agreement under identical dependence structures, and (iii) two-sided bounds implying topological equivalence. We begin with basic geometric properties:
\begin{proposition}\label{prop:geometric-properties}
    Let $X\sim \mu$ and $Y\sim \nu$ for measures $\mu = \Lambda(\bmu, \s)$ and $\nu = \Lambda(\bnu, Q)$ in $\Lambda(d)$. For a vector $x\in\R^d$ and a scalar $a > 0$, it holds that 
    \begin{enumerate}
        \item [(i)] $d_{NPT}(X + x, Y + x) = d_{NPT}(X, Y)$; 
        \item [(ii)] If $\E[X] = \E[Y]$, then $d_{NPT}^2(X + x, Y) = \|x\|^2 + d_{NPT}^2(X, Y)$;
        \item [(iii)] $d_{NPT}^2(aX, aY) = a^2\sum_{j=1}^{d} d_W^2(X_{j}, Y_{j}) + \Bcal^2(\s, Q)$.
    \end{enumerate}
\end{proposition}
Proposition~\ref{prop:geometric-properties} confirms that NPT mirrors the geometry-respecting properties of the Wasserstein distance:
(i) and (ii) demonstrate that NPT captures location differences between $\mu$ and $\nu$, while (iii) reflects the scale-invariance of the latent Gaussian part, with scale information captured only through the marginal components. While property (iii) suggests that the marginal distributions of $\mu$ and $\nu$ may need to be rescaled to prevent the BW term in $d_{NPT}$ from being dominated by the marginal distances, this causes no issues for our regression method in Section~\ref{sec:frechet-regression} due to separate handling of each component of NPT.

Next, we present a case in which NPT and the Wasserstein distance coincide.
For continuous distributions sharing a common copula, the Wasserstein distance is known to depend only on the marginals \citep{ghaffariParsevalsIdentityOptimal2021}. We prove that this fact remains true on the enlarged domain $\Lambda(d)$, yielding the equality $d_{NPT}=d_W$ when the latent correlation matrices agree.

\begin{proposition}\label{prop:same-copula}
Let $\mu = \Lambda(\bmu, \s)$ and $\nu = \Lambda(\bnu, \s)$ be measures in $\Lambda(d)$, sharing the same latent correlation matrix $\s$. Then,
$$d_{NPT}^2(\mu, \nu) = d_W^2(\mu, \nu) = \sum_{j=1}^d d_W^2(\mu_j, \nu_j).$$
\end{proposition}

Since $d_W^2(\mu, \nu)\ge \sum_{j=1}^d d_W^2(\mu_{j}, \nu_{j})$ always holds, Proposition~\ref{prop:same-copula} clarifies that the BW term in NPT targets the additional dependence-driven distributional differences at the latent Gaussian level. We further formalize this intuition by establishing explicit two-sided bound on multivariate Wasserstein distance via NPT.

To establish the two-sided bound, we consider additional notation. For a univariate function $\psi:\R\to\R$, we define its Sobolev seminorm $|\psi|_{H^1(\gamma)} = \big(\int_\R |\psi'|^2 d\gamma\big)^{1/2}$ under the standard Gaussian measure $\gamma$, and we set $|\psi|_{H^1(\gamma)}=\infty$ when $\psi$ is not differentiable or the integral diverges. 
We let $|\psi|_{\text{Lip}} = \sup_{z\in\R} |\psi'(z)|$ denote the Lipschitz seminorm, taking values in $\R\cup \{\infty\}$. 
For a collection $\Psi = (\psi_1, \ldots, \psi_d)$ of functions, we write $|\Psi|_{H^1(\gamma)} = \max_{j} |\psi_j|_{H^1(\gamma)}$ and $|\Psi|_{\text{Lip}} = \max_j|\psi_j|_{\text{Lip}}$, which always satisfy $|\Psi|_{H^1(\gamma)} \le |\Psi|_{Lip}$.

\begin{theorem}\label{thm:topological-equivalence}
    Let $\mu = \Lambda(\bmu, \s)$ and $\nu = \Lambda(\bnu, Q)$ be measures in $\Lambda(d)$. It holds that
    \begin{enumerate}
    \item[(i)] $d_W(\mu,\nu)\le \sqrt2\big\{1\vee\big(|T_{\bgamma}^{\bmu}|_{H^1(\gamma)}\wedge|T_{\bgamma}^{\bnu}|_{H^1(\gamma)}\big)\big\}\, d_{NPT}(\mu,\nu)$.
    \end{enumerate}
    Moreover, if $\mu,\nu\in NPN(d)$, let $T_{\bmu}^{\bgamma} = (T_{\mu_j}^\gamma)_{j=1}^d$ denote the transport map that corresponds to $f$ in Definition~\ref{def:nonparanormal-npn-d}, and define $T_{\bnu}^{\bgamma}$ analogously. Then it also holds that
    \begin{enumerate}
    \item[(ii)] $d_{NPT}(\mu,\nu)\le \big\{1+2\big(|T_{\bmu}^{\bgamma}|_{\mathrm{Lip}}\wedge |T_{\bnu}^{\bgamma}|_{\mathrm{Lip}}\big)\big\}\, d_W(\mu,\nu)$.
    \end{enumerate}
\end{theorem}

Theorem~\ref{thm:topological-equivalence} establishes bi-Lipschitz equivalence between $d_{NPT}$ and $d_W$ on $NPN(d)$.
Among the two inequalities, bound (i) is particularly consequential: it holds on the extended domain $\Lambda(d)$ and allows asymptotic convergence rates established under $d_{NPT}$ to transfer directly to those under $d_W$. Crucially, since Theorem~\ref{thm:topological-equivalence} requires regularity of \textit{only one} of the two distributions $\mu$ and $\nu$, the Sobolev condition $|T_{\bgamma}^{\bmu}|_{H^1(\gamma)} < \infty$ needs to hold only for the fixed population distribution, leaving the empirical counterpart completely unconstrained.

A major technical contribution of Theorem~\ref{thm:topological-equivalence} is establishing bound~(i) under this Sobolev regularity $|T_{\bgamma}^{\bmu}|_{H^1(\gamma)} < \infty$, which is substantially milder than the Lipschitz condition $|T_{\bgamma}^{\bmu}|_{\text{Lip}} < \infty$ frequently imposed in the optimal transport literature \citep{hutterMinimaxEstimationSmooth2021, manolePluginEstimationSmooth2024}. While proving an analogue of bound (i) using the Lipschitz seminorm $|T_{\bgamma}^{\bmu}|_{\text{Lip}}$ is relatively straightforward, it requires uniformly bounded derivatives of $T_\gamma^{\mu_j}$ and thus restricts $\mu$ to sub-Gaussian tails. In contrast, the Sobolev condition permits the derivatives to diverge if their growth is controlled after Gaussian weighting. Consequently, NPT upper bounds the Wasserstein distance over a broad class of nonparanormal distributions, including those with (1) log-concave marginals, (2) log-normal marginals, and (3) Student's $t$ marginals with degrees of freedom $>2$, with the latter two being heavy-tailed (verified in Section~\ref{sec:supp-gaussian-sobolev-verify} of the supplementary material). Proving bound~(i) under this relaxed condition is technically nontrivial; we develop a new Hermite-polynomial argument that tightly traces how changes in latent Gaussian dependence propagate through the nonlinear transport map $T_{\bgamma}^{\bmu}$.

\subsection{Distribution estimation and fast convergence}\label{sec:npt-empirical}

To estimate a nonparanormal distribution $\mu= \Lambda(\bmu, \s)$ from its i.i.d. samples $\{X_k=(X_{k1},\ldots, X_{kd})^\top\}_{k=1}^N\subset\R^d$, we define its semiparametric estimator $\hat\mu$ as
\begin{equation}\label{eq:nonpara-estimator}
    \hat\mu = \Lambda(\hat\bmu, \hat\s),
\end{equation}
where $\hat\bmu = (\hat\mu_j)_{j=1}^d$ with $\hat\mu_j$ the empirical distribution of $\{X_{kj}\}_{k=1}^N$, and $\hat\s$ is the rank-based estimator of $\s$ computed via Kendall's $\tau$ followed by sine transformation \citep{liuHighdimensionalSemiparametricGaussian2012}.

We investigate the convergence of the semiparametric estimator $\hat\mu$. We assume throughout the paper that the underlying distribution $\mu$ belongs to $NPN(d)$, i.e., it has continuous marginals and $\lambda_{\min}(\s) > 0$, and that the empirical marginals $\hat\mu_j$ have nonzero variances (thus ensuring $\hat\mu\in\Lambda(d)$). 
By developing a matrix perturbation bound that controls $\Bcal^2(\hat\s, \s)$ without requiring the empirical estimator $\hat\s$ to be strictly positive definite, we establish the following explicit finite-sample bound of the NPT estimation error $d_{NPT}^2(\hat\mu, \mu)$.

\begin{theorem}\label{thm:npt-convergence-rate}
    Let $\mu = \Lambda(\bmu, \s)\in NPN(d)$ be a continuous nonparanormal distribution, and let $\hat \mu\in\Lambda(d)$ be its estimator from $N$ samples constructed as in \eqref{eq:nonpara-estimator}. 
    Define
    $r_N=\max_{1\le j\le d}\sqrt{\E[d_W^2(\hat\mu_j,\mu_j)]}$. 
    Then,
    \begin{align*}
        \E[d_{NPT}(\hat\mu,\mu)] &\le \sqrt{d}\cdot r_N + \sqrt{\frac{2\pi^2\, d(d-1)}{N \,\lambda_{\min}(\s)}}. %
    \end{align*}
\end{theorem}

Consequently, as $N\to\infty$ with $d$ fixed, $\E[d_{NPT}(\hat\mu,\mu)]$ achieves the convergence rate $O(r_N)$, which is the rate of convergence for univariate distributions. For another measure $\nu\in NPN(d)$ and its estimator $\hat\nu$ with $N$ samples defined as in \eqref{eq:nonpara-estimator}, the triangle inequality implies the same $O(r_N)$ rate for the convergence of $d_{NPT}(\hat\mu, \hat\nu)$ to $d_{NPT}(\mu, \nu)$.
Precise rates of $r_N$ depend on the regularity of the marginals $\mu_j$, for which \citet{bobkovOnedimensionalEmpiricalMeasures2019} provides a comprehensive survey; for instance, the rate achieves $O(N^{-1/2})$ (up to $\log N$ factors) under mild conditions such as log-concavity.
Therefore, the rate of $\hat\mu$ in Theorem~\ref{thm:npt-convergence-rate} is considerably faster than the $O(N^{-1/(4\vee d)})$ rate for estimating $\mu$ in the Wasserstein distance via its empirical measure \citep{fournierRateConvergenceWasserstein2015}. Notably, by Theorem~\ref{thm:topological-equivalence}, our estimator $\hat \mu$ also alleviates this slow rate in Wasserstein distance as formalized below.

\begin{corollary}\label{cor:wasserstein-estimation}
    Under the same conditions as Theorem~\ref{thm:npt-convergence-rate}, assume additionally that 
    $|T_{\bgamma}^{\bmu}|_{H^1(\gamma)} < \infty$.
    Then, the following bound holds, achieving the $O(r_N)$ rate as $N\to\infty$ with $d$ fixed:
    \begin{align*}
        \E[d_W(\hat\mu, \mu)] &\le \sqrt{2}\big(1 \vee |T_{\bgamma}^{\bmu}|_{H^1(\gamma)}\big)\Bigg( \sqrt{d}\cdot r_N + \sqrt{\frac{2\pi^2\, d(d-1)}{N \,\lambda_{\min}(\s)}}\Bigg). %
    \end{align*}
\end{corollary}

In the optimal transport literature, overcoming the slow $O(N^{-1/(4\vee d)})$ rate in the Wasserstein distance has primarily been pursued by assuming higher-order smoothness on $\mu$; however, the minimax optimal rate in those settings still retains a $d$-dependent exponent of $N$ that deteriorates rapidly in higher dimensions \citep{niles-weedMinimaxEstimationSmooth2022}. Corollary~\ref{cor:wasserstein-estimation} demonstrates that this $d$-dependent exponent can be entirely circumvented by imposing a parametric structure solely on the latent dependence, while leaving univariate marginals fully nonparametric under the minimal Gaussian--Sobolev condition. This new semiparametric mitigation of the curse of dimensionality in $d_W$ may be of independent interest.

\section{Nonparanormal Fr\'echet Regression}\label{sec:frechet-regression}

In this section, we turn to the regression problem with multivariate distributional responses and Euclidean predictors, under the nonparanormal model for the response distributions. Let $(Z,\omega)$ be a random pair on $\R^p\times NPN(d)$, where $Z$ is a random predictor vector and $\omega$ is a random nonparanormal measure. As in Section~\ref{sec:nonpara-def-extension}, we write $\omega=\Lambda(\bomega,S)$, where $\bomega=(\omega_1,\ldots,\omega_d)\in\Pcal_c^d$ is the vector of random continuous marginals and $S\in\Ecal_d^{++}$ is the random latent correlation matrix. We adopt the Fr\'echet regression framework \citep{petersenFrechetRegressionRandom2019} to regress the response $\omega$ in the metric space $\Lambda(d)$ endowed with $d_{NPT}$, where the regression function is defined as the conditional Fr\'echet mean of $\omega$ given $Z=z$:
\begin{equation*}
    \omega^*(z) = \argmin_{\mu\in \Lambda(d)} \E[d_{NPT}^2(\omega, \mu) | Z = z].
\end{equation*}
To target the conditional Fr\'echet mean $\omega^*(z)$, we consider the global Fr\'echet regression function, which is given by the weighted Fr\'echet mean:
\begin{equation}\label{eq:frechet-reg}
    \omega_F (z) = \argmin_{\mu\in \Lambda(d)} \E[s(Z, z)d_{NPT}^2(\omega, \mu)],
\end{equation}
with the linear weight function $s(Z, z)$ defined by
$$s(Z, z) = 1 + (Z - \E[Z])^\top \cov(Z)^{-1} (z - \E[Z]),$$
where $\cov(Z) = \E[(Z - \E[Z])(Z - \E[Z])^\top]$. This weight satisfies $\E[s(Z,z)]=1$ and is extended from a reformulation of classical linear regression. The resulting $\omega_F(z)$ minimizes the best linear approximation (in $Z$) to $ \E[d_{NPT}^2(\omega,\mu)|Z]$, and therefore closely approximates $\omega^*(z)$ \citep{bhattacharjeeNonlinearGlobalFrechet2025}. In this paper, we focus on estimating the global regression function $\omega_F(z)$, but our approach also applies to local linear or global nonlinear extensions \citep{petersenFrechetRegressionRandom2019, bhattacharjeeNonlinearGlobalFrechet2025} by replacing $s(Z,z)$ with nonlinear weight functions, aiming at closer approximations of $\omega^*(z)$. 

In our nonparanormal setting, the additive structure of the NPT metric enables \textit{decoupling} the regression target \eqref{eq:frechet-reg}. Specifically, the minimization problem decomposes as:

\begin{align*}
    \omega_F(z) 
    &=\argmin_{\Lambda(\bmu, \s)\in\Lambda(d)} \left(\sum_{j=1}^d \E[s(Z, z)d_W^2(\omega_j, \mu_j)] + \E[s(Z, z)\Bcal^2(S, \s)] \right)\\
    &= \Lambda\bigg(\Big(\argmin_{\mu_j\in \Pcal_*} \E[s(Z, z)d_W^2(\omega_j, \mu_j)] \Big)_{j=1}^d,\ \argmin_{\s\in\Ecal_d} \E[s(Z, z)\Bcal^2(S, \s)] \bigg).
\end{align*}
Thus, the regression reduces to the following $d$ univariate Fr\'echet regressions for the marginals and one Fr\'echet regression for latent correlation matrices:
\begin{align}\label{eq:decoupled-frechet-population}
    \begin{split}
        \omega_{F, j}(z) &= \argmin_{\mu\in\Pcal_*}\E[s(Z, z) d_W^2(\omega_j, \mu)], \quad j=1,\ldots, d,\quad \text{and} \\
        {S_F}(z) &= \argmin_{\s\in\Ecal_d} \E[s(Z, z) \Bcal^2(S, \s)],
    \end{split}
\end{align}
which together construct the nonparanormal distribution $\omega_F(z) = \Lambda(\bomega_{F}(z), S_F(z))$, where $\bomega_{F}(z) = (\omega_{F, j}(z))_{j=1}^d$.
We term this procedure \textit{nonparanormal Fr\'echet regression}. 

This decoupled structure has not only computational but also interpretability advantages. Unlike Fr\'echet regression with the Wasserstein-variants \citep{fanConditionalWassersteinBarycenters2024, chenSlicedWassersteinRegression2025}, where predictor effects are interpreted only through the single multivariate distributional object, the decoupling in~\eqref{eq:decoupled-frechet-population} allows to disentangle the predictor's effect between individual marginal distributions $\omega_j$ and the latent dependence structure $S$.

\subsection{Estimation of nonparanormal Fr\'echet regression}\label{sec:estimation}

Given the decoupled scheme \eqref{eq:decoupled-frechet-population}, we separately estimate marginals $\omega_{F, j}(z)$ and the latent correlation matrix $S_F(z)$ to estimate the full distribution $\omega_F(z)$. Suppose that $\{(Z_i, \omega_i)\}_{i=1}^n$ is an i.i.d. sample from $(Z, \omega)$, with $\omega_i$ represented by $\omega_i= \Lambda(\bomega_i,  S_i)$, where $\bomega_i = (\omega_{ij})_{j=1}^d$. Write $\overline{Z} = n^{-1}\sum_{i=1}^n Z_i$ and $\widehat{\cov}(Z) = n^{-1}\sum_{i=1}^n(Z_i - \overline{Z})(Z_i - \overline{Z})^\top$, which define the empirical weight function as
$s_n(Z_i, z) = 1 + (Z_i - \overline{Z})^\top \widehat{\cov}(Z)^{-1}(z-\overline{Z})$. 

We first consider the oracle case, where the nonparanormal distributions $\omega_i$ are assumed to be fully observed. The separate estimators of the decoupled target \eqref{eq:decoupled-frechet-population} are given by:
\begin{align}\label{eq:oracle-decoupled-estimators}
    \begin{split}
        \tilde\omega_{F, j}(z) &= \argmin_{\mu\in\Pcal_*}\frac{1}{n}\sum_{i=1}^n s_n(Z_i, z) d_W^2(\omega_{ij}, \mu),
    \quad j=1,\ldots, d,\quad \text{and} \\ %
    \tilde  S_F(z) &=\argmin_{ \s\in\Ecal_d} \frac{1}{n}\sum_{i=1}^n s_n(Z_i, z)\Bcal^2(S_i, \s), %
    \end{split}
\end{align}
which defines the distributional estimator %
    as $\tilde \omega_{F}(z):= \Lambda\big((\tilde\omega_{F, j})_{j=1}^d, \tilde S_F(z)\big).$

In practice, we assume that for each $i$, $N_i$ independent samples $\{Y_{ik}\}_{k=1}^{N_i}\subset\R^d$ are drawn from $\omega_i$, yielding the empirical distribution estimate $\hat \omega_i = \Lambda(\hat\bomega_i, \hat S_i)$ as constructed in Section~\ref{sec:npt-empirical}. Then, the separate regression estimators are analogously defined as
\begin{align}
    \hat\omega_{F, j}(z) &= \argmin_{\mu\in\Pcal_*}\frac{1}{n}\sum_{i=1}^n s_n(Z_i, z) d_W^2(\hat\omega_{ij}, \mu),
    \quad j=1,\ldots, d,\quad \text{and} \label{eq:marginal-estimation} \\ 
    \hat S_F(z) &=\argmin_{ \s\in\Ecal_d} \frac{1}{n}\sum_{i=1}^n s_n(Z_i, z)\Bcal^2(\hat S_i, \s), \label{eq:correlation-estimation}
\end{align}
yielding the full estimator
$    \hat \omega_{F}(z):= \Lambda\big((\hat\omega_{F, j})_{j=1}^d, \hat S_F(z)\big)$.
To ensure the existence of $\omega_F(z)$ and its estimators, we require the following assumption on the marginal components.
\begin{assumption}\label{assump:non-constant-marginals}
    For every $j=1,\ldots, d$ and every $z\in\R^p$ in the support of $Z$, $\omega_{F, j}(z)$, $\tilde \omega_{F, j}(z)$, and $\hat\omega_{F, j}(z)$ exist in $\Pcal_*$, with the latter two almost surely.
\end{assumption}
This assumption is mild: since the objective functions for $\omega_{F, j}(z)$, $\tilde \omega_{F, j}(z)$, and $\hat\omega_{F, j}(z)$ attain minima over the entire domain $\Pcal$ \citep{petersenFrechetRegressionRandom2019}, it only excludes the possibility that the minimizers are point masses. The existence of correlation components $S_F(z)$, $\tilde S_F(z)$, and $\hat S_F(z)$ is guaranteed since $\Bcal^2$ is continuous on the compact domain $\Ecal_d$.

\subsection{Computation algorithm}\label{sec:computation}

We describe a computational procedure for the estimator $\hat\omega_F(z)$, computed separately for the marginal component \eqref{eq:marginal-estimation} and the latent correlation component \eqref{eq:correlation-estimation}. 

For the marginals $\hat \omega_{F, j}(z)$, we utilize the quantile-based closed formula \eqref{eq:univariate-wasserstein}. Let $\hat q_{ij}$ denote the empirical quantile function of the measure $\hat\omega_{ij}$ and let  $\overline{q}_j := n^{-1} \sum_{i=1}^n s_n(Z_i, z)  \hat q_{ij}$.
For each $j$, the quantile corresponding to $\hat \omega_{F,j}(z)$ is obtained as
\begin{equation*}
    \hat q_{F, j}(z) 
    = \argmin_{q\in\Kcal} \|\overline{q}_j - q\|_{L^2}^2,
\end{equation*}
where $\Kcal$ is the cone of quantile functions defined in \eqref{eq:quantile-embedding}. Following prior works \citep{petersenFrechetRegressionRandom2019, zhouWassersteinRegressionEmpirical2024}, we solve this optimization by approximating the $L^2$ norm via an $M$-equispaced grid on $(0, 1)$. Denoting by $\overline{\bq}_j \in\R^M$ the vector of $\overline{q}_j$ evaluated at each point on the grid, the optimization becomes a quadratic program:
$$\argmin_{\bq = (q_1,\ldots,q_M)^\top \in\R^M} \|\overline{\bq}_j - \bq \|^2, \quad\text{subject to}\quad q_1\le\cdots \le q_M,$$
whose solution $\hat\bq_j$ represents a discretized approximation of the quantile function $\hat q_{F, j}(z)$.
We use the R package \texttt{fastfrechet} \citep{coulterFastfrechetPackageFast2025} with a grid size of $M=200$ for implementation, which additionally handles support constraints for univariate marginals as in our motivating application in Section~\ref{sec:real-data}.

For the latent correlation estimator $\hat S_F(z)$, we develop projected Riemannian gradient descent on the BW manifold $\Scal_d^{++}$ of covariance matrices, where each Riemannian gradient descent step is followed by an efficient projection step onto the correlation set $\Ecal_d$ as shown in Algorithm~\ref{alg:proj-Riem-GD}. Lines 3 and 4 correspond to the Riemannian gradient descent algorithm for computing barycenters on the BW manifold $\Scal_d^{++}$ \citep{chewiGradientDescentAlgorithms2020}.
Line 3 computes the Riemannian gradient $G^{(t)}$ of the objective function in \eqref{eq:correlation-estimation}, where $T_{S^{(t-1)}}^{\hat S_i}$ is the Gaussian optimal transport matrix as defined in \eqref{eq:BW-transport}. The gradient $G^{(t)}$ lies in the tangent space of $\Scal_d^{++}$ at $S^{(t-1)}$, and line 4 retracts $G^{(t)}$ to $\Sigma^{(t)}$ on the manifold $\Scal_d^{++}$.
Then, since $\s^{(t)}$ is not generally a correlation matrix, we additionally perform the projection step to $\Ecal_d$ in line 5; here, the projection map $P_{\Ecal_d}$ is defined within the BW metric geometry and admits a simple, intuitive closed-form, the \textit{symmetric normalization}, which we prove in Lemma~\ref{lemma:BW-projection}.

\begin{algorithm}[t]
\spacingset{1}
\caption{Projected Riemannian gradient descent for regression of correlation matrices} 
\label{alg:proj-Riem-GD}  
\begin{algorithmic}[1] 
\Input Pairs of predictor vector and estimated correlation matrix $\{(Z_i, \hat{S}_i)\}_{i=1}^n\subset \R^p\times \Ecal_d$; evaluation point $z\in\R^p$; stepsize $\eta > 0$; max iterations $T_{\max} > 0$; tolerance $\varepsilon > 0$. 
\Output The estimated Fr\'echet regression function $S^{(t)} = \hat  S_F(z)$ for latent correlations.
\State \textit{Initialize:} $S^{(0)}\gets  \hat S_1$
\For{$t=1,\dots,T_{\max}$} 
\State 
$G^{(t)} \gets I_d - \eta\Big(I_d - \frac{1}{n}\sum_{i=1}^ns_n(Z_i, z) T_{S^{(t-1)}}^{\hat S_i}\Big)$
\State
$\Sigma^{(t)} \gets G^{(t)}S^{(t-1)}G^{(t)}$
\State $S^{(t)} \gets P_{\Ecal_d} \big(\Sigma^{(t)}\big) := \argmin_{Q\in\Ecal_d} \Bcal\big(\Sigma^{(t)}, Q\big) $ %
\If {$\|S^{(t)} - S^{(t-1)}\| < \varepsilon$} %
\State \textbf{break}
\EndIf
\EndFor 
\end{algorithmic} 
\end{algorithm}

\begin{lemma}\label{lemma:BW-projection}
    For any matrix $\s\in\Scal_d^+$, the projection $P_{\Ecal_d} (\Sigma) = \argmin_{Q\in\Ecal_d} \Bcal(\Sigma, Q)$ satisfies
    $$P_{\Ecal_d}(\s) = D(\s)^{-1/2}\s D(\s)^{-1/2},$$
    where $D(\s)$ is a diagonal matrix with $\diag(D(\s)) = \diag(\s)$.
\end{lemma}

Lemma~\ref{lemma:BW-projection} follows from the parametrization $M\in\R^{d\times d}\mapsto M^\top M \in \Scal_d^{+}$ \citep{bhatiaBuresWassersteinDistance2019}, under which column-wise normalization of $M$ yields the projection $P_{\Ecal_d}$. This closed-form projection enables fast implementation; we empirically found that Algorithm~\ref{alg:proj-Riem-GD} converges much faster than a projected gradient method using Frobenius projection onto $\Ecal_d$, which lacks a closed-form and is often computed by alternating minimization \citep{highamComputingNearestCorrelation2002}. %

We fix stepwise $\eta = 1$ throughout, motivated by the bivariate setting, where one iteration of Algorithm~\ref{alg:proj-Riem-GD} attains the exact solution $\hat S_F(z)$, except in rare extreme cases ($\hat S_F(z)=\pm 1$ or the objective in \eqref{eq:correlation-estimation} becomes constant) that are easy to diagnose from data. See Section~\ref{sec:supp-proof-one-step-eqcorr} of the supplementary material for details and a similar property of equicorrelation matrices.

\subsection{Component-wise assessment and permutation-based inference}\label{sec:compo-wise-r2}

To assess goodness-of-fit in Fr\'echet regression, \citet{petersenFrechetRegressionRandom2019} introduced the generalized $R^2$ by extending classical $R^2$ in linear regression. %
With NPT, it is defined as:
$$\hat R^2_F = 1 - \frac{\sum_{i=1}^nd_{NPT}^2(\hat\omega_i, \hat \omega_F(Z_i))}{\sum_{i=1}^nd_{NPT}^2(\hat\omega_i, \hat\omega_F(\overline{Z}))} = 1 - \frac{\sum_{j=1}^d\left(\sum_{i=1}^nd^2_W(\hat\omega_{ij}, \hat \omega_{F, j}(Z_i))\right) + \sum_{i=1}^n\Bcal^2(\hat S_i, \hat{S}_F(Z_i))}{\sum_{j=1}^d\left(\sum_{i=1}^nd^2_W(\hat\omega_{ij}, \hat\omega_{F, j}(\overline{Z}))\right) + \sum_{i=1}^n\Bcal^2(\hat S_i, \hat{S}_F(\overline{Z}))},
$$
which represents the fraction of total distributional variability explained by the regression fit. However, this single scalar summary can be misleading. Since the denominator and nominator aggregate variability across components (marginals and dependence structure), the overall $\hat R_F^2$ is dominated by components with large variability; if such a component has a poor fit, it may yield a small $\hat R_F^2$ even if the fit in other components is good. This limitation is not specific to NPT, but rather inherent to any scalar-based goodness-of-fit measure for multivariate distributional data, which collapses information into a single number.

To overcome this issue, we instead propose \textit{component-wise} $R^2$:
$$\hat R_j^2:= 1 - \frac{\sum_{i=1}^nd^2_W(\hat\omega_{ij}, \hat \omega_{F, j}(Z_i))}{\sum_{i=1}^nd^2_W(\hat\omega_{ij}, \hat\omega_{F, j}(\overline{Z}))},\quad j=1,\ldots, d,  \quad\text{and}\quad  \hat R_S^2:= 1 - \frac{\sum_{i=1}^n\Bcal^2(\hat S_i, \hat{S}_F(Z_i))}{\sum_{i=1}^n\Bcal^2(\hat S_i, \hat{S}_F(\overline{Z}))},$$
which naturally aligns with the decoupled scheme of nonparanormal Fr\'echet regression. These components form the vector $\hat\bR_F^2 = (\hat R_1^2,\ldots,\hat R_d^2,\hat R_S^2)$ of generalized $R^2$ values, which prevents information collapse and enables a more intuitive, component-wise evaluation of the predictor effect on marginal components and their dependence structure.

We then propose a permutation-based inference for the vector $\hat\bR_F^2$.
To test the null hypothesis of no predictor effect on each component, for each permutation $b=1,\ldots, B$, we randomly permute the predictors $\{Z_i\}_{i=1}^n$ to obtain $\{Z_i^{(b)}\}_{i=1}^n$, refit the nonparanormal Fr\'echet regression using the permuted data $\{(Z_i^{(b)},\hat\omega_i)\}_{i=1}^n$, and recompute the component-wise coefficients $\hat\bR_F^{2,(b)} = (\hat R_1^{2,(b)},\ldots,\hat R_d^{2,(b)},\hat R_S^{2,(b)})$. The resulting collection $\{\hat\bR_F^{2,(b)}\}_{b=1}^B$ provides empirical null distributions for each component. Since we test $d+1$ hypotheses simultaneously, we control the family-wise error rate (FWER) using the Westfall-Young min-$p$ adjustment \citep{westfallResamplingbasedMultipleTesting1993}. 
While permutation testing with thousands of replicates is typically computationally prohibitive in multivariate distributional data, our efficient correlation regression algorithm (Algorithm~\ref{alg:proj-Riem-GD}), combined with the \texttt{fastfrechet} package \citep{coulterFastfrechetPackageFast2025} for marginal regressions, makes this procedure scalable: in our real-data analysis with $n=968$ trivariate distributions (Section~\ref{sec:real-data}), computing all $n=968$ fitted distributions in a single permutation replicate takes under $10$ seconds on a single CPU core.

\section{Convergence of nonparanormal Fr\'echet regression}\label{sec:theory}

We study convergence of the nonparanormal Fr\'echet regression estimators $\tilde \omega_F(z)$ and $\hat \omega_F(z)$ defined in Section~\ref{sec:estimation}. Beyond their pointwise convergence to the population counterpart $\omega_F(z)$ at a single evaluation point $z\in\R^p$, we establish stronger uniform convergence rates over the region $\{z\in\R^p:\|z\|\le B\}$ for a fixed constant $B>0$, a common goal in the Fr\'echet regression literature \citep{petersenFrechetRegressionRandom2019, zhouWassersteinRegressionEmpirical2024, chenSlicedWassersteinRegression2025}.

Establishing these uniform rates requires novel theoretical development. Although \citet{petersenFrechetRegressionRandom2019} develop a general theory for Fr\'echet regression of metric space-valued responses, it requires fully observed distributional responses and relies on assumptions that can be restrictive; in particular, these assumptions do not hold for general multivariate distributions under the Wasserstein distance \citep{fanConditionalWassersteinBarycenters2024}. 
Rather than adopting this general theory, we derive a new convergence result for the oracle estimator $\tilde\omega_F(z)$ based on fully observed $\omega_i$ under NPT. Our proof avoids the condition on the empirical objective function imposed by \citet{petersenFrechetRegressionRandom2019} and yields the sharper parametric convergence rate. We then establish convergence of the estimator $\hat\omega_F(z)$ based on responses $\hat\omega_i$ estimated from $N_i$ samples from each $\omega_i$, where controlling distribution-wise sampling error within the Fr\'echet regression framework is nontrivial and has not been analyzed even in the simpler Gaussian setting \citep{xuWassersteinFtestsFrechet2025}. Finally, combined with Theorem~\ref{thm:topological-equivalence}, these NPT rates imply the same sharp convergence rates in the Wasserstein distance.

We begin by analyzing the oracle estimator $\tilde \omega_F(z) = \Lambda(\tilde\bomega_F(z), \tilde{S}_F(z))$, where distributional responses $\omega_i$ are fully observed. We analyze its convergence via the decomposition
$$d_{NPT}^2(\tilde \omega_F(z), \omega_F(z)) = \sum_{j=1}^d d_W^2(\tilde\omega_{F, j}(z), \omega_{F, j}(z)) + \Bcal^2(\tilde{S}_F(z), S_F(z)).$$
The marginal terms $\tilde\omega_{F, j}(z)$ are handled via the isometric $L^2$-space embedding \eqref{eq:quantile-embedding}. The correlation term $\tilde{S}_F(z)$ requires a more technical analysis, as the BW metric has curved geometry and lacks an isometric Hilbert space embedding \citep{massartCurvatureManifoldFixedRank2019}. To analyze $\tilde{S}_F(z)$, we let $F(z,\s):=\E[s(Z,z)\,\Bcal^2(S,\s)]$ denote the objective function for the population random pair $(Z, S)$ in \eqref{eq:decoupled-frechet-population} and impose the following regularity assumptions.

\begin{assumption}\label{assump:correlation-mineigen}
    There exists a constant $\lambda_0 >0$ such that (i) $\lambda_{\min}(S) \ge \lambda_0$ almost surely, and (ii) $\lambda_{\min}(S_F(z))\ge \lambda_0$ for all $\|z\|\le B$.
\end{assumption}

\begin{assumption}\label{assump:correlation-uniqueness}
    The population regression $S_F(z)$ in~\eqref{eq:decoupled-frechet-population} is unique for every $\|z\|\le B$.
\end{assumption}

\begin{assumption}\label{assump:local-convexity-radius}
    Let $H(z)$ denote the Hessian of the functional $\s \mapsto F(z, \s)$ evaluated at $S_F(z)$, viewed as a bilinear form on the tangent space $\mathcal{T} :=\{M \in\Scal_d: \diag(M) = 0_d\}$ of $\Ecal_d$. There exists a constant $\kappa > 0$ such that 
    $H(z)[M, M] \ge \kappa \|M\|_F^2$ for all $M\in \mathcal{T}$ and $\|z\|\le B$.
\end{assumption}

Assumption~\ref{assump:correlation-mineigen} is a mild condition requiring the random response $S$ and the regression function $S_F(z)$ to be strictly positive definite, {and Assumption~\ref{assump:correlation-uniqueness} guarantees uniqueness of the population regression target $S_F(z)$.}
Assumption~\ref{assump:local-convexity-radius} ensures strong convexity of the objective $F$ locally at the optimum $S_F(z)$ on the tangent space $\Tcal$, {an assumption analogous to condition~(U2) in \citet{petersenFrechetRegressionRandom2019} and Assumption~5 in \citet{xuWassersteinFtestsFrechet2025}.
Crucially, our theory avoids assuming global separation conditions for $M$-estimation used in prior work (condition (U0) in \citet{petersenFrechetRegressionRandom2019} and Assumption 4 in \citet{xuWassersteinFtestsFrechet2025}).} 
Furthermore, in Section~\ref{sec:supp-verify-assump-equicorr} of the supplementary material, we demonstrate that Assumption~\ref{assump:correlation-mineigen} is sufficient for Assumptions~\ref{assump:correlation-uniqueness}--\ref{assump:local-convexity-radius} to hold under weak dependence between $Z$ and $S$ and in the bivariate case.

We now establish convergence rates for the oracle case, as formalized below.

\begin{theorem}\label{thm:frechet-oracle-rate}
    Suppose Assumption~\ref{assump:non-constant-marginals} holds. Then, as $n\to\infty$ with $d$ fixed,
    $$\sup_{\|z\|\le B} d_W(\tilde\omega_{F,j}(z), \omega_{F,j}(z)) = O_p(n^{-1/2}), \quad j=1,\ldots,d.$$
    If Assumptions~\ref{assump:correlation-mineigen}--\ref{assump:local-convexity-radius} additionally hold, then
    $$\sup_{\|z\|\le B} \Bcal(\tilde{S}_F(z), S_F(z)) = O_p(n^{-1/2}),
    \quad\text{and thus}\quad
    \sup_{\|z\|\le B} d_{NPT}(\tilde\omega_F(z), \omega_F(z)) = O_p(n^{-1/2}).$$
\end{theorem}

The exact parametric $O_p(n^{-1/2})$ rates in Theorem~\ref{thm:frechet-oracle-rate} improve upon the $O_p(n^{-1/(2+\varepsilon)})$ rate (for any $\varepsilon>0$) established for general metric space-valued responses \citep[Theorem~2]{petersenFrechetRegressionRandom2019}.
For the marginal components, we obtain this sharper rate by using the isometric $L^2$-space embedding \eqref{eq:quantile-embedding}, Corollary~2 of \citet{petersenFrechetRegressionRandom2019} for Hilbert space-valued responses, and projection onto the quantile cone $\Kcal$.
For the correlation component, we use a Taylor expansion of the oracle objective $\tilde F_n(z, \cdot)$ on $\Tcal$, then translate the $O_p(n^{-1/2})$ rate of its tangent-projected gradient into the same rate for $\Bcal(\tilde{S}_F(z), S_F(z))$. 
While the general Taylor expansion approach parallels the covariance regression analysis of \citet{xuWassersteinFtestsFrechet2025} on the unbounded cone $\Scal_d^+$, we use compactness and different empirical process arguments over $\Ecal_d$ to establish a sharp parametric rate without the logarithmic factors appearing in \citet{xuWassersteinFtestsFrechet2025}, while avoiding their uniform separation condition (Assumption 4 therein).

Next, we analyze the estimator $\hat \omega_F(z) = \Lambda(\hat\bomega_F(z), \hat S_F(z))$, where the responses $\hat\omega_i$ are estimated from $N_i$ independent samples from each $\omega_i$.
For this asymptotic analysis, we require the minimum sample size $N = \min_i N_i$ to grow sufficiently fast relative to $\log n$:

\begin{assumption}\label{assump:observation-sizes}
    As $n\to\infty$, $N := \min\{N_1,\ldots,N_n\}$ satisfies $N\to \infty$ and $(\log n)/N \to 0$.%
\end{assumption}

Under this additional condition, we derive convergence rates in the empirical-response setting.
Let $\supp(\omega)$ denote the support of the law of the random response $\omega \in NPN(d)$, and for each marginal $\mu_j$ of $\mu\in\supp(\omega)$, let $\hat\mu_j^{(m)}$ denote the empirical distribution formed from $m$ i.i.d. observations from $\mu_j$.

\begin{theorem}\label{thm:frechet-real-rate}
    Suppose Assumptions~\ref{assump:non-constant-marginals} and \ref{assump:observation-sizes} hold, and there exists a decreasing sequence $r_m\to 0$ as $m\to\infty$ such that 
    $\max_{1\le j\le d}\sup_{\mu\in\supp(\omega)} \E\!\big[d_W^2\!\big(\hat\mu_j^{(m)},\mu_j\big)\big]
    \le r_m^2$ for $m\ge 1$.
    Then, as $n\to\infty$ with $d$ fixed,
    $$\sup_{\|z\|\le B} d_W(\hat\omega_{F,j}(z), \omega_{F,j}(z)) = O_p(n^{-1/2} + r_N), \quad j=1,\ldots,d.$$
    If Assumptions~\ref{assump:correlation-mineigen}--\ref{assump:local-convexity-radius} additionally hold, then
    $$ \sup_{\|z\|\le B} \Bcal(\hat S_F(z), S_F(z)) = O_p\big(n^{-1/2} + N^{-1/2}\big),
    \text{ and }
    \sup_{\|z\|\le B} d_{NPT}(\hat\omega_F(z), \omega_F(z)) = O_p\big(n^{-1/2} + r_N\big). $$
\end{theorem}

We prove Theorem~\ref{thm:frechet-real-rate} by bounding the discrepancy between $\hat\omega_F(z)$ and its oracle counterpart $\tilde\omega_F(z)$.
The marginal discrepancy $d_W(\hat\omega_{F, j}(z), \tilde\omega_{F,j}(z))$ achieves the $O_p(r_N)$ rate, following the isometric $L^2$-space embedding \eqref{eq:quantile-embedding} argument of \citet{zhouWassersteinRegressionEmpirical2024}. For the correlation part, the curved BW geometry necessitates a different approach. To control empirical sampling error, we develop a 
Hessian-perturbation argument leveraging the convexity of $\Ecal_d$ to bound $\Bcal(\hat S_F(z), \tilde S_F(z))$ by the empirical--oracle gradient discrepancy, whose $O_p(N^{-1/2})$ rate concludes the proof of Theorem~\ref{thm:frechet-real-rate}.

Finally, we translate the established NPT convergence rates into the Wasserstein distance using Theorem~\ref{thm:topological-equivalence}, which requires only a mild Sobolev regularity condition on $\omega_F(z)$.

\begin{corollary}\label{cor:wasserstein-rate}
    Suppose $\sup_{\|z\|\le B}|T_{\bgamma}^{\bomega_F(z)}|_{H^1(\gamma)} < \infty$ and Assumptions~\ref{assump:non-constant-marginals}--\ref{assump:local-convexity-radius} hold. Then, as $n\to\infty$ with $d$ fixed,
    $$\sup_{\|z\|\le B} d_W(\tilde\omega_F(z), \omega_F(z)) = O_p(n^{-1/2}).$$
    If Assumption~\ref{assump:observation-sizes} additionally holds, then with $r_N$ as defined in Theorem~\ref{thm:frechet-real-rate},
    $$\sup_{\|z\|\le B} d_W(\hat\omega_F(z), \omega_F(z)) = O_p(n^{-1/2} + r_N).$$
\end{corollary}

The first oracle Wasserstein rate of Corollary~\ref{cor:wasserstein-rate} is stronger than existing results on multivariate distributional regression by allowing a substantially more flexible response model than Gaussian-based approaches \citep{xuWassersteinFtestsFrechet2025}, establishing uniform rather than pointwise convergence \citep{fanConditionalWassersteinBarycenters2024}, and improving on sliced Wasserstein rates \citep{chenSlicedWassersteinRegression2025}, whose translation to Wasserstein rates incurs rate deterioration \citep{bonnotteUnidimensionalEvolutionMethods2013}.
Moreover, the second empirical Wasserstein rate in Corollary~\ref{cor:wasserstein-rate} is the first such result for empirically sampled responses stated directly in the multivariate Wasserstein distance.

\section{Simulations}\label{sec:simulations}

We use synthetic data to evaluate the performance of our nonparanormal Fr\'echet regression method (NPT-FR), and compare it with two alternative methods: \textit{Marginal Fr\'echet Regression (Marginal-FR)}, which fits univariate Fr\'echet regression separately to each marginal and ignores dependence; and \textit{Gaussian Fr\'echet Regression (Gaussian-FR)}, which assumes multivariate Gaussian responses \citep{xuWassersteinFtestsFrechet2025}. 
We omit the sliced Wasserstein approach \citep{chenSlicedWassersteinRegression2025} because we were unable to find publicly available code, and we omit the Sinkhorn distance method \citep{fanConditionalWassersteinBarycenters2024} as it is computationally infeasible in our $d=10$ setting due to its reliance on equidistant rectangular grids for a compact subset of $\R^d$, and number of computations growing exponentially with $d$.

We consider dimensions $d \in\{2, 10\}$ for distributional responses $\omega_i\in NPN(d)$. 
We generate $n \in \{50, 100, 200\}$ i.i.d. pairs $\{(Z_i, \omega_i)\}_{i=1}^n$ of predictor and nonparanormal distribution, where the predictor $Z_i\in\R^2$ is indepentently drawn from $\Ucal[-1,1]^2$, and the nonparanormal distributional response $\omega_i$ is drawn from a random measure $\omega = \Lambda(\bomega, S)\in NPN(d)$, which is specified by marginals and correlations conditional on $Z$ as follows.
   
   \textbf{Marginals.}
    Let $Y = (Y_{1}, \ldots, Y_{d})^\top$ be a random vector following the response distribution $\omega$. To induce skewness, the $j$th marginal distribution is specified such that
    $$(-1)^j Y_{j} \sim \text{Gamma}(2, \theta_{j}), \quad\text{where}\quad  \theta_{j}|Z \sim \text{Gamma}((\sigma +\beta_j^\top Z)^2/c, c/(\sigma + \beta_j^\top Z))$$
    with $\sigma=3$ and $c=1$. We set $\beta_1 = (0.5, 0)^\top$, $\beta_2 = (0.4, -0.3)^\top$, and with remaining $\beta_j$ ($j\ge 3$, $d=10$) drawn once from $\Ucal[-0.5, 0.5]^2$. The marginals correspond to signed Gamma distributions with scale parameters depending, on average, linearly on $Z$.
   
   \textbf{Latent correlations.} For $d=2$, we generate the random correlation $\rho$ of $S$ in two ways: with a  random noise $\varepsilon\sim \Ncal(0, 0.1^2)$ independent of $Z$, we consider $\rho = 0.3 Z^{(1)} + \varepsilon$ (linear) and $\rho = \tanh(2 Z^{(2)} + \varepsilon)$ (nonlinear). For $d=10$, to ensure positive definiteness, we use a matrix-exponential map and the projection $P_{\Ecal_d}$ (Lemma~\ref{lemma:BW-projection}):
    $ S = P_{\Ecal_d}\left( \exp\left(Z^{(1)} M_1 + Z^{(2)} M_2 + E \right) \right), $
    where $M_1, M_2$ are symmetric matrices with zero diagonals (generated once with off-diagonal entries from $\Ucal[-0.5, 0.5]$), behaving as linear coefficient matrices. The noise $E \in \Scal_d$ has zero diagonals and off-diagonal entries are i.i.d. $\Ncal(0, 0.1^2)$. The exponent has components largely falling within $(-1, 1)$, making $S$ approximately linear in $Z$.

As a result, the predictor vector $Z$ influences both marginals and latent correlation matrices of $\omega$. Finally, from each $\omega_i\sim \omega$, we generate empirical distributions $\hat\omega_i$ of sample size $N \in \{100, 1000\}$. 
For each $n$, $N$, and $d$, we consider $n_{\text{rep}}=100$ Monte Carlo simulations.

We evaluate the performance of the regression fit for each method using the mean squared prediction error (MSPE) on an independent test set $\{(Z_l^{\text{te}}, \omega_l^{\text{te}}=\Lambda(\bomega_l^{\text{te}}, S_l^{\text{te}}))\}_{l=1}^{n_{\text{te}}}$ of size $n_{\text{te}}=500$, generated i.i.d. from $(Z, \omega)$. For each $j$th marginal $\omega_{lj}^{\text{te}}$ of $\omega_{l}^{\text{te}}$, we let $q_{lj}^{\text{te}}$ denote the corresponding true quantile. We let $\hat q_{F, j}$ denote the fitted quantile corresponding to the marginal fit $\hat \omega_{F, j}$. 
To separately assess the model fit on the marginal components and the latent correlation component, as in Section~\ref{sec:compo-wise-r2}, we consider decoupled errors:
$$ \text{MSPE}_{\text{marg}} = \frac{1}{d} \sum_{j=1}^d\frac{1}{n_{\text{te}}} \sum_{l=1}^{n_{\text{te}}}  \| \hat q_{F, j}(Z_l^{\text{te}}) - q_{l,j}^{\text{te}} \|_{L^2}^2 \quad\! \text{and}\! \quad \text{MSPE}_{\text{corr}} = \frac{1}{n_{\text{te}}} \sum_{l=1}^{n_{\text{te}}} \Bcal^2(\hat  S_F(Z_l^{\text{te}}), S_l^{\text{te}}), $$
where marginal errors are averaged across $d$ components. 
The $L^2$ norms are evaluated on the equispaced grid of $M=200$ points in $(0,1)$. For Marginal-FR, we compute $\text{MSPE}_{\text{corr}}$ by setting $\hat  S_F(Z_l^{\text{te}}) = I_d$. For Gaussian-FR, the marginal quantiles $\hat q_{F, j}(Z_l^{\text{te}})$ and the latent correlation $\hat  S_F(Z_l^{\text{te}})$ are derived from the fitted Gaussian distribution.

\begin{figure}[t]
    \centering
    \spacingset{1}
    \includegraphics[width=0.98\linewidth]{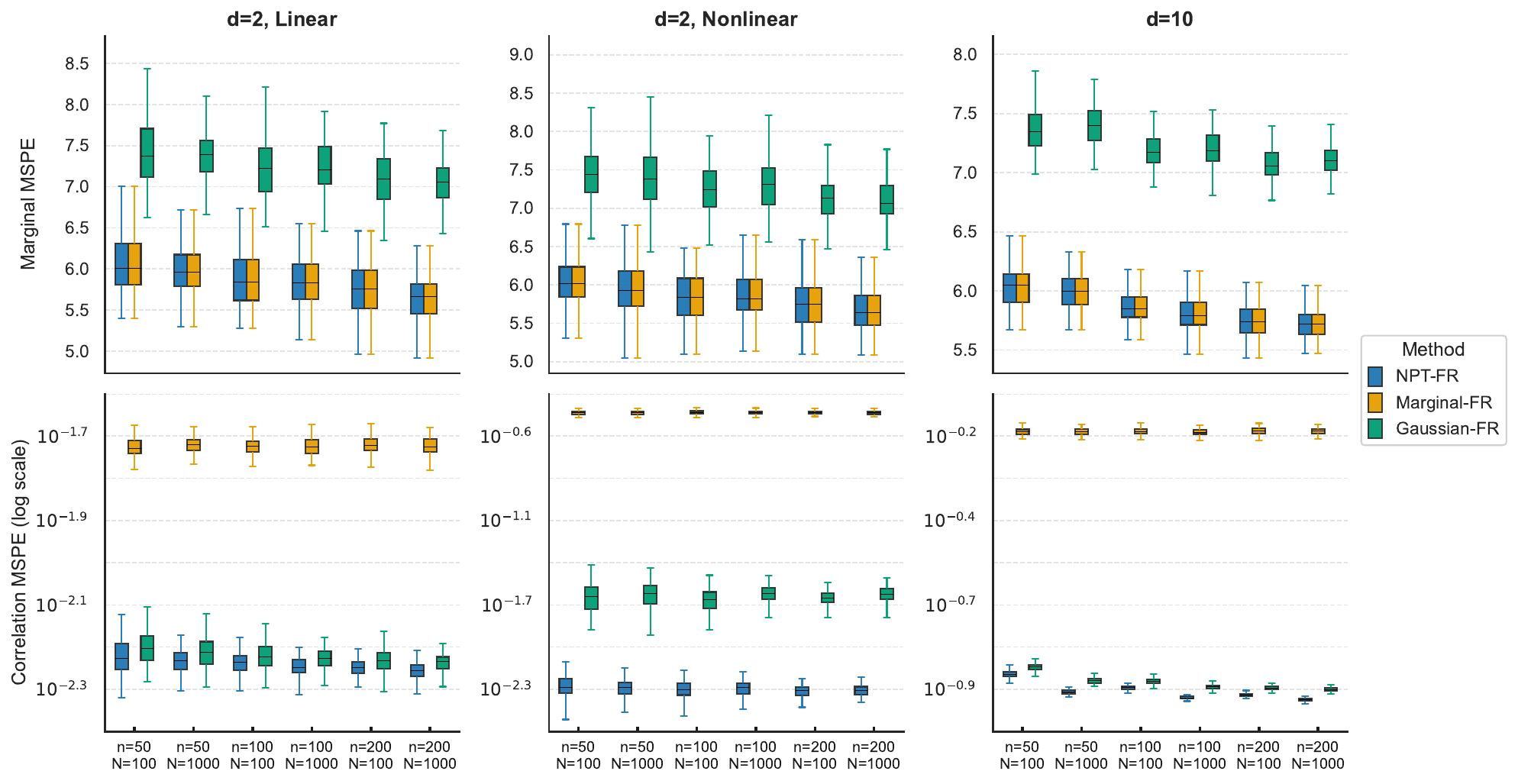}
    \caption{Out-of-sample mean squared prediction errors (MSPEs) on $n^{\text{te}}=500$ test samples over $n_{\text{rep}}=100$ replicates. Columns correspond to $d=2$ with linear correlation, $d=2$ with nonlinear correlation, and $d=10$. Within each column, boxplots report $\text{MSPE}_{\text{marg}}$ (top) and $\text{MSPE}_{\text{corr}}$ (bottom, log scale) across the $(n, N)$ settings shown on the $x$-axis.}
    \label{fig:simulation-boxplots}
\end{figure}

Figure~\ref{fig:simulation-boxplots} reports the out-of-sample errors for the marginal and correlation components. As expected, Marginal-FR achieves identical $\text{MSPE}_{\text{marg}}$ to NPT-FR but exhibits substantially larger $\text{MSPE}_{\text{corr}}$, since it ignores dependence structure changes induced by $Z$. Gaussian-FR yields large marginal errors because the skewed Gamma marginals violate the Gaussian assumption; its correlation fit is comparable to NPT-FR when correlations are generated linearly ($d=2$) or nearly linearly ($d=10$) in $Z$, but degrades under the nonlinear correlation structure. Overall, NPT-FR provides the best joint performance, with both error metrics decreasing as $n$ and $N$ increase; increasing $N$ shows relatively modest gains, especially for $\text{MSPE}_{\text{corr}}$ in the bivariate case $d=2$, reflecting a low sample complexity of estimating a single latent correlation. In Section~\ref{sec:supp-auxiliary-simulation} of the supplementary material, we also compare MSPE using the Wasserstein distance, where NPT-FR consistently outperforms the competitors.

\begin{figure}[t]
    \centering
    \spacingset{1}
    \includegraphics[width=0.9\linewidth]{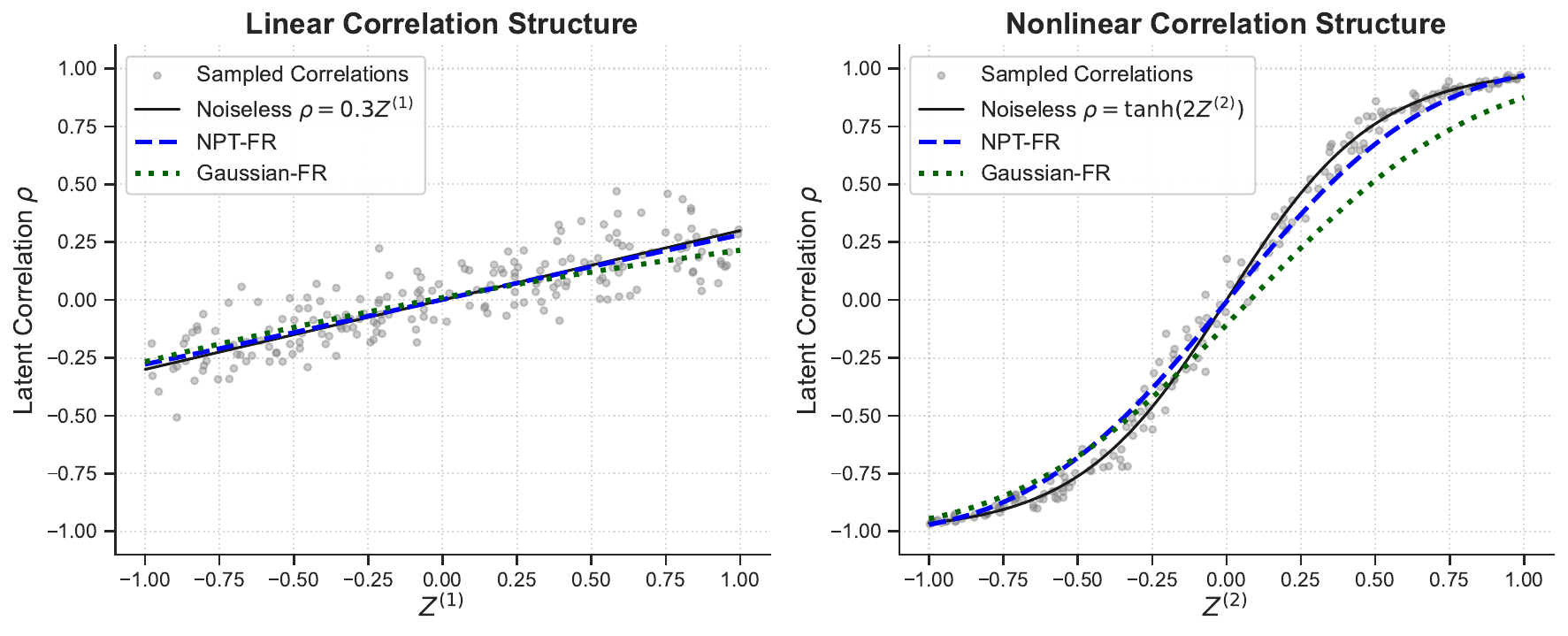}
    \caption{Bivariate correlation regression in one Monte Carlo replicate with $(n, N) = (200, 1000)$. Gray points are the generated correlations plotted against the associated predictor component, and solid black curves show the noise-free generating functions $0.3Z^{(1)}$ (left) and $\tanh(2Z^{(2)})$ (right). The fitted curves from NPT-FR (dashed) and Gaussian-FR (dotted) are overlaid.}
    \label{fig:correlation-simulation}
\end{figure}

Figure~\ref{fig:correlation-simulation} further illustrates the bivariate correlation fit in one replicate with $(n, N) = (200, 1000)$ and provides a closer look at $\text{MSPE}_{\text{corr}}$ in Figure~\ref{fig:simulation-boxplots}. In the linear setting, both NPT-FR and Gaussian-FR track the noise-free baseline $\rho=0.3Z^{(1)}$, with Gaussian-FR exhibiting a mild bias toward zero. In the nonlinear setting, NPT-FR follows the curved trend induced by $\rho=\tanh(2Z^{(2)})$ more closely, whereas Gaussian-FR produces an attenuated fit, consistent with the larger $\text{MSPE}_{\text{corr}}$. Furthermore, while global Fr\'echet regression generalizes linear regression, these curved regression fits of NPT-FR highlight the intrinsic curvature of the BW metric; this is further illustrated in Section~\ref{sec:supp-one-step-bivariate} of the supplementary material.

\section{Application to Continuous Glucose Monitoring Data}\label{sec:real-data}

We analyze continuous glucose monitoring (CGM) data from the AI-READI study \citep{ai-readiconsortiumAIREADIRethinkingAI2024, ai-readiconsortiumFlagshipDatasetType2024} to investigate how CGM-derived glycemic distributions relate to clinical biomarkers in individuals across the diabetes spectrum (no diabetes, prediabetes, and type 2 diabetes).
The study collected blinded Dexcom G6 CGM recordings of interstitial glucose levels (every 5 minutes) within the $[40, 400]$~mg/dL measurement limit for 10 days per participant. To ensure data quality, we excluded individuals with $>$5\% of readings at 400~mg/dL, or $<4.9$ days of data (70\% of one week). Values $\ge 400$~mg/dL were treated as missing, and missing values within 30-minute gaps were linearly interpolated; affecting only 0.09\% of all readings.
The final dataset comprises $n = 968$ individuals.

Recent studies analyzed CGM trajectories via univariate empirical distributions of glucose levels, showing improved performance over traditional scalar summaries \citep{matabuenaGlucodensitiesNewRepresentation2021,coulterFastVariableSelection2025}. However, these representations capture only \emph{global} information (overall glucose level and variability) and discard temporal glucose dynamics, failing to distinguish smooth variations from short-term fluctuations. To capture \textit{local glycemic variability}, we consider sliding 2-hour windows with 75\% overlap. On each valid window ($\ge75\%$ non-missing readings), we compute three features: mean glucose level (Mean), coefficient of variation (CV), and mean absolute 5-minute difference (MAD), which respectively quantify central tendency, relative variability, and short-term fluctuations.
For each individual, the triplets $(\text{Mean}, \text{CV}, \text{MAD})$ from all windows (468 triplets on average) define an empirical trivariate distribution. While \citet{matabuenaGlucodensityFunctionalProfiles2025} also considered multivariate CGM distributions, our approach differs in distributional representation and explicitly accounts for the dependence between marginal components.

We investigate the association of multivariate CGM distributions with two biomarker types: (a) HbA1c, a measure of average glycemia over the preceding 2--3 months (measured before CGM recordings); and (b) lipid profiles---HDL cholesterol (HDL-C),  total cholesterol (Total-C), and triglycerides (TG)---which reflect longer-term dietary and metabolic exposures and are established cardiovascular risk markers \citep{wangAssociationBloodGlucose2022}. TG is log-transformed due to skewness, and we omit LDL cholesterol since it is deterministically derived from HDL-C, Total-C, and TG.
We apply our nonparanormal Fr\'echet regression method to assess:
(i) to what extent HbA1c alone captures the global and local temporal structure of CGM-derived distributions, and
(ii) how lipid profiles are associated with these multivariate glycemic patterns beyond HbA1c.

\begin{figure}%
    \spacingset{1}
    \centering
    \begin{subfigure}[b]{0.98\linewidth}
        \hspace{1mm}
        \includegraphics[width=\linewidth]{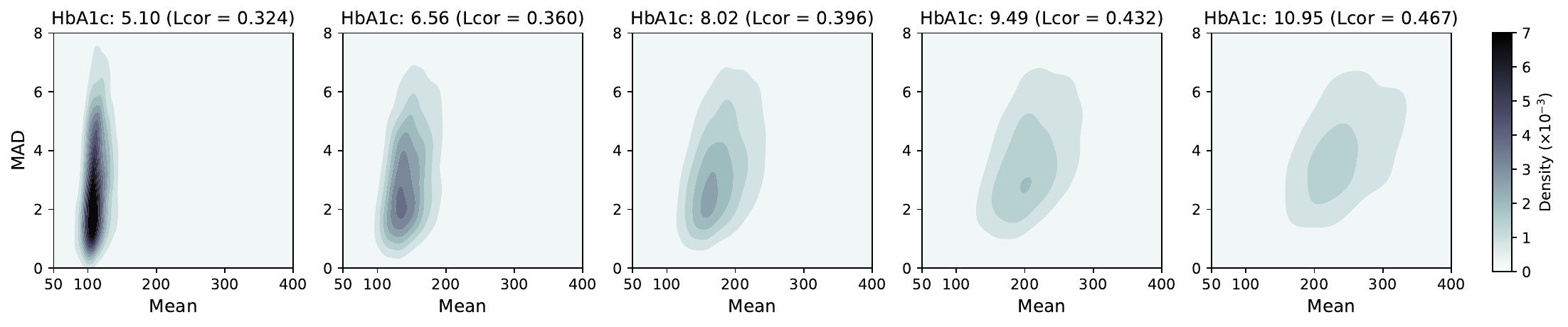}
        \label{fig:ai_readi_a1c_distribution}
    \end{subfigure}
    \vskip -3mm
    \begin{subfigure}[b]{0.7\linewidth}
        \centering
        \includegraphics[width=\linewidth]{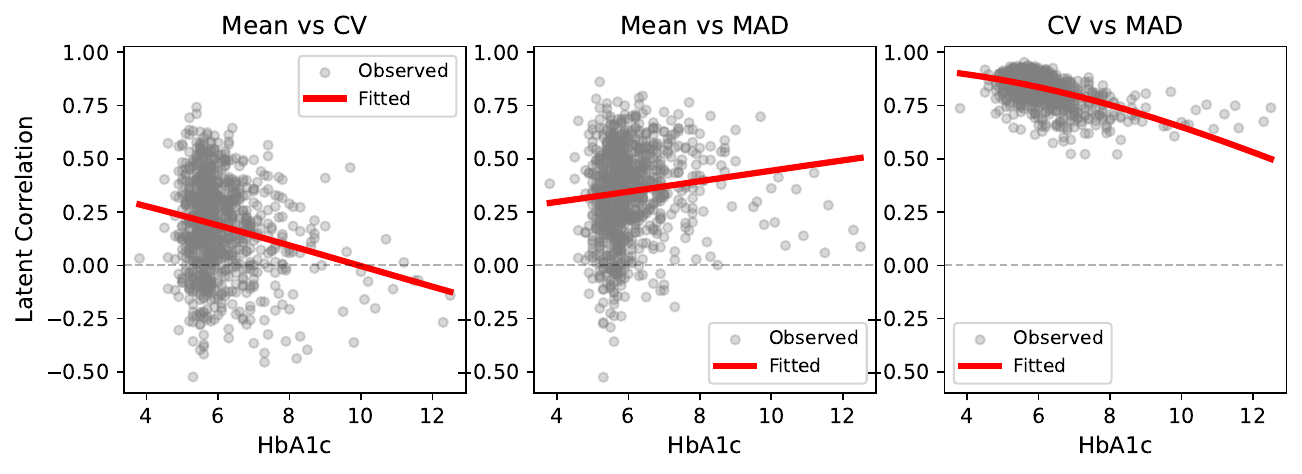}
        \label{fig:ai_readi_a1c_latentcor}
    \end{subfigure}
    \vskip -5mm
    \caption{Nonparanormal Fr\'echet regression fit of trivariate CGM distributions on HbA1c. Top: Heatmap of the (Mean, MAD) component of fitted trivariate distributions at different HbA1c levels within the interval $[5.10, 10.95]$; Lcor in the title of each panel denotes the fitted latent correlation between Mean and MAD. Bottom: Scatterplot of observed latent correlations (points) and the nonparanormal regression fit (red line) as functions of HbA1c.}
    \label{fig:ai_readi}
\end{figure}

We first illustrate the nonparanormal regression fit using only HbA1c as a predictor. Top panels of Figure~\ref{fig:ai_readi} visualize the fitted bivariate components (Mean, MAD) across HbA1c; the fitted distributions are skewed, indicating the benefits of the proposed nonparanormal approach relative to the Gaussian model. As clinically expected, the fitted glucose Mean distribution shifts upward as HbA1c increases. The fitted glucose MAD distribution also shifts upward but also exhibits reduced right-skewness as HbA1c increases. This suggests that, in advanced diabetes, short-term glucose variability tends to increase, whereas the attenuated right tail may reflect a reduced capacity to respond to rapid glucose elevations. The bottom panels of Figure~\ref{fig:ai_readi} illustrate how the regression model captures changes in latent correlations within (Mean, CV, MAD) across HbA1c. The latent correlation between CV and MAD is strong at low HbA1c but attenuates as HbA1c increases, with regression fit capturing this trend. Since high correlation between CV and MAD indicates less oscillatory short-term glucose fluctuations (e.g., largely monotonic changes such as postprandial rises and returns to steady-state level), the attenuation of correlation suggests increasingly heterogeneous glycemic patterns as diabetes progresses. Overall, these results highlight the component-wise interpretability afforded by the proposed nonparanormal Fr\'echet regression, revealing how distinct marginal and dependence components of CGM distributions vary with HbA1c.

\begin{table}[!t]%
\centering
 \spacingset{1}
\caption{Component-wise $R^2$ values from nonparanormal Fr\'echet regression of CGM distributions. ``Lcor'' indicates the latent correlation component. Westfall--Young-adjusted $p$-values are obtained via a permutation test (in parentheses).}
\label{tab:r2_results}
\footnotesize
\begin{tabular}{lllll}
\toprule
Predictor & \multicolumn{1}{c}{Mean} & \multicolumn{1}{c}{CV} & \multicolumn{1}{c}{MAD} & \multicolumn{1}{c}{Lcor} \\
\midrule
HbA1c & 0.557 ($<0.001$) & 0.005 (0.047) & 0.060 ($<0.001$) & 0.072 ($<0.001$) \\
TG & 0.043 ($<0.001$) & 0.014 ($<0.001$) & 0.003 (0.288) & 0.009 (0.004) \\
HDL-C & 0.068 ($<0.001$) & 0.014 (0.002) & 0.010 (0.004) & 0.040 ($<0.001$) \\
Total-C & 0.036 ($<0.001$) & 0.004 (0.138) & 0.005 (0.085) & 0.022 ($<0.001$) \\
All predictors & 0.574 ($<0.001$) & 0.029 ($<0.001$) & 0.078 ($<0.001$) & 0.114 ($<0.001$) \\
\bottomrule
\end{tabular}
\end{table}

Next, we quantitatively evaluate nonparanormal Fr\'echet regression models using the proposed component-wise $R^2$ for each predictor (HbA1c, TG, HDL-C, and Total-C) and for a joint model with all four predictors. Table~\ref{tab:r2_results} shows the resulting $R^2$ values, together with Westfall--Young adjusted $p$-values obtained from 2000 permutations (Section~\ref{sec:compo-wise-r2}). For Mean component, as expected, HbA1c has a strong association ($R^2 = 0.56$), with only little improvement in the joint model ($R^2=0.57$) over HbA1c alone.
However, for glycemic variability components (CV, MAD), joint models show that lipids play a stronger role. In particular, TG and HDL-C show stronger associations with CV than HbA1c.
For the latent correlation matrix among (Mean, CV, MAD), associated $R^2$ values are higher than for the CV and MAD marginals. While HbA1c shows modest associations with CV and MAD individually, it closely tracks their latent correlation trend as seen in the bottom-right panel of Figure~\ref{fig:ai_readi}, indicating that predictors can explain the dependence pattern even when their associations with individual marginals are weak. Overall, our findings that lipid profiles have associations with glycemic variability that are stronger than HbA1c alone align with recent evidence linking global CV to lipid variables in type 1 diabetes \citep{salsa-casteloAssociationGlycemicVariability2024}.

\section{Discussion}\label{sec:discussion}

We develop a new regression approach for multivariate distributional responses based on the decoupling idea within the nonparanormal family, leading to component-wise interpretability, computational efficiency, flexibility beyond Gaussian distributions, theory for empirically sampled responses, and mitigation of the curse of dimensionality.
The Python implementation is available at \texttt{\url{https://github.com/pjywang/NPTFrechet}}.

There are several promising avenues for extending the work.  First, while our current framework relies on permutation-based inference, developing an asymptotic inference with an explicit null distribution---analogous to Wasserstein $F$-tests \citep{petersenWasserstein$F$testsConfidence2021, xuWassersteinFtestsFrechet2025}---is an open problem. Second, our method yields a novel byproduct: a regression approach for correlation matrices under the BW metric, 
which holds potential for brain connectivity applications \citep{chenQuotientGeometryBounded2025}. More broadly, our theory provides a rigorous foundation for using NPT across diverse optimal-transport applications, from distributionally robust optimization \citep{fanConvergenceProjectedBuresWasserstein2024} to generative modeling \citep{chenInferentialWassersteinGenerative2022}. Finally, acknowledging that nonparanormal family may still be restrictive for some applications \citep{shaoFastDistanceComputation2026}, extension to broader copula families is an important direction for future work.

\section*{Acknowledgment}
This research was supported by NIH R01HL172785.

\newpage
\appendix
\bigskip

\renewcommand{\thetable}{S\arabic{table}}  
\renewcommand{\thefigure}{S\arabic{figure}}
\renewcommand{\thelemma}{S.\arabic{lemma}}
\renewcommand{\theproposition}{S.\arabic{proposition}}

\setcounter{figure}{0}
\setcounter{table}{0}
\setcounter{lemma}{0}
\setcounter{proposition}{0}

\begin{center}
\LARGE{
Supplementary material for ``Fr\'echet regression of multivariate distributions with nonparanormal transport''
}
\vskip 1em
\large{Junyoung Park, and Irina Gaynanova}\\[1em]
\normalsize{Department of Biostatistics, University of Michigan, Ann Arbor, Michigan, USA}
\end{center}
\bigskip

\begin{abstract}
    \noindent This appendix provides additional detailed proofs of the results provided in the manuscript. Section~\ref{sec:supp-sec3-proof} contains the proofs of results in Section~\ref{sec:nonpara-npt-metric} of the main paper. Section~\ref{sec:supp-sec4-proof} includes the proofs of results in Section~\ref{sec:frechet-regression} of the main paper.
    Section~\ref{sec:supp-sec5-proof-outline} outlines the proofs of Theorems~\ref{thm:frechet-oracle-rate} and \ref{thm:frechet-real-rate}, which are presented in detail in subsequent Sections~\ref{sec:supp-proof-oracle-rate} and \ref{sec:supp-empirical-rate-proof}. Section~\ref{sec:supp-verify-assump-equicorr} discusses assumption verifications of our asymptotic theory, and Section~\ref{sec:supp-auxiliary-simulation} provides auxiliary simulation results.
\end{abstract}

\section{Proofs of results in Section~\ref{sec:nonpara-npt-metric}}\label{sec:supp-sec3-proof}

This section presents the proofs of the results in Section~\ref{sec:nonpara-npt-metric}. Throughout this section, $\Phi$ denotes the cumulative distribution function of the standard Gaussian distribution.

\subsection{Proof of results in Section~\ref{sec:nonpara-def-extension}}

\paragraph{Proof of Lemma~\ref{lemma:injectivity}.} Let $\Lambda(\bmu, \Sigma)$ and $\Lambda(\bmu', \s')$ be measures in $\Lambda(d)$, where marginal distributions $\bmu = (\mu_j)_{j=1}^d$ and $\bmu' = (\mu_j')_{j=1}^d$ are contained in $\Pcal_*^d$. To show injectivity, assume that $\Lambda(\bmu, \Sigma) = \Lambda(\bmu', \s')$. 

To prove $\bmu = \bmu'$, let $X\in\R^d$ be a random vector following this common distribution. The marginal distributions of $X$ are the same as $\bmu$ and $\bmu'$, so $\mu_j = \mu_j'$ for all $j=1,\dots,d$. Consequently, $\bmu = \bmu'$.

To prove $\Sigma = \s'$, it suffices to show that for any pair $j \neq k$, the correlation $\s_{jk}$ of $\s=(\s_{jk})_{j, k=1}^d$ is uniquely determined by the joint distribution of $(X_j, X_k)$. Using the identity $X_j = q_j(\Phi(Z_j))$ in distribution, where $q_j$ is the quantile function of $\mu_j$ and $Z \sim \mathcal{N}(0, \Sigma)$, we consider the joint cdf of $(X_j, X_k)$:
$$ F_{jk}(x, y) = \P\big(X_j \le x, X_k \le y\big) = \P\Big(q_j(\Phi(Z_j)) \le x,\, q_k(\Phi(Z_k)) \le y\Big). $$
Using the property $q_j(p) \le x \iff p \le F_j(x)$, where $F_j$ is the cdf of $\mu_j$, we have
$$ F_{jk}(x, y) = \P\Big(\Phi(Z_j) \le F_j(x),\, \Phi(Z_k) \le F_k(y)\Big) = \P\Big(Z_j \le \Phi^{-1}(F_j(x)),\, Z_k \le \Phi^{-1}(F_k(y))\Big). $$
Since $\mu_j, \mu_k \in \Pcal_*$, they are not point masses, so there exist $x, y$ such that $F_j(x), F_k(y) \in (0, 1)$, making $a:=\Phi^{-1}(F_j(x))$ and $b:=\Phi^{-1}(F_k(y))$ \textit{finite} real numbers. In this case,
$$F_{jk}(x, y) = \Phi_{\s_{jk}}(a, b),$$
where $\Phi_\rho$ is the bivariate normal cdf with unit variances and correlation $\rho$. The following equality then shows that $F_{jk}(x, y)$ is strictly increasing with the correlation component $\s_{jk}$ in $(-1, 1)$:
$$\frac{\partial}{\partial \rho} \Phi_\rho(a, b) = \phi_\rho(a, b) > 0,$$ 
where $\phi_\rho$ is the corresponding density function of $\Phi_\rho$; see, e.g., Section 2.1 of \citet{genzComputationMultivariateNormal2009}. Furthermore, since $\rho \mapsto \Phi_\rho(a, b)$ is continuous on the closed interval $[-1, 1]$, the strict monotonicity extends to $[-1, 1]$. %
Therefore, the joint cdf value $F_{jk}(x, y)$ uniquely determines the correlation $\Sigma_{jk}$, which completes the proof.
\qed

\subsection{Proof of results in Section~\ref{sec:npt-properties}}\label{sec:supp-proof-npt-properties}

In the following subsections, we prove, in order, the main results of Section~\ref{sec:npt-properties}: Proposition~\ref{prop:geometric-properties}, Proposition~\ref{prop:same-copula}, and Theorem~\ref{thm:topological-equivalence}. Then, in the last subsection~\ref{sec:supp-gaussian-sobolev-verify}, we verify that a large class of distributions satisfies the finiteness of the Gaussian-Sobolev seminorm introduced in Theorem~\ref{thm:topological-equivalence}.

\subsubsection{Proof of Proposition~\ref{prop:geometric-properties}}

Translation or scaling of the random vectors $X, Y$ by $x = (x_1, \ldots, x_d)^\top \in \mathbb{R}^d$ or $a > 0$ does not affect their latent correlation matrices $\s$ and $Q$. Thus,
\begin{align*}
    d_{NPT}^2(X+x, Y + x) &= \sum_{j=1}^d d_W^2(X_j + x_j, Y_j + x_j) + \Bcal^2(\s, Q), \\ 
    d_{NPT}^2(X+x, Y) &= \sum_{j=1}^d d_W^2(X_j + x_j, Y_j) + \Bcal^2(\s, Q), \\
    d_{NPT}^2 (aX, aY) &= \sum_{j=1}^d d_W^2(aX_j ,aY_j) + \Bcal^2(\s, Q).
\end{align*}
The equalities (i), (ii), and (iii) then follow from the translation and scaling properties of the Wasserstein distance; see, e.g., Section~2 of \citet{panaretosStatisticalAspectsWasserstein2019}.
\qed

\subsubsection{Proof of Proposition~\ref{prop:same-copula}}\label{sec:supp-proof-same-copula}

We prove the equality $d_W^2(\mu, \nu)=\sum_{j=1}^d d_W^2(\mu_j, \nu_j)$ by establishing the two-sided inequalities. 

To upper bound $d_W^2(\mu, \nu)$, we construct an explicit coupling between $\mu$ and $\nu$. Let $Z$ be a random vector following the shared latent Gaussian distribution $\Ncal(0_d, \s)$, and consider two random vectors
$$ X = T_{\bgamma}^{\bmu} (Z) =  (T_{\gamma}^{\mu_1}(Z_1), \ldots, T_{\gamma}^{\mu_d}(Z_d))^\top, \quad Y = T_{\bgamma}^{\bnu} (Z) = (T_{\gamma}^{\nu_1}(Z_1), \ldots, T_{\gamma}^{\nu_d}(Z_d))^\top $$
so that $X\sim \mu$ and $Y \sim \nu$. The joint distribution of $(X, Y)$, denoted by $\eta$, is a coupling between $\mu$ and $\nu$, under which we get
\begin{align*}
    \E_\eta[\|X - Y\|^2] &= \sum_{j=1}^d \E_{Z_j}[(T_{\gamma}^{\mu_j}(Z_j) - T_{\gamma}^{\nu_j}(Z_j))^2].
\end{align*}
Recall that $Z_j \sim \mathcal{N}(0, 1)$ and $T_{\gamma}^{\mu_j} = q_{\mu_j} \circ \Phi$, where $q_{\mu_j}$ is the quantile function of $\mu_j$. Letting $U_j = \Phi(Z_j) \sim \Ucal(0, 1)$, each summand on the right-hand side then becomes:
$$ \E_{U_j}[(q_{\mu_j}(U_j) - q_{\nu_j}(U_j))^2] = \int_0^1 (q_{\mu_j}(u) - q_{\nu_j}(u))^2 du = d_W^2(\mu_j, \nu_j), $$
where the last equality is the closed-form expression for the univariate Wasserstein distance, which holds regardless of the continuity of the distributions \citep[Theorem~2.18]{villaniTopicsOptimalTransportation2003}.
Thus,
$$ d_W^2(\mu, \nu) \le \E_\eta[\|X - Y\|^2] = \sum_{j=1}^d d_W^2(\mu_j, \nu_j). $$

Conversely, using the optimal coupling $\pi \in \Gamma(\mu, \nu)$ that attains $d_W^2(\mu, \nu) = \E_\pi[\|X - Y\|^2]$, we obtain a lower bound:
\begin{equation}\label{eq:marginal-inequality}
    d_W^2(\mu, \nu) = \E_\pi[\|X - Y\|^2] = \sum_{j=1}^d \E_\pi[(X_j - Y_j)^2] \ge \sum_{j=1}^d d_W^2(\mu_j, \nu_j),
\end{equation}
concluding the equality
$$ d_W^2(\mu, \nu) = \sum_{j=1}^d d_W^2(\mu_j, \nu_j). $$
Since $\Bcal^2(\Sigma, \Sigma) = 0$, we have $d_{NPT}^2(\mu, \nu) = \sum_{j=1}^d d_W^2(\mu_j, \nu_j)$ as well.
\qed

\subsubsection{Proof of Theorem~\ref{thm:topological-equivalence}}\label{sec:supp-proof-topological-equivalence}

To prove bound~(i), we {establish a technical lemma building on the Hermite polynomial basis of the $L^2$ space under the standard Gaussian measure $\gamma$; here, we let $L^2(\gamma)$ denote the space of measurable maps $\psi:\R\to\R$ such that $\int_\R\psi^2d\gamma <\infty$.} The lemma and its proof are introduced below. Proofs of parts (i) and (ii) follow subsequently.

\begin{lemma}\label{lemma:gaussian-inequality}
    Suppose that $\xi, \zeta \sim N(0,1)$ are jointly Gaussian (possibly dependent). Then, for any function $h:\R\to\R$ with $h\in L^2(\gamma)$, 
    $$ \E[(h(\xi) - h(\zeta))^2] \le |h|_{H^1(\gamma)}^2 \E[(\xi - \zeta)^2]. $$
\end{lemma}

\begin{proof}
Without loss of generality, we assume that $h$ is differentiable and $h'\in L^2(\gamma)$; otherwise, $|h|_{H^1(\gamma)}=\infty$ by definition.

We utilize the probabilist's Hermite polynomials
$$ H_k(x) := (-1)^k e^{x^2/2} \frac{d^k}{dx^k}e^{-x^2/2}, \qquad k=1, 2, \ldots,$$
which are orthogonal in the Hilbert space $L^2(\gamma)$ under the Gaussian measure $d\gamma = \frac{1}{\sqrt{2\pi}}e^{-x^2/2}dx $. In particular, we rely on the following properties:
\begin{enumerate}
    \item $$\langle H_k, H_l\rangle_{L^2(\gamma)}:= \int_\R H_k(x)H_l(x) d\gamma(x) = \begin{cases}
        k! & \text{if }\ k =l,\\
        0 & \text{otherwise}.
    \end{cases}$$
    \item \begin{equation}\label{eq:Hermite-differential}
        H_k'(x) = k H_{k-1}(x).
    \end{equation}
\end{enumerate}

Let $h_k = H_k/\sqrt{k!}$, so that $\{h_k\}_{k=0}^\infty$ forms an orthonormal basis for $L^2(\gamma)$. Consider the basis expansion $h = \sum_{k=0}^\infty c_k h_k$, satisfying  $\|h\|_{L^2(\gamma)}^2 = \E[h(\xi)^2] = \E[h(\zeta)^2] = \sum_{k=0}^\infty c_k^2 <\infty$. By Proposition 2.2.1 of \citet{nourdinNormalApproximationsMalliavin2012}, we have
$$\E[h_k(\xi) h_l(\zeta)] = \begin{cases}
    \rho^k & \text{if }\ k= l,\\
    0 & \text{otherwise,}
    \end{cases}$$
where $\rho = \E[\xi\zeta]$ is the correlation between $\xi$ and $\zeta$. This implies
\begin{align*}
    \E[h(\xi) h(\zeta)] = \E\left[\left(\sum_{k=0}^\infty c_k h_k (\xi)\right)\left(\sum_{l=0}^\infty c_l h_l (\zeta)\right)\right]= \sum_{k=0}^\infty c_k^2 \rho^k,
\end{align*}
leading to
\begin{align*}
    \E[(h(\xi) - h(\zeta))^2] &= 2 \sum_{k=1}^\infty c_k^2(1-\rho^k)\\
    &\le 2\sum_{k=1}^\infty c_k^2 \,k(1-\rho).
\end{align*}
The latter inequality follows from $1-\rho^k = (1-\rho)(1+\cdots +\rho^{k-1}) \le k (1-\rho)$, which holds for $k\ge 1$ since $|\rho| \le 1$.

Furthermore, since $h_k' = \sqrt{k}h_{k-1}$ by \eqref{eq:Hermite-differential}, the $L^2(\gamma)$-norm of the derivative $h' = \sum_{k=1}^\infty c_k \sqrt{k}h_{k-1}$ is given by
$$|h|_{H^1(\gamma)}^2 =  \E[h'(\xi)^2] = \sum_{k=1}^\infty c_k^2 \, k. $$
Combining these results yields
\begin{align*}
    \E[(h(\xi) - h(\zeta))^2] &\le |h|_{H^1(\gamma)}^2 2(1-\rho)  \\
    &= |h|_{H^1(\gamma)}^2 \E[(\xi-\zeta)^2],
\end{align*}
completing the proof. 
\end{proof}

\paragraph{Proof of (i).}
Let $X\sim \mu = \Lambda(\bmu, \s)$ and $Y\sim \nu = \Lambda(\bnu, Q)$ be random vectors, and write $h = (h_1, \ldots, h_d)^\top = T_{\bgamma}^{\bmu} $ and $g = (g_1, \ldots, g_d)^\top = T_{\bgamma}^{\bnu}$. Let $U\sim \Ncal(0_d, \s)$ and $V\sim \Ncal(0_d, Q)$ the latent Gaussian vectors such that $X = h(U)$ and $Y = g(V)$.

Applying the triangle inequality to the Wasserstein distance yields
$$d_W(X, Y) \le d_W(X, g(U)) + d_W(g(U), Y),$$
which implies
\begin{equation*}
    d_W^2(X, Y) \le 2d_W^2(X, g(U)) + 2d_W^2(g(U), Y).
\end{equation*}
Since $X$ and $g(U)$ share the latent Gaussian distribution $N(0, \Sigma_1)$, Proposition~\ref{prop:same-copula} implies that
\begin{align*}
    d_W^2(X, g(U)) = \sum_{j=1}^d d_W^2(X_j, g_j(U_j)) = \sum_{j=1}^d d_W^2(X_j, Y_j),
\end{align*}
which corresponds to the marginal component of $d_{NPT}$. Thus, 
\begin{equation}\label{eq:triangle-bound}
    d_W^2(X, Y) \le 2\sum_{j=1}^d d_W^2(X_j, Y_j) + 2d_W^2(g(U), Y).
\end{equation}

Using the equality $Y = g(V)$, it remains to bound $d_W^2(g(U), g(V))$ in terms of the latent Wasserstein distance $d_W^2(U, V) = \Bcal^2(\s, Q)$. Let $\eta$ be the optimal coupling between the Gaussian laws of $U$ and $V$, satisfying $d_W^2(U, V) = \E_\eta[\|U - V\|^2]$.
The coupling $\eta$ is also multivariate Gaussian \citep[Section~1.6.3]{panaretosInvitationStatisticsWasserstein2020}.
Let $g\times g$ be a function mapping $(u, v)\mapsto(g(u),g(v))$, which yields the coupling $(g\times g)_\#\eta$ between $g(U)$ and $g(V)$. Then, by the definition of the Wasserstein distance \eqref{eq:Wasserstein},
\begin{align*}
    d_W^2(g(U), g(V)) &\le \int_{\R^d\times \R^d}\|x - y\|^2 d((g\times g)_\#\eta)(x, y) \\
    &= \int_{\R^d\times \R^d} \|g(u) - g(v)\|^2 d\eta(u,v) \\ 
    &= \int_{\R^d\times \R^d} \sum_{j=1}^d (g_j(u_j) - g_j(v_j))^2 d\eta(u, v) \\ 
    &= \sum_{j=1}^d \int_{\R\times \R}(g_j(u_j) - g_j(v_j))^2 d\eta_j(u_j, v_j).
\end{align*}
Here, $\eta_j$ denotes the bivariate Gaussian law between univariate components $U_j$ and $V_j$, and each summand equals $\E_{\eta_j}[(g_j(U_j) - g_j(V_j))^2]$. %
Since $Y= g(V)$ has a finite second moment, we have $\int g_j^2 (v) d\gamma(v) = \E[Y_j^2] < \infty$; i.e., $g_j\in L^2(\gamma)$ for all $j$. We can thus apply Lemma~\ref{lemma:gaussian-inequality} to each summand integral, which yields
\begin{align*}
    d_W^2(g(U), g(V)) &\le \sum_{j=1}^d |g_j|_{H^1(\gamma)}^2 \int_{\R\times \R}(u_j -v_j)^2 d\eta_j(u_j, v_j) \\
    & \le |g|^2_{H^1(\gamma)} \sum_{j=1}^d \int_{\R\times \R} (u_j - v_j)^2 d\eta_j(u_j, v_j)\\
    &= |g|^2_{H^1(\gamma)} \int_{\R^d\times\R^d} \|u- v\|^2 d\eta(u, v) \\
    &=|g|^2_{H^1(\gamma)} d_W^2(U, V),
\end{align*}
where the last equality holds since $\eta$ is the optimal coupling between $U$ and $V$. 
Plugging this upper bound into \eqref{eq:triangle-bound} and changing the notations to measures, we have
\begin{align*}
    d_W^2(\mu, \nu) \le 2\sum_{j=1}^d d_W^2(\mu_j, \nu_j) + 2|T_{\bgamma}^{\bnu}|^2_{H^1(\gamma)}\Bcal^2(\s, Q).
\end{align*}
A symmetric argument provides a similar bound
\begin{align*}
    d_W^2(\mu, \nu) \le 2\sum_{j=1}^d d_W^2(\mu_j, \nu_j) + 2|T_{\bgamma}^{\bmu}|^2_{H^1(\gamma)}\Bcal^2(\s, Q).
\end{align*}
Therefore,
\begin{align*}
    d_W^2(\mu, \nu) &\le 2\sum_{j=1}^d d_W^2(\mu_j, \nu_j) + 2 (|T_{\bgamma}^{\bmu}|^2_{H^1(\gamma)} \wedge |T_{\bgamma}^{\bnu}|^2_{H^1(\gamma)}) \Bcal^2(\s, Q)\\
    &\le 2(1 \vee  \{|T_{\bgamma}^{\bmu}|^2_{H^1(\gamma)} \wedge |T_{\bgamma}^{\bnu}|^2_{H^1(\gamma)}) \}\,d_{NPT}^2(X, Y), 
\end{align*}
completing the proof of the first inequality. \qed

\paragraph{Proof of (ii).}

We continue to use the same notations as in the proof of (i), where the maps $h = T_{\bgamma}^{\bmu}$ and $g = T_{\bgamma}^{\bnu}$ are now invertible due to the continuity of marginals. We begin with the inequality
\begin{equation*}
    \sum_{j=1}^d d_W^2(X_j, Y_j) \le d_W^2(X, Y),
\end{equation*}
which is proved in \eqref{eq:marginal-inequality}, implying that
$$d_{NPT}^2(X, Y) \le d_W^2(X, Y) + \Bcal^2(\s, Q).$$
To bound $\Bcal^2(\s, Q) = d_W^2(U, V)$ in terms of $d_W^2(X, Y) = d_W^2(h(U), g(V))$, we invoke the Lipschitz property of the Wasserstein distance:
$$d_W^2(U, V) \le |h^{-1}|_{\text{Lip}}^2 d_W^2(h(U), h(V)) = |h^{-1}|_{\text{Lip}}^2 d_W^2(X, h(V)),$$
where the inequality is a standard fact in optimal transport theory; see, e.g., \citet[Lemma~4.1]{gangboDifferentialFormsWasserstein2011}. The distance in the right-hand side can be bounded by the triangle inequality:
\begin{align*}
    d_W^2(X, h(V)) &\le 2d_W^2(X, Y) + 2d_W^2(Y, h(V)) \\
    &=2 d_W^2(X, Y) + 2\sum_{j=1}^d d_W^2 (X_j, Y_j)\\
    &\le 4 d_W^2(X, Y),
\end{align*}
where the equality uses Proposition~\ref{prop:same-copula}, and the last inequality uses \eqref{eq:marginal-inequality} again. This yields
$$d_W^2(U, V) \le 4|h^{-1}|^2_{\text{Lip}}d_W^2(X, Y),$$
and by symmetry,
$$d_W^2(U, V) \le 4(|h^{-1}|^2_{\text{Lip}} \wedge |g^{-1}|^2_{\text{Lip}}) d_W^2(X, Y).$$
Therefore,
$$d_{NPT}^2(X, Y) \le \{1 + 4(|h^{-1}|^2_{\text{Lip}} \wedge |g^{-1}|^2_{\text{Lip}})\} d_W^2(X, Y),$$
which completes the proof. \qed

\subsubsection{Distributions satisfying the Gaussian-Sobolev seminorm condition}\label{sec:supp-gaussian-sobolev-verify}

We verify that the Gaussian--Sobolev condition $|T_{\bgamma}^{\bmu}|_{H^1(\gamma)} < \infty$ holds for broad classes of nonparanormal distributions $\mu = \Lambda(\bmu, \s)$. Since the seminorm decomposes across dimensions, it suffices to consider univariate distributions $\mu \in \Pcal$ for which we verify the condition $|T_\gamma^\mu|_{H^1(\gamma)} < \infty$. Let $T = T_\gamma^\mu = q_\mu \circ \Phi$ for brevity.
Recall that the condition requires the finiteness of the integral
$$ |T|_{H^1(\gamma)}^2 = \int_{-\infty}^\infty |T'(x)|^2 \phi(x) dx < \infty, $$
where $\phi(x) = (2\pi)^{-1/2}e^{-x^2/2}$ is the standard normal density. In the following examples, $T'$ is continuous, so the integrability is determined by the growth of $T'$ at the tails.

\paragraph{Log-concave distributions.}
Let $\mu$ be a distribution with a log-concave density $p(x) = e^{-V(x)}$. For such distributions, it is known that the tail probability $1 - F(y)$ is also log-concave \citep[Theorem~3]{bagnoliLogconcaveProbabilityIts2005}, making the function $L(y) = -\ln(1 - F(y))$ convex. The transport map $T$ satisfies $1 - \Phi(x) = e^{-L(T(x))}$. Differentiating this identity yields $\phi(x) = L'(T(x)) e^{-L(T(x))} T'(x)$. Substituting $e^{-L(T(x))} = 1 - \Phi(x)$, we obtain
$$ T'(x) = \frac{\phi(x)}{1 - \Phi(x)} \frac{1}{L'(T(x))}. $$
Since $L$ is convex, there exists $\lambda > 0$ such that $L'(y) \ge \lambda$ for sufficiently large $y$. The standard Gaussian distribution has the Mills ratio approximation $\frac{\phi(x)}{1 - \Phi(x)} \simeq x$, where $f(x)\simeq g(x)$ indicates $f(x)/g(x) \to 1 $ as $x\to\infty$. Thus, $T'(x) = O(x)$ as $x \to \infty$. A symmetric argument applies as $x \to -\infty$ and yields $|T'(x)|^2 = O(x^2)$, which verifies $|T|^2_{H^1(\gamma)} < \infty$.

\paragraph{Log-normal distribution.}
Let $\mu$ be the log-normal distribution with parameters $\mu_0 \in \R$ and $\sigma > 0$. The optimal transport map from the standard normal is explicitly $T(x) = \exp(\mu_0 + \sigma x)$. Its derivative satisfies $T'(x) = \sigma T(x)$. The finite Sobolev norm is thus explicitly given by
$$ |T|_{H^1(\gamma)}^2 = \int_{-\infty}^\infty \sigma^2 e^{2(\mu_0 + \sigma x)} \phi(x) dx = \sigma^2 e^{2(\mu_0 + \sigma^2)} < \infty. $$

\paragraph{Student's $t$-distribution.}
Let $\mu$ be the Student's $t$-distribution with degrees of freedom $\alpha > 0$. The density $p$ and the tail probability satisfy
$$ p(y) \simeq C_\alpha y^{-(\alpha+1)} \quad \text{and} \quad 1 - F(y) \simeq \frac{C_\alpha}{\alpha} y^{-\alpha} \quad \text{as } y \to \infty $$
for some constant $C_\alpha > 0$ depending only on $\alpha$. The transport map $T$ satisfies $1 - F(T(x)) = 1 - \Phi(x)$. Using Mills ratio $1 - \Phi(x) \simeq \phi(x)/x$, we deduce the asymptotic relation
$$ T(x)^{-\alpha} \propto \frac{\phi(x)}{x} \implies T(x) \simeq \left( \frac{x}{\phi(x)} \right)^{1/\alpha}. $$
Using the derivative formula $T'(x) = \phi(x)/p(T(x))$, the integrand in the Sobolev norm behaves as
\begin{align*}
    |T'(x)|^2 \phi(x) &= \frac{\phi(x)^3}{p(T(x))^2} \propto \frac{\phi(x)^3}{(T(x)^{-(\alpha+1)})^2} = \phi(x)^3 T(x)^{2(\alpha+1)} \\
    &\propto \phi(x)^3 \left( \frac{x}{\phi(x)} \right)^{\frac{2(\alpha+1)}{\alpha}} = x^{2 + \frac{2}{\alpha}} \phi(x)^{1 - \frac{2}{\alpha}}.
\end{align*}
Substituting $\phi(x) \propto e^{-x^2/2}$, the integrand is proportional to $x^{2 + 2/\alpha} \exp\left( -\frac{1}{2}\left(1 - \frac{2}{\alpha}\right)x^2 \right)$. For this integral to converge, the coefficient in the exponent must be negative, which is equivalent to $\alpha > 2$.

\subsection{Proof of results in Section~\ref{sec:npt-empirical}}\label{sec:supp-proof-npt-rate}

\paragraph{Proof of Theorem~\ref{thm:npt-convergence-rate}.}
Using the definition of NPT,
$$\E[d_{NPT}^2(\hat\mu, \mu)] \le \max_{1\le j\le d} d\cdot \E[d_W^2(\hat\mu_j, \mu_j)] + \E[\Bcal^2(\hat\s, \s)] \le d\, r_N^2 + \E[\Bcal^2(\hat\s, \s)]. $$
The Jensen's inequality $\E[X]\le \E[X^2]^{1/2}$ and the subadditivity of the square-root imply that
\begin{align}\label{eq:prop3-intermediate-bound}
    \begin{split}
        \E[d_{NPT}(\hat\mu,\mu)] &\le \sqrt{d}\cdot r_N + \sqrt{\E[\Bcal^2(\hat\s, \s)]} \\
        &\le \sqrt{d}\cdot r_N + \frac{1}{\sqrt{\lambda_{\min}(\s)}}\sqrt{\E[\|\hat\s-\s\|^2_F]},
    \end{split}
\end{align}
where the second inequality follows from Lemma~\ref{lemma:lipschitz-holder-bw}-2---the central ingredient enabling us to avoid assuming $\lambda_{\min}(\hat\s) > 0$ for all $N$.

We thus bound the Frobenius norm $\|\hat\s-\s\|^2_F$. 
As derived in the proof of Theorem~4.2 of \citet{liuHighdimensionalSemiparametricGaussian2012}, Hoeffding's inequality implies that the $(j, k)$ components satisfy the sub-Gaussian concentration:
\begin{equation}\label{eq:element-wise-concentration}
    \P(|\hat\s_{jk} - \s_{jk}| > t ) \le 2 \exp (-Nt^2/\pi^2).
\end{equation}
This deduces the element-wise bound in expectation
\begin{align*}
    \E[|\hat\s_{jk} - \s_{jk}|^2] = \int_0^\infty \P(|\hat\s_{jk} - \s_{jk}|^2 > t)\,dt  \le  2 \int_0^\infty\exp (-Nt/\pi^2)\,dt = \frac{2\pi^2}{N},
\end{align*}
which implies
\begin{equation}\label{eq:Frob-upperbound-expectation}
    \E[\|\hat\s - \s\|_F^2]\le
    \frac{2\pi^2 \,d(d-1)}{N}
\end{equation}
since $\diag(\tilde\s - \s) = 0_d$.
Plugging this bound into \eqref{eq:prop3-intermediate-bound}, 
$$\E[d_{NPT}(\hat\mu,\mu)] \le \sqrt{d}\cdot r_N + \sqrt{\frac{2\pi^2\, d(d-1)}{N \,\lambda_{\min}(\s)}},
$$
completing the proof. %
\qed

\paragraph{Proof of Corollary~\ref{cor:wasserstein-estimation}.} By Theorems~\ref{thm:topological-equivalence} and \ref{thm:npt-convergence-rate}, one immediately has
$$\E[d_W(\hat\mu, \mu)] \le \sqrt{2}\big(1 \vee |T_{\bgamma}^{\bmu}|_{H^1(\gamma)}\big)\left( \sqrt{d}\cdot r_N + \sqrt{\frac{2\pi^2\, d(d-1)}{N \,\lambda_{\min}(\s)}}\right), $$
which satisfies $O(r_N)$ as $N\to\infty$ with $d$ fixed. %
\qed

\section{Proofs and computational discussions for Section~\ref{sec:frechet-regression}}\label{sec:supp-sec4-proof}

\subsection{Proof of Lemma~\ref{lemma:BW-projection}}

The squared BW distance $\Bcal^2(Q, \s)$ satisfies
$$\Bcal^2(Q, \s) = \inf_{M^\top M = Q, L^\top L = \s} \|M - L\|_F^2,$$
where the infimum is taken over $M, L\in\R^{d\times d}$ \citep[Theorem~1]{bhatiaBuresWassersteinDistance2019}. Fix $L = \s^{1/2}$ so that $L^\top L = \s$. Since the Frobenius norm is invariant under multiplication by orthogonal matrices, 
$$\Bcal^2(Q, \s) = \inf_{M^\top M = Q} \|M - L\|_F^2. $$
The minimization in the projection map $\Pcal_{\Ecal_d}(\s)=\argmin_{Q \in \Ecal_d} \Bcal^2(Q, \s)$ is then equivalent to
\begin{equation}\label{eq:parametrized-optimization}
    \min_{Q \in \Ecal_d} \Bcal^2(Q, \s) =\min_{M: \diag(M^\top M) = 1_d} \|M - L\|_F^2,
\end{equation}
since the set $\{M\in\R^{d\times d}: \diag(M^\top M) = 1_d\}$ parametrizes all matrices of $\Ecal_d$ via the mapping $M\mapsto M^\top M$. Thus, a minimizer $\hat M\in\R^{d\times d}$ of the right-hand side of \eqref{eq:parametrized-optimization} yields the projection $\hat M^\top \hat M = \Pcal_{\Ecal_d}(\s)$. 
To solve this, note that the constraint $\diag(M^\top M) = 1_d$ is equivalent to $\|m_j\| = 1$ for every $j$th column $m_j$ of $M$. Thus, the last optimization decomposes into column-wise problems:
$$ \argmin_{\|m_j\|=1} \|m_j - \ell_j\|^2, \quad j=1, \ldots, d, $$
where $\ell_j$ is the $j$-th column of $L$. For each $j$, the solution is given by $\hat m_j = \ell_j / \|\ell_j\|$. Since $\|\ell_j\|^2 = (L^\top L)_{jj} = \s_{jj}$, we have $\hat M = [\hat m_1,\ldots, \hat m_d] = L\,D(\s)^{-1/2}$, implying that
$$ P_{\Ecal_d}(\s) = \hat M^\top \hat M = D(\s)^{-1/2} L^\top L \,D(\s)^{-1/2} = D(\s)^{-1/2} \s\, D(\s)^{-1/2}. $$
\qed

\subsection{Exactness of Algorithm~\ref{alg:proj-Riem-GD} in the bivariate and equicorrelation case}\label{sec:supp-proof-one-step-eqcorr}

In Section~\ref{sec:supp-one-step-bivariate}, we prove that one iteration of Algorithm~\ref{alg:proj-Riem-GD} attains the exact correlation regression solution $\hat S_F(z)$ in the bivariate case. In Section~\ref{sec:supp-equicorrelation-extension}, we then extend the argument to the equicorrelation case.

\subsubsection{Bivariate case and its geometric illustration}\label{sec:supp-one-step-bivariate}

\begin{proposition}\label{prop:one-step-bivar}
    Fix $z\in\R^p$, and suppose the $\hat S_i$ are bivariate correlation matrices with parameters $\hat\rho_i\in(-1,1)$, $i=1,\ldots,n$. Let $A_n(z) = \frac{1}{n}\sum_{i=1}^n s_n(Z_i, z) \sqrt{1+\hat\rho_i}$ and $D_n(z) = \frac{1}{n}\sum_{i=1}^ns_n(Z_i, z) \sqrt{1-\hat\rho_i}$. If $A_n(z), D_n(z) > 0$, then the minimizer $\hat S_F(z)$ is unique and positive definite and is attained by a single iteration of Algorithm~\ref{alg:proj-Riem-GD} with $\eta=1$. 
\end{proposition}
The condition $A_n(z), D_n(z) > 0$ is readily verifiable from data and is violated only in rare, extreme cases. If either quantity is nonpositive, then the fitted correlation parameter of $\hat S_F(z)$ lies on the boundary $\pm 1$, or the objective function in \eqref{eq:correlation-estimation} becomes constant.

\paragraph{Proof of Proposition~\ref{prop:one-step-bivar}.
}
Fix $z$ and write $s_i := s_n(Z_i, z)$, so that $n^{-1}\sum_{i=1}^n s_i = 1$. Write 
$$\hat S_i=\s(\hat\rho_i),\quad \text{where}\quad \s(\rho):=\begin{pmatrix}1 & \rho\\ \rho & 1\end{pmatrix}.$$
Define
$$\hat F_n(\rho):= \frac{1}{n}\sum_{i=1}^n s_i\,\Bcal^2\!\left(\s(\hat\rho_i), \s(\rho)\right),\qquad \rho\in[-1,1],$$
which corresponds to the objective function in \eqref{eq:correlation-estimation}. For $r,\rho\in[-1,1]$, the BW metric on the bivariate family satisfies
\begin{equation*}
\Bcal^2\!\left(\s(r), \s(\rho)\right)
= 2d - 2\big(\sqrt{1+r}\sqrt{1+\rho} + \sqrt{1-r}\sqrt{1-\rho}\big).
\end{equation*}
Thus, the objective function is written as
\begin{equation}\label{eq:bivariate-objective}
    \hat F_n(\rho) = 2d- 2 A_n(z)\sqrt{1+\rho} -2 D_n(z) \sqrt{1-\rho}. 
\end{equation}
By the assumption of $A_n(z), D_n(z) >0$, we can compute the first and second derivatives of $\hat F_n$ with respect to $\rho$ in the interior $(-1, 1)$:
\begin{align*}
    \hat F_n'(\rho) &= -A_n(z)(1+\rho)^{-1/2} + D_n(z)(1-\rho)^{-1/2},\\
    \hat F_n''(\rho) &= \frac{1}{2}A_n(z)(1+\rho)^{-3/2} + \frac{1}{2}D_n(z)(1-\rho)^{-3/2}.
\end{align*}
The second derivative is strictly positive, so $\hat F_n$ is convex and has a unique minimizer within $(-1, 1)$. Solving the first-order optimality condition $\hat F_n'(\rho) = 0$, we obtain the unique solution
$$\hat \rho_F(z) := \frac{A_n(z)^2 - D_n(z)^2}{A_n(z)^2+ D_n(z)^2},$$
with $\s(\hat\rho_F(z)) = S_F(z)$. 

Now consider one iteration of Algorithm~\ref{alg:proj-Riem-GD} with $\eta=1$ and $S^{(0)}=\hat S_1$. Since bivariate correlation matrices commute, for each $i$,
$$T_{S^{(0)}}^{\hat S_i}=\hat S_i^{1/2}(S^{(0)})^{-1/2}.$$
Then in line 3, the Riemannian gradient computation with $\eta=1$ becomes
$$G^{(1)} = \frac{1}{n}\sum_{i=1}^n s_i T_{S^{(0)}}^{\hat S_i} = \bigg(\frac{1}{n}\sum_{i=1}^n s_i \hat S_i^{1/2}\bigg) (S^{(0)})^{-1/2}. $$
As the matrices share the same eigenspaces, the computations can be done as follows. Let $P_{1}:=2^{-1}1_21_2^\top$ and $P_\perp:=I_2-P_1$ denote the projections onto the respective eigenspaces with eigenvalues $1 + \rho$ and $1-\rho$ of the bivariate correlation matrix $\s(\rho)$. Then for each $i$, 
\begin{equation*}
\hat S_i^{1/2} = \sqrt{1+\hat\rho_i}\,P_{1} + \sqrt{1-\hat\rho_i}\,P_\perp,
\end{equation*}
hence letting
\begin{equation*}
M_n:=\frac{1}{n}\sum_{i=1}^n s_i\,\hat S_i^{1/2} = A_n(z)P_{1} + D_n(z)P_\perp
\end{equation*}
gives the Riemannian gradient $G^{(1)}=M_n(S^{(0)})^{-1/2}$. Line 4 then retracts $G^{(1)}$ onto $\Scal_d^{++}$ as
\begin{equation*}
\Sigma^{(1)} = M_n^2 = A_n(z)^2P_{1} + D_n(z)^2P_\perp.
\end{equation*}
The diagonal entries of $\s^{(1)}$ are constant:
\begin{equation*}
\diag(\Sigma^{(1)}) = c(z)\,1_2,\qquad c(z):=\frac{A_n(z)^2 + D_n(z)^2}{2},
\end{equation*}
and the off-diagonal entry of $\s^{(1)}$ is $(A_n(z)^2 - D_n(z)^2)/2$.
Thus by Lemma~\ref{lemma:BW-projection},
\begin{equation*}
S^{(1)} = P_{\Ecal_d}(\Sigma^{(1)}) = c(z)^{-1}\Sigma^{(1)} = \s(\rho^{(1)}),
\end{equation*}
with
\begin{equation*}
\rho^{(1)} = \frac{A_n(z)^2 - D_n(z)^2}{A_n(z)^2 +D_n(z)^2}.
\end{equation*}
Comparing the two formulas for $\hat\rho_F(z)$ and $\rho^{(1)}$, we get $\rho^{(1)}=\hat\rho_F(z)$ and thus $S^{(1)}=\hat S_F(z)$.
\qed

\paragraph{Geometric illustration.}
We provide a further illustration of the $A_n(z), D_n(z) > 0$ condition in Proposition~\ref{prop:one-step-bivar} with the BW metric geometry in the $d=2$ case. Parametrizing $\Ecal_2$ endowed with the BW metric
$$\Bcal^2(\rho_1, \rho_2) = 4 - 2\sqrt{(1+\rho_1)(1+\rho_2)} - 2\sqrt{(1-\rho_1)(1-\rho_2)}$$
on $[-1, 1]$, we can isometrically embed $\Ecal_2$ into the first quadrant of the circle of radius $\sqrt{2}$, $\sqrt{2}\,\mathbb{S}_+^1\subset\R^2$, with the Euclidean chordal metric. Specifically, the isometry is given by
$$\varphi:\Ecal_2\to\sqrt{2}\,\mathbb{S}_+^1,\qquad \rho\mapsto\big(\sqrt{1+\rho},\,\sqrt{1-\rho}\big)^\top,$$
whose inverse is $\varphi^{-1}(x,y)=(x^2-y^2)/2$.

Using this isometry, consider embeddings of correlations $\varphi(\hat\rho_i) =: (x_i, y_i)\in\sqrt{2}\,\mathbb{S}_+^1 $. Then the quantities $A_n(z)$ and $D_n(z)$ becomes the affine combination in the ambient $\R^2$ space:
$$(A_n(z), D_n(z)) = \sum_{i=1}^n s_i\, (x_i, y_i) \in\R^2. $$
Therefore, the requirement of $A_n(z), D_n(z) > 0$ is equivalent to require the affine combination $\sum_{i=1}^n s_i \,\varphi(\hat\rho_i)$ to lie in the first quadrant of $\R^2$. Since the embedded points $\varphi(\hat\rho_i)$ already lie in the first quadrant, violating the positivity of $(A_n(z), D_n(z))$ requires an extreme case that makes large negative linear weights $s_i$.

Specifically, if $A_n(z) = D_n(z) = 0$, then the objective function $\hat F_n(\rho)$ in \eqref{eq:bivariate-objective} becomes constant, a meaningless case in the sense of regression. If only one of them is zero, the objective $\hat F_n$ becomes monotone in $\rho$, meaning that the minimizer $\hat\rho_F(z)$ lies in the boundary $\pm 1$. If all of them are nonzero but one is negative, it again yields monotone $\hat F_n(\rho)$, forcing the solution to lie on the boundary. Taken together, violation of the condition $A_n(z), D_n(z) > 0$ corresponds to fitting an extreme correlation matrix with $\hat\rho_F(z) \in \{\pm 1\}$ (i.e., perfect correlation), which is of less interest when the input data correlations $\hat\rho_i$ do not lie on the boundary.
\qed

\subsubsection{Equicorrelation extension}\label{sec:supp-equicorrelation-extension}

The same argument of Proposition~\ref{prop:one-step-bivar} applies when $\hat S_i$ are equicorrelation matrices, but under the additional \textit{uniqueness} assumption on $\hat S_F(z)$ (this assumption is mild in most practical scenarios). Under this assumption, for the objective function $\hat F_n(z, \s) = \frac{1}{n}\sum_{i=1}^n s_i\Bcal^2(\hat S_i, \s)$, the unique $\hat S_F(z)$ must be permutation-invariant; i.e.,
$$P S_F(z) P^\top = S_F(z)$$
for any permutation matrices $P$. This is because every $\hat S_i$ satisfies $P\hat S_i P^\top = \hat S_i$ and the BW metric satisfies $\Bcal(PSP^\top, P\s P^\top) = \Bcal(S, \s)$ \citep{bhatiaBuresWassersteinDistance2019}. Thus, the minimizer $\hat S_F(z)$ is enforced to lie within the equicorrelation family, in which we can perform a similar computational procedure to the bivariate case, concluding that a single iteration of Algorithm~\ref{alg:proj-Riem-GD} also yields the unique equicorrelation minimizer $\hat S_F(z)$.
\qed

\section{Outline and background of the proofs of results in Section~\ref{sec:theory}}\label{sec:supp-sec5-proof-outline}

To prove the asymptotic results presented in Section~\ref{sec:theory}, this section provides an outline and the necessary background. From here on, for symmetric matrices $M$ and $\lambda>0$, we write $M \succeq \lambda I_d$ when $\lambda_{\min}(M)\ge \lambda$ and $M \succ 0$ when $\lambda_{\min}(M)>0$. We also consider an interior of $\Ecal_d^{++}$ where eigenvalues are lower-bounded; specifically, for any positive number $\lambda\in(0,1)$, we define
\begin{equation}\label{eq:corr-eigen-lowerbound}
    \Ecal_d(\lambda) = \{ \s\in\Ecal_d: \lambda_{\min}(\s) \ge \lambda\},
\end{equation}
which is a compact, convex subset of $\Ecal_d$.

\subsection*{Outline of the proof}

The steps involved in the proofs are organized as follows:
\begin{itemize}
    \item Section~\ref{sec:supp-prelim-matcal} introduces the necessary background on multilinear algebra and differential calculus, with their application to the set of positive definite correlation matrices $\Ecal_d^{++}$. In Section~\ref{sec:supp-diff-bw-prop}, we introduce some technical results, computational/differential properties of the squared BW distance on $\Ecal_d^{++}$ that will be used throughout in the proof.
    \item Section~\ref{sec:supp-proof-oracle-rate} presents the proof of Theorem~\ref{thm:frechet-oracle-rate}, which establishes the consistency of the nonparanormal Fr\'echet regression estimator in the oracle case. 
    \item Section~\ref{sec:supp-empirical-rate-proof} provides the proof of Theorem~\ref{thm:frechet-real-rate}, extending the oracle result to the empirical-response setting.
\end{itemize}

\noindent We also verify the technical assumptions in some specific cases in Section~\ref{sec:supp-verify-assump-equicorr}.

\subsection{Differential calculus on matrix manifolds}\label{sec:supp-prelim-matcal}

We briefly review the differential calculus on vector spaces and its specialization to the set of correlation matrices. Further details can be found in Chapter 2 of \citet{abrahamManifoldsTensorAnalysis2012}.

\subsubsection{Symmetric multilinear operators}\label{sec:supp-multilinear}

Let $\mathbb{V}$ and $\mathbb{W}$ be finite-dimensional normed vector spaces, with norms denoted by $\|\cdot\|$. Define $\Lcal(\mathbb{V}, \mathbb{W})$ as the space of linear operators from $\mathbb{V}$ to $\mathbb{W}$, equipped with the operator norm $\|L\|_{op} = \sup_{\|v\| \le 1} \|L[v]\|$. We define the spaces of multilinear operators recursively: $\Lcal^1(\mathbb{V}, \mathbb{W}) = \Lcal(\mathbb{V}, \mathbb{W})$ and $\Lcal^k(\mathbb{V}, \mathbb{W}) = \Lcal(\mathbb{V}, \Lcal^{k-1}(\mathbb{V}, \mathbb{W}))$ for $k \ge 2$. Elements of $\Lcal^k(\mathbb{V}, \mathbb{W})$ are identified with $k$-multilinear maps $L: \mathbb{V} \times \dots \times \mathbb{V} \to \mathbb{W}$. When $\mathbb{W}=\R$, we denote the space of $k$-multilinear forms by $\Lcal^k(\mathbb{V})$.

An operator $L \in \Lcal^k(\mathbb{V}, \mathbb{W})$ is called \textit{symmetric} if $L[v_1, \dots, v_k] = L[v_{\sigma(1)}, \dots, v_{\sigma(k)}]$ for any permutation $\sigma:\{1, \dots, k\}\to\{1, \dots, k\}$, $\forall v_1,\ldots,v_k\in\mathbb{V}$. 
For a symmetric $k$-linear form $L \in \Lcal^k(\mathbb{V}, \mathbb{W})$, we distinguish between the operator norm $$\|L\|_{op} = \sup_{\|v_1\|\le 1, \dots, \|v_k\|\le 1} \|L[v_1, \dots, v_k]\|$$ and the polar norm $$\|L\|_{P} = \sup_{\|v\|\le 1} \|L[v, \dots, v]\|.$$ These norms are topologically equivalent \citep[Section~2.2B]{abrahamManifoldsTensorAnalysis2012}, satisfying 
\begin{equation}\label{eq:polar-norm-equivalence}
    \|L\|_{P} \le \|L\|_{op} \le \frac{k^k}{k!} \|L\|_{P}.
\end{equation}
For a matrix $M$ viewed as a linear operator, $\|M\|_{op}$ is defined using Euclidean norms; i.e., $\|M\|_{op}$ equals the spectral norm.

\subsubsection{Differential calculus on normed vector spaces}\label{sec:supp-diff-cal}

Let $U \subset \mathbb{V}$ be an open subset. A map $F: U \to \mathbb{W}$ is said to be Fr\'echet differentiable at $x \in U$ if there exists a linear operator $d_x F \in \Lcal(\mathbb{V}, \mathbb{W})$ such that 
$$\lim_{\|h\| \to 0} \|F(x+h) - F(x) - d_x F[h]\| / \|h\| = 0.$$
Higher-order differentials are defined recursively using the map $x \mapsto d_x F$ from $U$ to $\Lcal(\mathbb{V}, \mathbb{W})$. If this map is differentiable at $x$, its derivative is denoted by $d^2_x F \in \Lcal(\mathbb{V}, \Lcal(\mathbb{V}, \mathbb{W})) \cong \Lcal^2(\mathbb{V}, \mathbb{W})$.
In general, the $k$-th order differential $d^k_x F$ is defined as differential of $x \mapsto d^{k-1}_x F$, identified as a symmetric $k$-multilinear form in $\Lcal^k(\mathbb{V}, \mathbb{W})$. If the target space is real-valued $\mathbb{W} = \R$, we also write $\nabla_xF = d_xF$, $\nabla_x^2F = d_x^2F$, $\nabla_x^3 F= d_x^3F,$ ... to indicate the usual gradient, Hessian, and higher-order differentials.

We say $F$ is of class $C^k$ on $U$ if the $k$-th differential $d^k_x F$ exists for all $x \in U$ and the map $x \mapsto d^k_x F$ is continuous from $U$ to $\Lcal^k(\mathbb{V}, \mathbb{W})$. The exponent $k$ can be infinite, meaning that $F$ is infinitely differentiable.
It is well-known that the composition of $C^k$ maps is again $C^k$; see, e.g., Theorem~2.4.3 of \citet{abrahamManifoldsTensorAnalysis2012}.

The following mean value inequality of vector-valued smooth functions will be useful to derive Lipschitzness and boundedness of such operators over compact sets:
\begin{lemma}[Mean Value Inequality, Proposition 2.4.8 of \citealt{abrahamManifoldsTensorAnalysis2012}]\label{lemma:mean-value-inequality}
    Let $U \subset \mathbb{V}$ be an open set and $F: U \to \mathbb{W}$ be a $C^1$ map. Let $x, y \in U$ be such that the line segment connecting them lies in $U$. Then,
    $$ \|F(y) - F(x)\| \le \sup_{t \in [0, 1]} \|d_{x+t(y-x)} F\|_{op} \cdot \|y - x\|. $$
\end{lemma}

\subsubsection{Differential calculus on $\Ecal_d$}\label{sec:supp-diff-calculus-corr}

Recall that $\Scal_d^{++}$ is an open subset of $\Scal_d$. The set of positive definite correlation matrices $\Ecal_d^{++}$ is the intersection of $\Scal_d^{++}$ and the affine subspace $\{M\in\Scal_d: \diag(M) = 1_d\}$; i.e., $\Ecal_d^{++}$ is a submanifold of $\Scal_d^{++}$. Thus, for any differentiable function $F:\Ecal_d^{++}\to\mathbb{W}$, the domain of its differentials $d_Q F$ at $Q \in \Ecal_d^{++}$ is identified with the tangent space 
$$\Tcal_d = \{M \in \Scal_d : \diag(M) = 0_d\},$$
associated with the affine subspace $\{M\in\Scal_d: \diag(M) = 1_d\}$, and we write $\Tcal = \Tcal_d$ for brevity. For higher-order differentials $d_Q^k F \in \Lcal^k(\Tcal, \mathbb{W})$, we use the notation
$$\|d_Q^k F\|_\Tcal  := \|d_Q^kF\|_{P}$$
to indicate the polar norm on $\Lcal^k(\Tcal, \mathbb{W})$.

When there is a map $G: \Scal_d^{++} \to \mathbb{W}$ that extends $F$, i.e., $G|_{\Ecal_d^{++}} = F$, the differential calculus of $F$ can be inherited from that of $G$. Specifically, the differential of $F$ equals the differential of $G$ restricted to the tangent space $\Tcal$:
$$ d_Q F[M] = d_Q G[M], \quad \forall Q \in \Ecal_d^{++}, \ M \in \Tcal. $$
This relationship extends recursively to higher-order differentials. 
In particular, the polar norm of $F = G|_{\Ecal_d^{++}}$ can be represented as
$$\big\|d_Q^kF\big\|_\Tcal = \sup_{M\in\Tcal,\,\|M\|_F\le 1} \|d_Q^kF[M, \ldots, M]\|, $$
which is thus less than or equal to the ambient polar norm $\big\|d_Q^kG\big\|_P$.

\subsection{Technical lemmas: computational/differential properties of $\Bcal^2$ on $\Ecal_d$}\label{sec:supp-diff-bw-prop}

We collect basic computational and differential properties of the squared BW distance $\Bcal^2(Q, \s)$ for our asymptotic analysis on the correlation domain $\Ecal_d$. Note that eigenvalues are at most $d$ for all correlation matrices.

Viewing $f_Q(\s) := \Bcal^2(Q, \s)$ as a function on $\s\in \Scal_d^{++}$, the gradient of $f_Q$ is given by 
$$\nabla f_Q(\s) = I_d - T_\s^Q,$$
where $T_\s^Q = \s^{-1/2}(\s^{1/2}Q \s^{1/2})^{1/2}\s^{-1/2}$ is the optimal transport map from $\Ncal(0, \s)$ to $\Ncal(0, Q)$ \citep{kroshninStatisticalInferenceBures2021}. The higher-order derivatives along a direction $M \in \Scal_d$ are then represented as
$$ \nabla^2 f_Q(\s)[M, M] := - \langle dT_\s^Q[M], M \rangle_F, \quad \nabla^3 f_Q(\s)[M, M, M] := - \langle d^2T_\s^Q[M, M], M \rangle_F, $$
where $dT_\s^Q$ and $d^2T_\s^Q$ are differentials of the map $\s \mapsto T_\s^Q$ ($d_\s$ is written as $d$ for brevity). For our analysis with correlation matrices, we will only consider tangential directions $M\in\Tcal$.

We first present the smoothness of the transport map:

\begin{lemma}\label{lemma:transport-smoothness}
    The map $(\s, Q)\mapsto T_\s^Q$ is $C^\infty$ on $\Ecal_d^{++}\times \Ecal_d^{++}$.
\end{lemma}

\begin{proof}
    Recall $T_\s^Q = \s^{-1/2}(\s^{1/2}Q \s^{1/2})^{1/2}\s^{-1/2}$.
    The matrix square root map $A \mapsto A^{1/2}$ is $C^\infty$ on $\Scal_d^{++}$ \citep{delmoralTaylorExpansionSquare2018}.
    Matrix multiplication is bilinear, hence $C^\infty$, and matrix inversion is $C^\infty$ on $\Scal_d^{++}$. 
    The map $(\s, Q)\mapsto T_\s^Q$ is thus a composition of $C^\infty$ maps and is therefore $C^\infty$.
\end{proof}

Next, we present uniform boundedness and Lipschitz properties of $T_\s^Q$ and its differentials within the compact interior of $\Ecal_d^{++}$. In what follows, we view $dT_\s^Q$ and $d^2T_\s^Q$ as multilinear operators on $\Tcal$, endowed with the polar norm $\|\cdot\|_\Tcal$.

\begin{lemma}\label{lemma:derivative-bounds}
    Suppose $Q, \s \in \Ecal_d(\lambda)$ for some $\lambda\in(0, 1)$. Then, there exists a constant $C > 0$ depending only on $\lambda$ and $d$ such that
    \begin{align*}
        \|T_\s^Q\|_F &\le C, \\
        \|dT_\s^Q[M]\|_F &\le C \|M\|_F, \\
        \|d^2T_\s^Q[M, M]\|_F &\le C \|M\|_F^2,
    \end{align*}
    for all $M \in \Tcal$. Consequently, the higher-order derivatives of $f_Q(\s) = \Bcal^2(Q, \s)$ are uniformly bounded on the tangent space $\Tcal$:
    $$ \|\nabla f_{Q}(\s)\|_F \le C,\quad \|\nabla^2 f_Q(\s)\|_{\Tcal} \le C, \quad \|\nabla^3 f_Q(\s)\|_{\Tcal} \le C. $$
\end{lemma}

\begin{proof}

    Since $\Ecal_d(\lambda)\times \Ecal_d(\lambda)$ is compact and $(\s, Q)\mapsto T_\s^Q$ is $C^\infty$ (Lemma~\ref{lemma:transport-smoothness}), we have $$\sup_{Q,\s\in\Ecal_d(\lambda)}\|T_\s^Q\|_F < \infty.$$
    Moreover, the partial differentials $(\s, Q)\mapsto dT_\s^Q$ and $(\s, Q)\mapsto d^2T_\s^Q$ (with respect to $\s$) are continuous as maps into $\Lcal(\Tcal,\Scal_d)$ and $\Lcal^2(\Tcal,\Scal_d)$, respectively. Therefore, their operator norms are uniformly bounded on $\Ecal_d(\lambda)\times \Ecal_d(\lambda)$:
    $$\sup_{Q,\s\in\Ecal_d(\lambda)} \sup_{\|M\|_F\le 1}\|dT_\s^Q[M]\|_F < \infty, \qquad        \sup_{Q,\s\in\Ecal_d(\lambda)} \sup_{\|M\|_F\le 1}\|d^2T_\s^Q[M, M]\|_F < \infty.$$
    This implies all the stated bounds, including those defined with the function $f_Q$.
\end{proof}

\begin{lemma}\label{lemma:lipschitz-transport}
    Fix $\lambda \in (0, 1)$. There exists a constant $C > 0$ depending only on $d$ and $\lambda$ such that:
    \begin{enumerate}
        \item for any $\s \in \Ecal_d(\lambda)$ and any $Q_1, Q_2 \in \Ecal_d(\lambda)$,
        \begin{align*}
            \|T_{\s}^{Q_1} - T_{\s}^{Q_2}\|_F &\le C \|Q_1 - Q_2\|_F,\\
            \|dT_{\s}^{Q_1} - dT_{\s}^{Q_2}\|_{\Tcal} &\le C \|Q_1 - Q_2\|_F;
        \end{align*}
        \item for any $Q\in \Ecal_d(\lambda)$ and any $\s_1, \s_2 \in \Ecal_d(\lambda)$,
        \begin{align*}
            \|T_{\s_1}^Q - T_{\s_2}^Q\|_F &\le C \|\s_1 - \s_2\|_F,\\
            \|dT_{\s_1}^Q - dT_{\s_2}^Q\|_{\Tcal} &\le C \|\s_1 - \s_2\|_F,\\
            \|d^2T_{\s_1}^Q - d^2T_{\s_2}^Q\|_{\Tcal} &\le C \|\s_1 - \s_2\|_F.
        \end{align*}
    \end{enumerate}
\end{lemma}

\begin{proof}
    By Lemma~\ref{lemma:transport-smoothness}, the maps $(\s, Q)\mapsto T_\s^Q$, $(\s, Q)\mapsto dT_\s^Q$, and $(\s, Q)\mapsto d^2T_\s^Q$ are $C^\infty$, and hence their partial differentials (with respect to $\s$ or $Q$) are continuous.
    Since $\Ecal_d(\lambda)\times \Ecal_d(\lambda)$ is compact, the corresponding operator norms of these partial differentials are uniformly bounded.
    Because $\Ecal_d(\lambda)$ is convex, the line segment joining any two points stays in $\Ecal_d(\lambda)$. The stated Lipschitz bounds then follow from the mean value inequality (Lemma~\ref{lemma:mean-value-inequality}).
\end{proof}

Finally, the following lemma tightly bounds the BW distance by the Frobenius norm. The first result establishes the H\"older continuity of $f_Q(\s) := \Bcal^2(Q, \s)$ with respect to the Frobenius norm and does not depend on the minimum eigenvalues of correlation matrices. The second inequality gives an upper bound of $\Bcal$ when \textit{only one} of the minimum eigenvalues is controlled. 

\begin{lemma}\label{lemma:lipschitz-holder-bw}
    Let $Q, \s, \s' \in\Ecal_d$ be any correlation matrices.
    \begin{enumerate}
        \item $|\Bcal^2(Q, \s) - \Bcal^2(Q, \s')| \le 2d^{3/4}\|Q\|_{op}^{1/2}\|\s - \s'\|_F^{1/2}. $
        \item Suppose that $Q\succeq \lambda I_d$ for some $\lambda > 0$. Then,
        $$\Bcal(Q , \s)  \le \frac{1}{\sqrt{\lambda}}\|Q - \s\|_F.$$
        
    \end{enumerate}
\end{lemma}

\begin{proof}
    \textit{1.} Note that $\Tr[Q] = \Tr[\s] = \Tr[\s'] = d$ since they are correlation matrices. By the closed form expression of $\Bcal^2$,
    \begin{align*}
        |\Bcal^2(Q, \s) - \Bcal^2(Q, \s')| &= 2 \big|\!\Tr\big[(Q^{1/2}\s' Q^{1/2})^{1/2} - (Q^{1/2}\s Q^{1/2})^{1/2}\big]\big| \\
        & \le 2\sqrt{d} \big\|(Q^{1/2}\s' Q^{1/2})^{1/2} - (Q^{1/2}\s Q^{1/2})^{1/2}\big\|_F,
    \end{align*}
    where the inequality follows from the Cauchy--Schwarz inequality $\Tr(MI_d)\le \sqrt{d}\|M\|_F$. Let $A = Q^{1/2}\s Q^{1/2}$ and $B = Q^{1/2}\s' Q^{1/2}$, which are positive semidefinite matrices. By Theorem X.1.3 of \citet{bhatiaMatrixAnalysis1997}, 
    \begin{align*}
        \|A^{1/2} - B^{1/2}\|_F &\le \| |A- B|^{1/2}\|_F = \left(\Tr[|A- B|]\right)^{1/2}  = \left(\sum_{j=1}^d \sigma_j(A-B)\right)^{1/2},
    \end{align*}
    where $|M| := (M^\top M)^{1/2}$ and $\sigma_j(M)$ denotes the $j$th singular value of $M$. Applying the Cauchy--Schwarz inequality to the sum of singular values, 
    \begin{align*}
        \left(\Tr[|A- B|]\right)^{1/2} \le d^{1/4} \|A- B\|_F^{1/2},
    \end{align*}
    which yields
    $$|\Bcal^2(Q, \s) - \Bcal^2(Q, \s')| \le 2d^{3/4}\big\|Q^{1/2}\s' Q^{1/2} - Q^{1/2}\s Q^{1/2}\big\|_F^{1/2}.$$
    Applying the inequality $\|ML\|_F \le \|M\|_{op}\|L\|_F$ twice, we obtain
    $$|\Bcal^2(Q, \s) - \Bcal^2(Q, \s')| \le2d^{3/4} \|Q\|_{op}^{1/2} \|\s - \s'\|_F^{1/2}.$$

    \vspace{4mm}
    \noindent\textit{2.} The BW distance satisfies 
    $$\Bcal(Q, \s) = \inf_{M^\top M = Q, L^\top L = \s} \|M - L\|_F \le \|Q^{1/2} - \s^{1/2}\|_F,$$
    where the infimum is taken over $M, L\in\R^{d\times d}$ \citep[Theorem~1]{bhatiaBuresWassersteinDistance2019}. Since $$\lambda_{\min}(Q^{1/2} + \s^{1/2}) \ge \lambda_{\min}(Q^{1/2}) + \lambda_{\min}(\s^{1/2}) \ge \sqrt\lambda,$$
    Problem X.5.5 of \citet{bhatiaMatrixAnalysis1997} implies that 
    $$\|Q^{1/2} - \s^{1/2}\|_F \le \frac{1}{\sqrt{\lambda}}\|Q - \s\|_F,$$
    completing the proof.

\end{proof}

\section{Proof of the oracle case}\label{sec:supp-proof-oracle-rate}

This section is dedicated to proving the convergence of the oracle nonparanormal Fr\'echet regression estimator (Theorem~\ref{thm:frechet-oracle-rate}).

\subsection{Notation and outline}

The main part of the proof is the Fr\'echet regression of the latent correlation component. Recall that $F(z, \s) = \E[s(Z, z)\Bcal^2(S, \s)]$ is the population objective function, whose oracle empirical counterpart is $\tilde F_n(z, \s) = \frac{1}{n}\sum_{i=1}^ns_n(Z_i, z)\Bcal^2(S_i, \s)$. We derive the convergence rate from the rate of the gradient of $\tilde F_n$ at the population optimum $S_F(z)$, facilitated by a local Taylor expansion. To this end, we introduce the following notations of the gradient, Hessian, and third derivative of the objective functions:
\begin{align*}
    \tilde G_n(z) &= \frac{1}{n}\sum_{i=1}^ns_n(Z_i, z)(I_d - T_{ S_F(z)}^{S_i}), & G(z) &= \E[s(Z, z)(I_d - T_{ S_F(z)}^{S})], 
    \\
    \tilde H_n(z) &= -\frac{1}{n}\sum_{i=1}^ns_n(Z_i, z) dT_{ S_F(z)}^{S_i}, & H(z) &= -\E[s(Z, z)dT_{ S_F(z)}^{S}], 
    \\ 
    \tilde K_n(z, \s) &= -\frac{1}{n}\sum_{i=1}^ns_n(Z_i, z) d^2T_{\s}^{S_i}, & K(z, \s) &=  -\E[s(Z, z) d^2T_{\s}^{S}].
\end{align*}

\paragraph{Section outline.} 
Theorem~\ref{thm:frechet-oracle-rate} is proved in the subsequent Section~\ref{sec:supp-nptrate-pf}. The proof leverages technical results on the uniform convergence of the empirical objective function $\tilde F_n$, gradient $\tilde G_n$, Hessian $\tilde H_n$, and third-order derivative $\tilde K_n$ to their population counterparts, which are deferred to Section~\ref{sec:supp-uniform-convergences}. 

\subsection{Proof of Theorem~\ref{thm:frechet-oracle-rate}}\label{sec:supp-nptrate-pf}
We prove Theorem~\ref{thm:frechet-oracle-rate} through the following four steps:
\begin{enumerate}
    \item \textit{Preliminary uniform consistency (Lemma~\ref{lemma:slow-rate}):} We first derive a uniform convergence of $\tilde S_F(z)$ in terms of the Frobenius norm. This ensures that the estimator eventually lies within a local neighborhood where the objective function is strongly convex.
    \item \textit{Rate refinement via local Taylor expansion (Lemma~\ref{lemma:Frobenius-rate}):} Inside this local neighborhood, we employ a Taylor expansion of the gradient $\nabla \tilde F_n$ to establish that the estimator's rate can be inherited from the parametric convergence rate of the gradient $\tilde G_n$ projected to $\Tcal$.
    \item \textit{Metric transition:} The Frobenius norm rate is then translated into the BW metric $\Bcal$, completing the proof of the uniform parametric rate for the correlation component $\Bcal(\tilde S_F(z), S_F(z))$.
    \item \textit{Marginal parametric rates and final overall NPT rate:} We lastly derive the same uniform parametric rate for the marginal components $d_W(\tilde\omega_{F,j}(z), \omega_{F,j}(z))$ through a sharper analysis than \citet{petersenFrechetRegressionRandom2019}, leading to the same parametric rate for the overall NPT rate $d_{NPT}(\tilde\omega_F(z),\omega_F(z))$. This concludes the proof of Theorem~\ref{thm:frechet-oracle-rate}.
\end{enumerate}

\paragraph{Step 1: Preliminary uniform consistency.}

We first establish the uniform separation of the minimizer (often \textit{assumed} in standard $M$-estimation theory) using the compactness principle and joint continuity of the objective function we will establish in Lemma~\ref{lemma:F-joint-continuity}.

\begin{lemma}\label{lemma:uniform-separation}
    Under Assumption~\ref{assump:correlation-uniqueness}, the map $z\mapsto S_F(z)$ is continuous on $\{z\in\R^p:\|z\|\le B\}$. Moreover, for every $\delta>0$ such that the constraint set below is nonempty,
    $$\inf_{\|z\|\le B}\,
        \inf_{\s\in\Ecal_d:\|\s-S_F(z)\|_F\ge\delta}
        \big\{F(z,\s)-F(z,S_F(z))\big\}>0.$$
\end{lemma}

\begin{proof}
    By Lemma~\ref{lemma:F-joint-continuity}, $F$ is \textit{jointly continuous} on the compact product set $\{z\in\R^p:\|z\|\le B\}\times\Ecal_d$. This joint continuity and uniqueness Assumption~\ref{assump:correlation-uniqueness} imply that the \textit{regression function $z\mapsto S_F(z)$ is continuous} through the following standard compactness argument. 
    
    Fix any sequence $z_m\to z$. Since $\Ecal_d$ is compact, every subsequence of $\{S_F(z_m)\}$ has a further convergence subsequence $S_F(z_{m_k})\to S^*$. This subsequential limit $S^*$ has to minimize $F(z,\cdot)$ since $F(z_{m_k}, S_F(z_{m_k})) \le F(z_{m_k}, Q)$ for every $Q\in\Ecal_d$, which implies $F(z, S^*) \le F(z, Q), \ \forall Q\in\Ecal_d$ due to the \textit{joint continuity} of $F$. Since the minimizer is unique, it equals $S_F(z)$. Therefore, every subsequence of $\{S_F(z_m)\}$ has a convergent subsequence converging to the fixed point $S_F(z)$, implying that the entire sequence $\{S_F(z_m)\}$ also converges to $S_F(z)$.
    
    Hence,
    \[
        \mathcal K_\delta
        :=
        \big\{(z,\s)\in\R^p\times\Ecal_d: \|z\|\le B,\
        \|\s-S_F(z)\|_F\ge\delta\big\}
    \]
    is compact. The function
    \[
        (z,\s)\mapsto F(z,\s)-F(z,S_F(z))
    \]
    is continuous and strictly positive on $\mathcal K_\delta$ by uniqueness. It therefore attains a strictly positive minimum on $\mathcal K_\delta$.
\end{proof}

This separation in turn yields the following convergence result.

\begin{lemma}\label{lemma:slow-rate}
    Under Assumption~\ref{assump:correlation-uniqueness}, we have
    $$\sup_{\|z\|\le B}\|\tilde S_F(z) -  S_F(z)\|_F = o_p(1).$$
    Consequently, if Assumption~\ref{assump:correlation-mineigen} holds in addition, %
    $$\lim_{n\to\infty} \P \bigg( \inf_{\|z\|\le B}\lambda_{\min}(\tilde S_F(z)) \ge \frac{\lambda_0}{2}\bigg) = 1. $$
\end{lemma}

\begin{proof}
    The first convergence follows the standard $M$-estimation argument from the uniform convergence of objective functions, as elaborated in the following. Let 
    $$r_n = \sup_{(z, \s)\in\Theta} |\tilde F_n(z, \s) - F(z, \s)|,$$ 
    which is $O_p(n^{-1/2})$ by Lemma~\ref{lemma:F-uniform-convergence}. Since $\tilde S_F(z)$ minimizes $\tilde F_n(z, \cdot)$, we have $\tilde F_n(z, \tilde S_F(z)) \le \tilde F_n(z,  S_F(z))$. This implies
    \begin{align*}
        F(z, \tilde S_F(z)) - F(z,  S_F(z)) &= F(z, \tilde S_F(z)) - \tilde F_n(z, \tilde S_F(z)) + \tilde F_n(z, \tilde S_F(z)) - \tilde F_n(z,  S_F(z)) \\
        &\quad + \tilde F_n(z,  S_F(z)) - F(z,  S_F(z)) \\
        &\le r_n + 0 + r_n = 2r_n.
    \end{align*}
    Lemma~\ref{lemma:uniform-separation} and the rate $r_n = O_p(n^{-1/2})$ imply that, for any $\delta>0$, 
    $$2r_n < \inf_{\|z\|\le B} \inf_{\s:\|\s- S_F(z)\|_F \ge \delta} F(z, \s) - F(z,  S_F(z))$$
    on an event with probability approaching 1. On this event, it is impossible that
    $\|\tilde S_F(z)-S_F(z)\|_F\ge\delta$ for some $z$, hence
    $$\sup_{\|z\|\le B}\|\tilde S_F(z)-S_F(z)\|_F < \delta.$$
    Therefore, we have $\sup_{\|z\|\le B}\|\tilde S_F(z)-S_F(z)\|_F=o_p(1)$.

    For the second part, let $\lambda_0 > 0$ be the constant from Assumption~\ref{assump:correlation-mineigen}.
    By Weyl's inequality,
    $$ \lambda_{\min}(\tilde S_F(z)) \ge \lambda_{\min}( S_F(z)) - \|\tilde S_F(z) -  S_F(z)\|_{op} \ge \lambda_0 - \|\tilde S_F(z) -  S_F(z)\|_F. $$
    Since $\sup_{\|z\|\le B} \|\tilde S_F(z) -  S_F(z)\|_F = o_p(1)$, the right-hand side is eventually greater than $\lambda_0/2$ with probability approaching 1, uniformly in $z$, which concludes the proof.
\end{proof}

\paragraph{Step 2: Rate refinement via Taylor expansion.}

\begin{lemma}\label{lemma:Frobenius-rate}
    Under Assumptions \ref{assump:correlation-mineigen}--\ref{assump:local-convexity-radius}, we have
    $$\sup_{\|z\|\le B}\|\tilde S_F(z) -  S_F(z)\|_F = O_p(n^{-1/2}).$$
\end{lemma}

\begin{proof}
    Let $\Delta_n(z) = \tilde S_F(z) -  S_F(z) \in\Tcal$. We proceed on the event of Lemma~\ref{lemma:slow-rate}, where the uniform eigenvalue lower bound $\lambda_{\min}(\tilde S_F(z)) \ge \lambda_0/2$ holds with probability approaching one.
    
    Consider the scalar-valued function $g(\s) = \langle \nabla \tilde F_n(z, \s), \Delta_n(z) \rangle_F$ defined on $\Ecal_d$. Applying Taylor's theorem to $g(S)$ around $ S_F(z)$ yields
    \begin{align*}
        \langle \nabla \tilde F_n(z, \tilde S_F(z)), \Delta_n(z) \rangle_F &= g(\tilde S_F(z)) \\
        &= g( S_F(z)) + \nabla g( S_F(z))[\Delta_n(z)] + \frac{1}{2}\nabla^2 g(\overline{S}(z))[\Delta_n(z), \Delta_n(z)] \\
        &= \langle \nabla \tilde F_n(z,  S_F(z)), \Delta_n(z) \rangle_F + \nabla^2\tilde F_n(z,  S_F(z))[\Delta_n(z), \Delta_n(z)] \\
        &\quad + \frac{1}{2}\nabla^3 \tilde F_n(z, \overline{S}(z)) [\Delta_n(z), \Delta_n(z), \Delta_n(z)],
    \end{align*}
    where $\overline{S}(z)$ lies on the line segment connecting $ S_F(z)$ and $\tilde S_F(z)$. Since $\Ecal_d$ is convex, $\overline{S}(z) \in \Ecal_d$. Since $\tilde S_F(z)$ minimizes $\tilde F_n(z, \cdot)$ in the interior of $\Ecal_d$, it satisfies the first-order optimality condition $\langle \nabla \tilde F_n(z, \tilde S_F(z)), M \rangle_F = 0$ for all $M \in \Tcal$; i.e., the left-hand side of above is zero since $\Delta_n(z)\in\Tcal$.
    Rearranging the terms and using the notations $\tilde G_n, \tilde H_n, \tilde K_n$, we obtain
    \begin{equation}\label{eq:supp-pf-taylor-expansion}
        -\langle \tilde G_n(z), \Delta_n(z) \rangle_F = \tilde H_n(z)[\Delta_n(z), \Delta_n(z)] + \frac{1}{2} \tilde K_n(z, \overline{S}(z))[\Delta_n(z), \Delta_n(z), \Delta_n(z)].
    \end{equation}
    
    We bound each term in \eqref{eq:supp-pf-taylor-expansion} uniformly over $\|z\|\le B$ using convergence results in Section~\ref{sec:supp-uniform-convergences}.
    By Lemma~\ref{lemma:H-K-boundedness}, with probability approaching one, the empirical Hessian $\tilde H_n(z)$ satisfies
    \begin{equation}\label{eq:supp-pf-hessian-lower}
        \tilde H_n(z)[\Delta_n(z), \Delta_n(z)] \ge \frac{\kappa}{2} \|\Delta_n(z)\|_F^2,
    \end{equation}
    where $\kappa$ is from Assumption~\ref{assump:local-convexity-radius}.
    To obtain a bound of the cubic term $\tilde K_n$, recall that we are on the event on which $\lambda_{\min}(\tilde S_F(z)) \ge \lambda_0/2$ holds uniformly in $z$. Since $\overline{S}(z)$ is a convex combination of $\tilde S_F(z)$ and $ S_F(z)$ and since $\lambda_{\min}( S_F(z)) \ge \lambda_0$ uniformly by Assumption~\ref{assump:correlation-mineigen}, we also have $\lambda_{\min}(\overline{S}(z)) \ge \lambda_0/2$ on this joint event.
    Then, Lemma~\ref{lemma:H-K-boundedness} implies that 
    $$C_n := \sup_{\|z\|\le B}\|\tilde K_n(z, \overline{S}(z))\|_{\Tcal} = O_p(1).$$
    Thus,
    \begin{equation}\label{eq:supp-pf-third-upper}
        \left| \tilde K_n(z, \overline{S}(z))[\Delta_n(z), \Delta_n(z), \Delta_n(z)] \right| \le C_n \|\Delta_n(z)\|_F^3.
    \end{equation}
    Combining the bounds \eqref{eq:supp-pf-hessian-lower} and \eqref{eq:supp-pf-third-upper} with \eqref{eq:supp-pf-taylor-expansion}, we have
    \begin{align*}
        \frac{\kappa}{2} \|\Delta_n(z)\|_F^2 - \frac{C_n}{2} \|\Delta_n(z)\|_F^3 &\le \left| \langle \tilde G_n(z), \Delta_n(z) \rangle_F \right| \\
        &= \left| \langle \Pi_{\Tcal} \tilde G_n(z), \Delta_n(z) \rangle_F \right| \\
        &\le \|\Pi_{\Tcal} \tilde G_n(z)\|_F \|\Delta_n(z)\|_F
    \end{align*}
    with probability approaching one; here, $\Pi_\Tcal$ denotes the projection onto the tangent space $\Tcal$. This implies either $\|\Delta_n(z)\|_F = 0$ or
    $$ \frac{\kappa}{2} \|\Delta_n(z)\|_F - \frac{C_n}{2} \|\Delta_n(z)\|_F^2 \le \|\Pi_{\Tcal} \tilde G_n(z)\|_F. $$
    By Lemma~\ref{lemma:slow-rate}, $\sup_{\|z\|\le B}\|\Delta_n(z)\|_F = o_p(1)$. Since $C_n = O_p(1)$, we have $C_n\sup_{\|z\|\le B}\|\Delta_n(z)\|_F = o_p(1)$. Thus, for sufficiently large $n$, $\frac{C_n}{2} \|\Delta_n(z)\|_F \le \frac{\kappa}{4}$ uniformly over $\|z\|\le B$, implying that $\|\Delta_n(z)\|_F$ is controlled by the empirical gradient:
    $$ \frac{\kappa}{4} \|\Delta_n(z)\|_F \le \|\Pi_{\Tcal} \tilde G_n(z)\|_F. $$
    Finally, Lemma~\ref{lemma:G-uniform-convergence} states that $\sup_{\|z\|\le B} \|\Pi_{\Tcal} \tilde G_n(z)\|_F = O_p(n^{-1/2})$.
    Therefore,
    $$ \sup_{\|z\|\le B} \|\Delta_n(z)\|_F \le \frac{4}{\kappa} \sup_{\|z\|\le B} \|\Pi_{\Tcal} \tilde G_n(z)\|_F = O_p(n^{-1/2}), $$
    which completes the proof.
\end{proof}

\paragraph{Step 3: Convergence in BW metric.}

We then translate the Frobenius rate of Lemma~\ref{lemma:Frobenius-rate} into the BW metric, concluding the rate derivation for the latent correlation estimator $\tilde S_F(z)$.

By Lemma~\ref{lemma:slow-rate} and Assumption~\ref{assump:correlation-mineigen}, we have $\tilde S_F(z), S_F(z) \in \Ecal_d(\lambda_0 /2)$ uniformly over $\|z\|\le B$, on the event with probability approaching one. Then, on this event, the second part of Lemma~\ref{lemma:lipschitz-holder-bw} implies that there exists a constant $C_{\lambda_0} > 0$ depending only on $\lambda_0$ such that
$$ \Bcal(\tilde S_F(z),  S_F(z)) \le C_{\lambda_0} \|\tilde S_F(z) -  S_F(z)\|_F $$
uniformly over $\|z\|\le B$.
Combining this with Lemma~\ref{lemma:Frobenius-rate}, we conclude
$$ \sup_{\|z\|\le B} \Bcal(\tilde S_F(z),  S_F(z)) = O_p(n^{-1/2}). $$

\paragraph{Step 4: Marginal rates and the final overall rate in NPT. }
Finally, we complete the proof of Theorem~\ref{thm:frechet-oracle-rate} by establishing the parametric rates for the marginal components $d_W(\tilde \omega_{F, j}(z), \omega_{F, j}(z))$, $j=1,\dots,d$, uniformly over $\|z\|\le B$.

To this end, let $\tilde q_{F, j}(z)$ and $q_{F, j}(z)$ denote the corresponding quantile functions of $\tilde \omega_{F, j}(z)$ and $\omega_{F, j}(z)$, respectively. Recall that $q_{ij}$ denotes the quantile of the $j$th marginal $\omega_{ij}$ of the response $\omega_i$. We first perform Fr\'echet regression on the pairs $\{(Z_i, q_{ij})\}_{i=1}^n$, where the response metric space is the Hilbert space $L^2(0,1)$, to obtain a Fr\'echet regression estimator
$$\tilde q_{F, j}'(z) = \frac{1}{n}\sum_{i=1}^n s_n(Z_i, z)q_{ij}$$
and its population counterpart
$$q_{F, j}'(z) = \E[s(Z, z) q_j], $$
where $q_j$ is the quantile of the population marginal $\omega_j$ of $\omega$. These closed-form derivations rely crucially on the Hilbert norm $\|\cdot\|_{L^2}$ and satisfy
$$ \tilde q_{F, j}(z) = \Pi_{\Kcal} \tilde q_{F, j}'(z), \quad\text{and}\quad q_{F, j}(z) = \Pi_\Kcal q_{F, j}'(z), $$
where $\Pi_\Kcal: L^2(0,1)\to\Kcal$ denotes the $L^2$-projection map onto the quantile cone $\Kcal$; for detailed derivation, see, e.g., the proof of Theorem~1 of \citet{zhouWassersteinRegressionEmpirical2024}. 
Since $\Kcal$ is closed and convex in $L^2(0,1)$ \citep{bigotGeodesicPCAWasserstein2017}, we have
$$d_W(\tilde \omega_{F, j}(z), \omega_{F, j}(z)) = \|\tilde q_{F, j}(z)- q_{F, j}(z)\|_{L^2} \le \|\tilde q_{F, j}'(z)- q_{F, j}'(z)\|_{L^2} $$
due to the contraction property of the projection $\Pi_\Kcal$. Then, Corollary~2 of \citet{petersenFrechetRegressionRandom2019} implies
$$\sup_{\|z\|\le B} \|\tilde q_{F, j}'(z)- q_{F, j}'(z)\|_{L^2} = O_p(n^{-1/2}), $$
which deduces the desired uniform parametric rate
$$\sup_{\|z\|\le B} d_W(\tilde \omega_{F, j}(z), \omega_{F, j}(z)) = O_p(n^{-1/2}). $$

Lastly, these parametric rates of the marginal components $d_W(\tilde \omega_{F, j}(z), \omega_{F, j}(z))$ and the correlation $\Bcal(\tilde S_F(z), S_F(z))$, derived in Step 3, conclude the overall NPT rate
$$\sup_{\|z\|\le B} d_{NPT}(\tilde\omega_F(z), \omega_F(z)) = O_p(n^{-1/2}).$$
\qed

\subsection{Technical lemmas for the oracle analysis}\label{sec:supp-uniform-convergences}

Here, we establish various technical lemmas: continuity and uniform convergence results, with their consequences.

We first recall the uniform convergence of the global Fr\'echet regression weights $s_n(Z_i, z)$ to $s(Z, z)$, established in \citet{petersenFrechetRegressionRandom2019}. 
Write
\begin{align*}
    W_{0n}(z) &= -\overline{Z}^\top \widehat{\cov}(Z)^{-1}(z- \overline{Z}) + \E[Z]^\top \cov(Z)^{-1} (z- \E[Z]), \\
    W_{1n}(z) &= \widehat{\cov}(Z)^{-1}(z - \overline{Z}) - \cov(Z)^{-1}(z-\E[Z]),
\end{align*}
which yields the following decomposition
$$s_n(Z_i, z) - s(Z_i, z) = W_{0n}(z) + W_{1n}(z)^\top Z_i.$$
They implicitly assumed that $Z$ has finite fourth moment, $\cov(Z)^{-1}$ is well-defined, and $\|\cov(Z)^{-1}\|_{op} <\infty$. Then, the central limit theorem for $W_{0n}$ and $W_{1n}$ and the law of large numbers imply the following two uniform convergence/boundedness:
\begin{align}
    \sup_{\|z\|\le B} \frac{1}{n}\sum_{i=1}^n |s_n(Z_i, z) - s(Z_i, z)| &= O_p(n^{-1/2}), \quad\text{and}  \label{eq:weight-convergence} \\ 
    \sup_{\|z\|\le B} \frac{1}{n}\sum_{i=1}^n |s_n(Z_i, z)| &= O_p(1). \label{eq:weight-uniform-bound}
\end{align}
Using these results, we establish the \textit{joint continuity} of the population objective function.

\begin{lemma}\label{lemma:F-joint-continuity}
    The population objective
    $$F(z,\s)=\E\big[s(Z,z)\Bcal^2(S,\s)\big]$$
    is jointly continuous on $\{z\in\R^p:\|z\|\le B\}\times\Ecal_d$.
\end{lemma}

\begin{proof}
    We introduce the notation $f_{z,\s}(Z,S):=s(Z,z)\Bcal^2(S,\s)$ for later empirical process argument. For $(z,\s),(z',\s')\in\{z\in\R^p:\|z\|\le B\}\times\Ecal_d$,
    \begin{align*}
        |f_{z,\s}(Z,S)-f_{z',\s'}(Z,S)|
        &\le
        |s(Z,z)-s(Z,z')|\Bcal^2(S,\s)\\
        &\quad+
        |s(Z,z')|
        \big|\Bcal^2(S,\s)-\Bcal^2(S,\s')\big|.
    \end{align*}
    Since $\Bcal^2(S,\s)\le2d$ and
    \[
        |s(Z,z)-s(Z,z')|
        \le
        \|Z-\E[Z]\|\,
        \|\cov(Z)^{-1}\|_{op}\,
        \|z-z'\|,
    \]
    while Lemma~\ref{lemma:lipschitz-holder-bw} gives
    \[
        \big|\Bcal^2(S,\s)-\Bcal^2(S,\s')\big|
        \le 2d^{5/4}\|\s-\s'\|_F^{1/2},
    \]
    there exists $M(Z)=C(1+\|Z\|)$ with $C>0$ such that
    \begin{equation}\label{eq:objective-holder-envelope}
        |f_{z,\s}(Z,S)-f_{z',\s'}(Z,S)|
        \le
        M(Z)\left(
        \|z-z'\|+\|\s-\s'\|_F^{1/2}
        \right).
    \end{equation}
    Moreover, $\E[M(Z)]<\infty$. Taking expectations in~\eqref{eq:objective-holder-envelope} shows that
    \[
        |F(z,\s)-F(z',\s')|
        \le
        \E[M(Z)]
        \left(
        \|z-z'\|+\|\s-\s'\|_F^{1/2}
        \right),
    \]
    proving joint continuity.
\end{proof}

Next, we introduce intermediate quantities using the population weights to establish the required uniform convergence results:

\begin{align*}
    \check{F}_n(z, \s) &= \frac{1}{n}\sum_{i=1}^n s(Z_i, z)\Bcal^2(S_i, \s), & \check G_n(z) &= \frac{1}{n}\sum_{i=1}^n s(Z_i, z)\big(I_d - T_{ S_F(z)}^{S_i}\big), \\
    \check H_n(z) &= -\frac{1}{n}\sum_{i=1}^n s(Z_i, z) dT_{ S_F(z)}^{S_i}, & \check K_n(z,\s)&=  -\frac{1}{n}\sum_{i=1}^n s(Z_i, z)d^2T_{\s}^{S_i}.
\end{align*}
We prove the uniform convergence results for $\tilde F_n$, $\tilde G_n$, $\tilde H_n$, and $\tilde K_n$ and their consequences. 

\begin{lemma}\label{lemma:F-uniform-convergence}
    Let $\Theta = \{(z, \s) \in \R^p \times \Ecal_d: \|z\|\le B\}$. Then,
    $$\sup_{(z, \s)\in\Theta} \big|\tilde F_n (z, \s) - F(z, \s)\big| = O_p(n^{-1/2}).$$
\end{lemma}

\begin{proof}
    Write $\theta = (z, \s)$ for brevity.
    By the triangle inequality, $|\tilde F_n(\theta) - F(\theta)| \le |\tilde F_n(\theta) - \check F_n(\theta)| + |\check F_n(\theta) - F(\theta)|$.
    The first term is bounded by
    \begin{align*}
        \sup_{\theta\in\Theta}|\tilde F_n(\theta) - \check F_n(\theta)| &\le \sup_{(z, \s)\in\Theta} \frac{1}{n}\sum_{i=1}^n |s_n(Z_i, z) - s(Z_i, z)| \Bcal^2(S_i, \s),
    \end{align*}
    which is $O_p(n^{-1/2})$ by \eqref{eq:weight-convergence} and $\Bcal^2(S_i, \s) \le 2d$ for any $S_i, \s \in\Ecal_d$.
    
    For the second term, we follow the standard empirical process argument \citep[Section~19.2]{vaartAsymptoticStatistics2000}. Given the joint distribution $J= (Z, S)$, define the function class
    $$\Fcal = \{ f_{\theta}(J) = s(Z, z)\Bcal^2(S, \s): \theta=(z, \s) \in \Theta \}.$$ 
    We endow $\Theta$ with the metric $\|\theta - \theta'\|^2 = \|z - z'\|^2 + \|\s - \s'\|_F^2$ and show that $\Fcal$ is Donsker using the continuity argument in Lemma~\ref{lemma:F-joint-continuity}. For any $\theta = (z, \s), \theta' = (z', \s') \in \Theta$, the inequality \eqref{eq:objective-holder-envelope} yields the envelope
    $$|f_{\theta}(J) - f_{\theta'}(J)| \le M(Z) \|\theta - \theta'\|^{1/2},$$
    where $M(Z) = C (1 + \|Z\|)$ for some constant $C$, having a finite second moment. %
    Since $\Theta$ is a compact subset of a finite-dimensional Euclidean space (dimension $K$), the bracketing integral of $\Fcal$ is finite \citep[Example~19.7]{vaartAsymptoticStatistics2000}, implying that $\Fcal$ is a Donsker class. This concludes the uniform rate $\sup_{\Theta} |\check F_n - F| = O_p(n^{-1/2})$. %
\end{proof}

\begin{lemma}\label{lemma:G-uniform-convergence}
    Let $\Pi_\Tcal$ denote the orthogonal projection from $\Scal_d$ onto the tangent space $\Tcal = \{S\in\Scal_d: \diag(S) = 0\}$. Under Assumption~\ref{assump:correlation-mineigen}, we have
    $$\sup_{\|z\|\le B} \|\Pi_\Tcal \tilde G_n(z)\|_F = O_p(n^{-1/2})$$
\end{lemma}

\begin{proof}
    By triangle inequality and using the fact $\|\Pi_\Tcal \|\le 1$,
    \begin{align*}
        \sup_{\|z\|\le B}\|\Pi_\Tcal \tilde G_n(z)\|_F &\le \sup_{\|z\|\le B} \|\Pi_\Tcal (\tilde G_n(z) -\check G_n(z))\|_F + \sup_{\|z\|\le B} \|\Pi_\Tcal \check G_n(z)\|_F \\
        &\le \underbrace{\sup_{\|z\|\le B} \|\tilde G_n(z) - \check G_n(z)\|_F}_{(i)} + \underbrace{\sup_{\|z\|\le B} \|\Pi_\Tcal \check G_n(z)\|_F}_{(ii)}.
    \end{align*}

    For part (i), we use the uniform bound $\|I_d - T_{ S_F(z)}^{S_i}\|_F \le C $ over $\|z\|\le B$ for some constant $C >0 $ dependent on $\lambda_0$ (Lemma~\ref{lemma:derivative-bounds}), where $\lambda_0$ is from Assumption~\ref{assump:correlation-mineigen}. By the weight convergence \eqref{eq:weight-convergence}, it holds that
    \begin{align*}
        \sup_{\|z\|\le B} \|\tilde G_n(z) - \check G_n(z)\|_F &\le \sup_{\|z\|\le B} \frac{1}{n}\sum_{i=1}^n |s_n(Z_i, z) - s(Z_i, z)| \|I_d - T_{ S_F(z)}^{S_i}\|_F = O_p(n^{-1/2}).
    \end{align*}

    For part (ii), recall the notations $J=(Z, S)$ and $J_i = (Z_i, S_i)$ and consider the stochastic process in the \textit{variational form}
    $$f_{(z, \s, M)}(J) = s(Z, z)\langle I_d - T_{\s}^{S}, M\rangle_F$$
    parametrized by
    $$ \Omega' = \{(z, \s, M)\in\R^p\times \Ecal_d(\lambda_0)\times \Tcal:\|z\|\le B,\ \|M\|_F \le 1 \} $$
    where $\Ecal_d(\lambda_0)$ is defined in \eqref{eq:corr-eigen-lowerbound}.
    Since $S_F(z)\in\Ecal_d(\lambda_0)$ uniformly over $\|z\|\le B$ (Assumption~\ref{assump:correlation-mineigen}), the function class $\{f_{(z,  S_F(z), M)}:\|z\|\le B,\, M\in\Tcal\}$ is a subclass of $\{f_{(z, \s, M)}:(z, \s, M)\in\Omega'\}$.
    For any two indices $(z, \s, M), (z', \s', M')\in\Omega'$, by the triangle inequality and Lemmas~\ref{lemma:derivative-bounds} and \ref{lemma:lipschitz-transport},
    \begin{align*}
        |f_{(z, \s, M)}(J) - f_{(z', \s', M')}(J)|
        &\le |s(Z, z) - s(Z, z')| |\langle I_d - T_{\s}^{S}, M\rangle_F| \\
        &\quad + |s(Z, z')| |\langle I_d - T_{\s}^{S}, M - M'\rangle_F| \\
        &\quad + |s(Z, z')| |\langle T_{\s}^{S} - T_{\s'}^{S}, M'\rangle_F| \\
        &\le M'(Z) (\|z - z'\| + \|\s - \s'\|_F + \|M - M'\|_F),
    \end{align*}
    where $M'(Z) = C'(1+\|Z\|)$ for some $C'>0$.
    As $\Omega'$ is finite-dimensional and compact, an analogous argument to Lemma~\ref{lemma:F-uniform-convergence} implies that the class $\{f_{(z, \s, M)}:(z, \s, M)\in\Omega'\}$ is Donsker. Restricting the resulting uniform convergence to the subclass $\{f_{(z,  S_F(z), M)}:\|z\|\le B,\, M\in\Tcal\}$ and using the optimality condition $\Pi_\Tcal G(z) = 0 $, we have 
    \begin{align*}
        \sup_{\|z\|\le B} \|\Pi_\Tcal \check G_n(z)\|_F &= \sup_{\|z\|\le B} \sup_{M\in\Tcal,\,\|M\|_F\le 1} \big|\langle \check G_n(z) - G(z), M\rangle_F \big|\\
        &= \sup_{\|z\|\le B} \sup_{M\in\Tcal,\,\|M\|_F\le 1} \Bigg| \frac{1}{n}\sum_{i=1}^n f_{(z, S_F(z), M)}(J_i) - f_{(z, S_F(z), M)}(J) \Bigg| 
        = O_p(n^{-1/2}).
    \end{align*}
\end{proof}

\begin{lemma}\label{lemma:H-K-uniform-convergence}
    Under Assumption~\ref{assump:correlation-mineigen}, we have
    $$\sup_{\|z\|\le B} \|\tilde H_n(z) - H(z)\|_{\Tcal} = O_p(n^{-1/2}),$$
    and for any $\lambda > 0$,
    $$\sup_{\|z\|\le B,\,  \s\in\Ecal_d(\lambda)} \|\tilde K_n(z, \s) - K(z, \s)\|_{\Tcal} = O_p(n^{-1/2}),$$
    where $\Ecal_d(\lambda)$ is defined in \eqref{eq:corr-eigen-lowerbound}.
\end{lemma}

\begin{proof}
    As ideas parallel to those in the proof of Lemma~\ref{lemma:G-uniform-convergence}, we sketch the main steps. For the Hessian,
    \begin{align*}
        \sup_{\|z\|\le B} \|\tilde H_n(z) - H(z)\|_{\Tcal}
        &\le \sup_{\|z\|\le B} \|\tilde H_n(z) - \check H_n(z)\|_{\Tcal} + \sup_{\|z\|\le B} \|\check H_n(z) - H(z)\|_{\Tcal}.
    \end{align*}
    The first term is bounded by Lemma~\ref{lemma:derivative-bounds} and \eqref{eq:weight-convergence}.

    For the second term, we similarly consider the stochastic process in the variational form
    $$g_{(z, \s, M)}(J) = -s(Z, z) dT_{\s}^{S}[M, M]$$
    over the compact set $\{(z, \s, M)\in\R^p\times \Ecal_d(\lambda_0)\times \Tcal: \|z\|\le B,\, \|M\|_F \le 1\}$. By the analogous argument to Lemma~\ref{lemma:G-uniform-convergence}, the bounds in Lemmas~\ref{lemma:derivative-bounds} and \ref{lemma:lipschitz-transport} imply that the process $g_{(z, \s, M)}(J)$ satisfies the Lipschitz property, proving that the class is Donsker. This yields the convergence
    $$\sup_{\|z\|\le B} \|\check H_n(z) - H(z)\|_{\Tcal} = \sup_{\|z\|\le B}\sup_{M\in\Tcal,\,\|M\|_F\le 1}  (\check H_n(z) - H(z))[M, M] = O_p(n^{-1/2}).$$

    For the third differential, decompose
    \begin{align*}
        \sup_{\|z\|\le B,\,  \s\in\Ecal_d(\lambda)} \|\tilde K_n(z, \s) - K(z, \s)\|_{\Tcal}
        &\le \sup_{\|z\|\le B,\,  \s\in\Ecal_d(\lambda)} \|\tilde K_n(z, \s) - \check K_n(z, \s)\|_{\Tcal} \\
        &\qquad + \sup_{\|z\|\le B,\,  \s\in\Ecal_d(\lambda)} \|\check K_n(z, \s) - K(z, \s)\|_{\Tcal}.
    \end{align*}
    The first term is $O_p(n^{-1/2})$ by \eqref{eq:weight-convergence} and Lemma~\ref{lemma:derivative-bounds} (applied with the eigenvalue lower bound $\lambda\wedge\lambda_0$). For the second term, consider the variational form
    $$h_{(z, \s, M)}(J) = -s(Z, z) d^2T_{\s}^{S}[M, M, M]$$
    over the compact set $\{(z, \s, M)\in\R^p\times \Ecal_d(\lambda)\times \Tcal: \|z\|\le B,\, \|M\|_F \le 1\}$. By Lemmas~\ref{lemma:derivative-bounds} and \ref{lemma:lipschitz-transport}, this class is Donsker, which yields the desired $O_p(n^{-1/2})$ rate.
\end{proof}

\begin{lemma}\label{lemma:H-K-boundedness}
    Under Assumptions~\ref{assump:correlation-mineigen} and \ref{assump:local-convexity-radius}, we have
    $$\lim_{n\to\infty} \P \bigg(\inf_{\|z\|\le B} \inf_{M\in\Tcal,\, \|M\|_F= 1} \tilde H_n(z)[M, M] \ge \frac{\kappa}{2}\bigg) = 1. $$
    Furthermore, for any $\lambda > 0$,
    $$\sup_{\|z\|\le B,\,  \s\in\Ecal_d(\lambda)} \|\tilde K_n(z, \s)\|_{\Tcal}  = O_p(1).$$
\end{lemma}

\begin{proof}
    By Assumption~\ref{assump:local-convexity-radius}, the population Hessian satisfies $H(z)[M, M] \ge \kappa \|M\|_F^2$ for all $M\in\Tcal$ and all $\|z\|\le B$.
    Using the uniform convergence result in Lemma~\ref{lemma:H-K-uniform-convergence}, we have
    \begin{align*}
        \tilde H_n(z)[M, M] &= H(z)[M, M] + (\tilde H_n(z) - H(z))[M, M] \\
        &\ge \kappa \|M\|_F^2 - \|\tilde H_n(z) - H(z)\|_{\Tcal} \|M\|_F^2 \\
        &\ge (\kappa - o_p(1)) \|M\|_F^2
    \end{align*}
    uniformly over $\|z\|\le B$ and $M\in\Tcal$. This proves the first claim.
    
    For the second claim, we use the triangle inequality:
    $$\|\tilde K_n(z, \s)\|_{\Tcal} \le \|K(z, \s)\|_{\Tcal} + \|\tilde K_n(z, \s) - K(z, \s)\|_{\Tcal}.$$
    The second term is $O_p(n^{-1/2})$ by Lemma~\ref{lemma:H-K-uniform-convergence}.
    The first term is uniformly bounded because the function $(z, \s) \mapsto K(z, \s)$ is continuous (Lemma~\ref{lemma:transport-smoothness}) on the compact set $\{(z, \s): \|z\|\le B,  \s\in\Ecal_d(\lambda)\}$.
\end{proof}

\section{Proof of the empirical-response case and the Wasserstein rate}\label{sec:supp-empirical-rate-proof}

In this section, we prove Theorem~\ref{thm:frechet-real-rate} in the case responses $\hat\omega_i$ are empirically estimated from $N_i$ samples from each $\omega_i$. 

By the triangle inequality and the oracle result (Theorem~\ref{thm:frechet-oracle-rate}),
$$d_{NPT}(\hat\omega_F(z), \omega_F(z)) \le d_{NPT}(\hat\omega_F(z), \tilde\omega_F(z)) + d_{NPT}(\tilde\omega_F(z), \omega_F(z)), $$
it remains to control the sampling error within the discrepancy $d_{NPT}(\hat\omega_F(z), \tilde\omega_F(z))$. Its marginal components $d_W(\hat\omega_{F, j}(z), \tilde\omega_{F,j}(z))$ achieves the $O_p(r_N)$ rate as proved in \citet{zhouWassersteinRegressionEmpirical2024}; in their paper, the authors expressed the rate as $O_p(N^{-1/4})$ under the compact support condition, while here we use the general expression $O_p(r_N)$ (the same argument obtains this). Therefore, it remains to control the correlation component $\Bcal^2(\hat S_F(z), \tilde S_F(z))$.

To this end, we control the gradient of the empirical objective function $\hat F_n(z, \s)$ analogously to the oracle case. We recall the notations for the objective function and its minimizer, as well as introduce notations for the gradient and the Hessian:
\begin{align*}
    \hat F_n(z, \s) &= \frac{1}{n}\sum_{i=1}^n s_n(Z_i, z) \Bcal^2(\hat{S}_i, \s), & \hat S_F(z) &= \argmin_{ \s\in\Ecal_d} \hat F_n(z, \s), \\
    \hat G_n(z, \s) &= \frac{1}{n}\sum_{i=1}^n s_n(Z_i, z)\big(I_d - T_{\s}^{\hat{S}_i}\big), & \hat H_n(z, \s) &=  -\frac{1}{n}\sum_{i=1}^n s_n(Z_i, z) dT_{\s}^{\hat{S}_i}.
\end{align*}
Our proof relies on a technical \textit{perturbation argument for the Hessian} $\hat H_n$ (Lemma~\ref{lemma:uniform-pd-hat}) and \textit{convexity of the set} $\Ecal_d$ (Section~\ref{sec:supp-real-case-proof}) to bound $\Bcal^2(\hat S_F(z), \tilde S_F(z))$ in terms of $\hat G_n$ and $\tilde G_n$. We then derive the gradient $O_p(N^{-1/2})$ rate using differential properties in Section~\ref{sec:supp-diff-bw-prop} and Cauchy--Schwarz inequality, combined with Lemma~\ref{lemma:supp-weight-and-correlation-mse}. The rest of the details of the proof are given in Section~\ref{sec:supp-real-case-proof}, and necessary technical lemmas are deferred to Section~\ref{sec:supp-realcase-lemmas}.

\subsection{Proof of Theorem~\ref{thm:frechet-real-rate}}\label{sec:supp-real-case-proof}

We bound the oracle--empirical discrepancy $\Delta_n(z) = \hat S_F(z) - \tilde S_F(z)$ in the Frobenius norm, then translate it into the BW metric.

Differentiating the empirical functional $\hat F_n(z, \cdot)$, consider the scalar-valued function 
$$h(\s) = \big\langle \nabla_\s \hat F_n(z, \s), \Delta_n(z) \big\rangle_F = \big\langle \hat G_n(z, \s), \Delta_n(z)\big\rangle_F.$$
Lemma~\ref{lemma:consistency-hat} implies that the event $\{ \inf_{\|z\|\le B}\lambda_{\min}(\hat S_F(z))\ge \lambda_0/2\}$ holds with probability approaching one. Proceeding on this event, since $\hat S_F(z)$ minimizes $\hat F_n(z, \cdot)$ over $\Ecal_d$ and $\Delta_n(z)\in\Tcal$, we have $h(\hat S_F(z)) = 0$ by the optimality condition. Applying the mean value theorem to $h(\s)$ around $\tilde S_F(z)$, there exists an intermediate point $\overline{S}(z)$ on the line segment connecting $\hat S_F(z)$ and $\tilde S_F(z)$ such that
\begin{equation}\label{eq:empirical-response-Taylor}
    0 = h(\hat S_F(z)) = \langle \hat G_n(z, \tilde S_F(z)), \Delta_n(z) \rangle_F + \hat H_n(z, \overline{S}(z))[\Delta_n(z), \Delta_n(z)]. 
\end{equation}
Here, it is crucial to note that $\overline S(z)$ is contained in the interior of $\Ecal_d$ because the domain is \textit{convex}.

Define the gradient of the oracle objective $\tilde F_n$ by
$$\tilde G_n(z, \s) := \nabla_\s \tilde F_n(z, \s) = \frac{1}{n}\sum_{i=1}^n s_n(Z_i, z)\big(I_d - T_{\s}^{S_i}\big).$$
Since $\tilde S_F(z)$ minimizes $\tilde F_n(z, \cdot)$, we also have $\langle \tilde G_n(z, \tilde S_F(z)), \Delta_n(z) \rangle_F = 0$ on the event with probability approaching 1 in Lemma~\ref{lemma:slow-rate}, which also ensures $\tilde S_F(z)\succ 0$ with high probability. Subtracting this from the right-hand side of \eqref{eq:empirical-response-Taylor} and rearranging the terms, we obtain
$$ - \langle \hat G_n(z, \tilde S_F(z)) - \tilde G_n(z, \tilde S_F(z)), \Delta_n(z) \rangle_F = \hat H_n(z, \overline{S}(z))[\Delta_n(z), \Delta_n(z)]. $$

To bound the right-hand side, we need the uniform positive definiteness of the empirical Hessian $\hat H_n(z, \cdot)$ at $\overline{S}(z)$.
First, Lemma~\ref{lemma:consistency-hat} establishes the consistency $\sup_{\|z\|\le B} \|\hat S_F(z) -  S_F(z)\|_F = o_p(1)$. Since $\tilde S_F(z)$ is also uniformly consistent, the intermediate point $\overline{S}(z)$ converges uniformly to $S_F(z)$.
Then, the perturbation argument Lemma~\ref{lemma:uniform-pd-hat} guarantees that with probability approaching one,
$$ \inf_{\|z\|\le B} \hat H_n(z, \overline{S}(z))[M, M] \ge \frac{\kappa}{4}\|M\|_F^2 $$
for all $M\in\Tcal$. Using this lower bound and the Cauchy--Schwarz inequality,
$$ \frac{\kappa}{4} \|\Delta_n(z)\|_F^2 \le \|\hat G_n(z, \tilde S_F(z)) - \tilde G_n(z, \tilde S_F(z))\|_F \|\Delta_n(z)\|_F, $$
which implies either $\|\Delta_n(z)\|_F= 0$ or 
$$ \frac{\kappa}{4} \|\Delta_n(z)\|_F \le \|\hat G_n(z, \tilde S_F(z)) - \tilde G_n(z, \tilde S_F(z))\|_F. $$

Therefore, it remains to bound this difference between the gradients. By Lemma~\ref{lemma:slow-rate}, the event $\{ \inf_{\|z\|\le B}\lambda_{\min}(\tilde S_F(z))\ge \lambda_0/2\}$ holds with probability approaching one. Also, by Lemma~\ref{lemma:supp-empirical-corrlowerbound}, the event $\{\inf_i \lambda_{\min}(\hat S_i) \ge \lambda_0/2\}$ holds with probability approaching one as $n\to \infty$ under Assumption~\ref{assump:observation-sizes}. On this joint event, Lemma~\ref{lemma:lipschitz-transport} implies that the optimal transport map $Q \mapsto T_{\tilde S_F(z)}^Q$ is $C$-Lipschitz continuous with respect to the Frobenius norm, whenever $Q\in\Ecal_d(\lambda_0/2)$. Using this Lipschitz property and applying the Cauchy--Schwarz inequality, we have uniformly over $\|z\|\le B$:
\begin{align*}
    \|\hat G_n(z, \tilde S_F(z)) - \tilde G_n(z, \tilde S_F(z))\|_F &\le \frac{1}{n}\sum_{i=1}^n |s_n(Z_i, z)| \|T_{\tilde S_F(z)}^{\hat{S}_i} - T_{\tilde S_F(z)}^{S_i}\|_F \\
    &\le \frac{C}{n}\sum_{i=1}^n |s_n(Z_i, z)| \|\hat{S}_i -  S_i\|_F \\
    &\le C \left(\frac{1}{n}\sum_{i=1}^n s_n(Z_i, z)^2\right)^{1/2} \left(\frac{1}{n}\sum_{i=1}^n \|\hat{S}_i -  S_i\|_F^2\right)^{1/2}.
\end{align*}
By Lemma~\ref{lemma:supp-weight-and-correlation-mse}, the right-hand side is uniformly $O_p(N^{-1/2})$, and thus
$$ \sup_{\|z\|\le B} \|\hat G_n(z, \tilde S_F(z)) - \tilde G_n(z, \tilde S_F(z))\|_F = O_p(N^{-1/2}). $$
Consequently, the convergence rate of $\Delta_n(z)$ satisfies:
$$\sup_{\|z\|\le B} \|\hat S_F(z) - \tilde S_F(z)\|_F = O_p(N^{-1/2}),$$
which yields the combined rate
$$\sup_{\|z\|\le B} \|\hat S_F(z) -  S_F(z)\|_F = O_p(n^{-1/2} + N^{-1/2}).$$
As in the proof of Theorem~\ref{thm:frechet-oracle-rate}, this Frobenius rate translates into the BW metric since $\hat S_F(z)$ eventually lies in the interior of $\Ecal_d$:
$$\sup_{\|z\|\le B} \Bcal(\hat S_F(z),  S_F(z)) = O_p(n^{-1/2} + N^{-1/2}).$$
Finally, as mentioned in the beginning of this section, for each $j=1,\ldots,d$, the same argument as in \citet{zhouWassersteinRegressionEmpirical2024} yields the marginal rates $\sup_{\|z\|\le B} d_W(\hat\omega_{F, j}(z), \omega_{F, j}(z)) = O_p(n^{-1/2} + r_N)$.
Combining these rates completes the proof of Theorem~\ref{thm:frechet-real-rate}.
\qed

\subsection{Technical lemmas}\label{sec:supp-realcase-lemmas}

\begin{lemma}\label{lemma:supp-empirical-corrlowerbound}
    Under Assumptions~\ref{assump:correlation-mineigen} and \ref{assump:observation-sizes},
    $$ \lim_{n\to\infty} \P \bigg(\min_{i=1,\ldots, n} \lambda_{\min}(\hat{S}_i) \ge \frac{\lambda_0}{2} \bigg) = 1.$$
\end{lemma}

\begin{proof}
    By Weyl's inequality, we have $|\lambda_{\min}(\hat{S}_i) - \lambda_{\min}( S_i)| \le \|\hat{S}_i -  S_i\|_{op}$. Using the matrix norm relation $\|A\|_{op}\le \|A\|_F $, we have
    $$\lambda_{\min}(\hat{S}_i) \ge \lambda_{\min}( S_i) - \|\hat{S}_i -  S_i\|_{op} \ge \lambda_{\min}( S_i) - \|\hat{S}_i -  S_i\|_F.$$
    Under Assumption~\ref{assump:correlation-mineigen}, $\lambda_{\min}( S_i) \ge \lambda_0$ almost surely for all $i$.
    Thus, the event $\{\min_{i} \lambda_{\min}(\hat{S}_i) \ge \lambda_0/2\}$ is implied by the event $\{\max_{i} \|\hat{S}_i -  S_i\|_{F} \le \lambda_0/2\}$.

    It thus remains to show that
    $\P\big( \max_{i=1,\ldots,n} \|\hat{S}_i -  S_i\|_{F} > \lambda_0/2 \big)\to 0$.
    Using the element-wise concentration of $\hat S_i$ as used in \eqref{eq:element-wise-concentration}, expanding the squared Frobenius norm followed by the union bound implies
    \begin{align*}
        \P(\|\hat S_i - S_i\|_F > \lambda_0/2)
        &\le 2d^2\exp\left(-N_i(\lambda_0/2\pi d)^2\right).
    \end{align*}
    Applying the union bound again and $N=\min_i N_i$, 
    \begin{align*}
        \P\left( \max_{i=1,\ldots,n} \|\hat{S}_i -  S_i\|_{F} > \frac{\lambda_0}{2} \right) 
        &\le \sum_{i=1}^n \P\left( \|\hat{S}_i -  S_i\|_{F} > \frac{\lambda_0}{2} \right) \\
        &\le 2n d^2 \exp\left( - N \left(\frac{\lambda_0}{2\pi d}\right)^2 \right).
    \end{align*}
    As $N/\log{n} \to \infty$ as $n\to \infty$ (Assumption~\ref{assump:observation-sizes}), the final term converges to zero, establishing the desired convergence in probability.

\end{proof}

\begin{lemma}\label{lemma:supp-weight-and-correlation-mse}
    Under Assumption~\ref{assump:observation-sizes}, we have
    \begin{align*}
        \sup_{\|z\|\le B} \frac{1}{n}\sum_{i=1}^n s_n(Z_i, z)^2 &= O_p(1),\\
        \frac{1}{n}\sum_{i=1}^n \|\hat{S}_i -  S_i\|_F^2 &= O_p(N^{-1}).
    \end{align*}
    Consequently,
    $$\sup_{\|z\|\le B}\left(\frac{1}{n}\sum_{i=1}^n s_n(Z_i, z)^2\right)^{1/2}\left(\frac{1}{n}\sum_{i=1}^n \|\hat{S}_i -  S_i\|_F^2\right)^{1/2} = O_p(N^{-1/2}).$$
\end{lemma}

\begin{proof}
    By a direct computation, for all $z$,
    $$ \frac{1}{n}\sum_{i=1}^n s_n(Z_i, z)^2 = 1 + (z - \overline{Z})^\top \widehat{\cov}(Z)^{-1} (z - \overline{Z}). $$
    This is uniformly $O_p(1)$ by the law of large numbers and compactness of $\|z\|\le B$.

    For the second claim, we use the upper bound \eqref{eq:Frob-upperbound-expectation} obtained in the proof of Theorem~\ref{thm:npt-convergence-rate}: 
    $$\E[\|\hat{S}_i -  S_i\|_F^2]\le 2d(d-1)\pi^2/N_i \le 2d(d-1)\pi^2/N.$$
    By the linearity of expectation and Markov's inequality,
    $$ \frac{1}{n}\sum_{i=1}^n \|\hat{S}_i -  S_i\|_F^2 = O_p(N^{-1}), $$
    completing the proof.
\end{proof}

\begin{lemma}\label{lemma:consistency-hat}
    Under Assumptions \ref{assump:correlation-mineigen}, \ref{assump:correlation-uniqueness} and \ref{assump:observation-sizes},
    $$ \sup_{\|z\|\le B} \|\hat S_F(z) -  S_F(z)\|_F = o_p(1). $$
\end{lemma}

\begin{proof}
    We first show the uniform convergence of the objective function $\hat F_n$ to $F$ on the set $\Theta = \{(z, \s)\in\R^p\times \Ecal_d: \|z\|\le B \}$ as defined in Lemma~\ref{lemma:F-uniform-convergence}.
    By triangle inequality,
    $$ \sup_{(z, \s)\in\Theta} |\hat F_n(z, \s) - F(z, \s)| \le \sup_{(z, \s)\in\Theta} |\hat F_n(z, \s) - \tilde F_n(z, \s)| + \sup_{(z, \s)\in\Theta} |\tilde F_n(z, \s) - F(z, \s)|. $$
    Lemma~\ref{lemma:F-uniform-convergence} establishes that the second term is $O_p(n^{-1/2}) = o_p(1)$.
    For the first term, the first inequality of Lemma~\ref{lemma:lipschitz-holder-bw} implies that there exists a constant $C > 0$ such that
    $$ |\hat F_n(z, \s) - \tilde F_n(z, \s)| \le C\,\frac{1}{n}\sum_{i=1}^n |s_n(Z_i, z)| \cdot \|\hat{S}_i -  S_i\|_F^{1/2}.
    $$
    Applying the Cauchy--Schwarz inequality twice yields
    \begin{align*}
        \sup_{(z, \s)\in\Theta} |\hat F_n(z, \s) - \tilde F_n(z, \s)|
        &\le C \sup_{\|z\|\le B}\left(\frac{1}{n}\sum_{i=1}^n s_n(Z_i, z)^2\right)^{1/2} \left(\frac{1}{n}\sum_{i=1}^n \|\hat S_i -  S_i\|_F\right)^{1/2}.
        \\
        &\le C \sup_{\|z\|\le B}\left(\frac{1}{n}\sum_{i=1}^n s_n(Z_i, z)^2\right)^{1/2} \left(\frac{1}{n}\sum_{i=1}^n \|\hat S_i -  S_i\|_F^2\right)^{1/4}.
    \end{align*}
    By Lemma~\ref{lemma:supp-weight-and-correlation-mse}, we obtain the rate of the right-hand side as 
    $$\sup_{(z, \s)\in\Theta} |\hat F_n(z, \s) - F(z, \s)| = O_p(n^{-1/2} + N^{-1/4}) = o_p(1).$$
    
    Building on this uniform convergence of $\hat F_n$ to $F$ and the well-separation of the minimum (Lemma~\ref{lemma:uniform-separation}), the standard M-estimation argument (analogous to Lemma~\ref{lemma:slow-rate}) implies the uniform consistency of $\hat S_F(z)$.
\end{proof}

\begin{lemma}\label{lemma:uniform-pd-hat}
    Let $S_n(z)$ be \textit{any} sequence of random $\Ecal_d$-valued functions such that 
    $$\sup_{\|z\|\le B} \|S_n(z) -  S_F(z)\|_F = o_p(1).$$
    Under Assumptions \ref{assump:correlation-mineigen}--\ref{assump:observation-sizes}, the empirical Hessian $\hat H_n(z,\cdot)$ becomes uniformly positive definite at $S_n(z)$:
    $$\lim_{n\to\infty} \P\bigg(\inf_{\|z\|\le B} \inf_{M\in\Tcal,\, \|M\|_F= 1}\hat H_n(z, S_n(z))[M, M] \ge \frac{\kappa}{4} \bigg) = 1. $$
\end{lemma}

\begin{proof}
    We decompose the Hessian as:
    $$ \hat H_n(z, S_n(z)) = \underbrace{\tilde H_n(z,  S_F(z))}_{(i)} + \underbrace{[\tilde H_n(z, S_n(z)) - \tilde H_n(z,  S_F(z))]}_{(ii)} + \underbrace{[\hat H_n(z, S_n(z)) - \tilde H_n(z, S_n(z))]}_{(iii)}. $$
    From Lemma~\ref{lemma:H-K-boundedness}, the first term $(i)$ satisfies $$\inf_{\|z\|\le B}\inf_{M\in\Tcal,\, \|M\|_F= 1}\tilde H_n(z,  S_F(z))[M, M] \ge \kappa/2$$ with probability approaching one.
    
    For the second term $(ii)$, we use the smoothness of the oracle Hessian $\tilde H_n(z, \cdot)$.
    Recall that $\tilde H_n(z, \s) = \nabla_\s^2 \tilde F_n(z, \s)$ and $\tilde K_n(z, \s) = \nabla_\s^3 \tilde F_n(z, \s)$ as defined in Section~\ref{sec:supp-proof-oracle-rate}. By Lemma~\ref{lemma:H-K-boundedness}, 
    $$\sup_{\|z\|\le B,\,  \s\in\Ecal_d(\lambda_0/2)} \|\tilde K_n(z, \s)\|_{\Tcal}  = O_p(1),$$
    with $\lambda_0$ from Assumption~\ref{assump:correlation-mineigen}. Since $S_n(z)$ converges uniformly to $ S_F(z)$, $\inf_{\|z\|\le B} \lambda_{\min}(S_n(z))\ge \lambda_0/2 $ holds with probability approaching one. Then, by mean value inequality (Lemma~\ref{lemma:mean-value-inequality}) and smoothness,
    \begin{align*}
        \sup_{\|z\|\le B}\|\tilde H_n(z, S_n(z)) - \tilde H_n(z,  S_F(z))\|_{\Tcal} &\le \left(\sup_{\|z\|\le B,\,  \s\in\Ecal_d(\lambda_0/2)} \|\tilde K_n(z, \s)\|_{\Tcal}\right)\cdot \sup_{\|z\|\le B}\|S_n(z) -  S_F(z)\|_F \\
        &= o_p(1).
    \end{align*}
    
    For the third term $(iii)$, we control the gap between $\hat H$ and $\tilde H$ using the global Lipschitz continuity of $Q \mapsto dT_\s^Q$ on $\Ecal_d(\lambda_0/2)$ (Lemma~\ref{lemma:lipschitz-transport}).
    On the joint event $\{\inf_{\|z\|\le B}\lambda_{\min}(S_n(z)) \ge \lambda_0/2\}\cap\{\inf_i \lambda_{\min}(\hat S_i) \ge \lambda_0/2\}$, which has probability approaching one, there exists a constant $C>0$ such that
    \begin{align*}
        \sup_{\|z\|\le B}\|\hat H_n(z, S_n(z)) - \tilde H_n(z, S_n(z))\|_{\Tcal}
        &\le \sup_{\|z\|\le B}\frac{1}{n}\sum_{i=1}^n |s_n(Z_i, z)| \|dT_{S_n(z)}^{\hat S_i} - dT_{S_n(z)}^{S_i}\|_{\Tcal} \\
        &\le \sup_{\|z\|\le B}\frac{C}{n}\sum_{i=1}^n |s_n(Z_i, z)| \|\hat S_i -  S_i\|_F \\
        &\le C \left(\sup_{\|z\|\le B}\frac{1}{n}\sum_{i=1}^n s_n(Z_i, z)^2\right)^{1/2} \left(\frac{1}{n}\sum_{i=1}^n \|\hat S_i -  S_i\|_F^2\right)^{1/2},
    \end{align*}
    which is $o_p(1)$ by Lemma~\ref{lemma:supp-weight-and-correlation-mse}.
    
    Therefore,
    $$ \inf_{\|z\|\le B} \inf_{M\in\Tcal,\, \|M\|_F= 1}\hat H_n(z, S_n(z))[M, M] \ge \frac{\kappa}{2} - o_p(1) - o_p(1), $$
    which is at least $\kappa/4$ for sufficiently large $n$ with probability approaching 1.
\end{proof}

\subsection{Proof of Corollary~\ref{cor:wasserstein-rate}}\label{sec:supp-wass-cor-proof}

By Theorem~\ref{thm:topological-equivalence}, 
\begin{align*}
    d_{W}(\tilde \omega_F(z),\omega_F(z)) &\le \sqrt{2}\big(1 \vee |T_{\bgamma}^{\bomega_F(z)}|_{H^1(\gamma)}\big) d_{NPT}(\tilde\omega_F(z),\omega_F(z)), \quad\text{and} \\ 
    d_{W}(\hat \omega_F(z),\omega_F(z)) &\le \sqrt{2}\big(1 \vee |T_{\bgamma}^{\bomega_F(z)}|_{H^1(\gamma)}\big) d_{NPT}(\hat \omega_F(z),\omega_F(z)).
\end{align*}
Thus, Theorems~\ref{thm:frechet-oracle-rate} and \ref{thm:frechet-real-rate} deduce the desired rates of Corollary~\ref{cor:wasserstein-rate}.
\qed

\section{Verification of Assumptions \ref{assump:correlation-uniqueness} and \ref{assump:local-convexity-radius}}\label{sec:supp-verify-assump-equicorr}

In this section, we verify Assumptions~\ref{assump:correlation-uniqueness} and \ref{assump:local-convexity-radius} in several settings: when $S$ is supported in the bivariate correlation family, and when $Z$ and $S$ are independent or weakly dependent. The idea is to verify the \textit{convexity} of the objective function $F(z,\cdot)$ under nondegeneracy Assumption~\ref{assump:correlation-mineigen}. While we here focus on this global property, we note that Assumptions~\ref{assump:correlation-uniqueness} and \ref{assump:local-convexity-radius} are substantially milder than requiring global convexity.

\subsection{Bivariate case}

In the bivariate case, we follow a similar argument in the empirical case discussed in Section~\ref{sec:supp-one-step-bivariate}. Any correlation matrix $\s \in \Ecal_2$ is parameterized by a single correlation coefficient $\rho \in [-1, 1]$, given by
$$ \s(\rho) = \begin{pmatrix} 1 & \rho \\ \rho & 1 \end{pmatrix}. $$
The tangent space $\Tcal$ corresponds to matrices of the form $M = \begin{pmatrix} 0 & h \\ h & 0 \end{pmatrix}$ for $h\in\R$, with squared Frobenius norm $\|M\|_F^2 = 2h^2$. We require Assumption~\ref{assump:correlation-mineigen}, which implies the minimizer $\rho_F(z)$ (corresponding to $ S_F(z)$) is uniformly bounded away from the boundary values $\pm 1$.

Let $S = \s(r)$ and $\s(\rho)$ be defined as above with $r, \rho \in (-1, 1)$. The squared Bures--Wasserstein distance can be written as
$$ \Bcal^2(r, \rho) = 4 - 2\left(\sqrt{1+r}\sqrt{1+\rho} + \sqrt{1-r}\sqrt{1-\rho}\right). $$
Then, the first and second derivatives of the population objective function $F(z, \s(\rho)) = \E[s(Z, z) \Bcal^2(S, \s(\rho))]$ with respect to $\rho$ are
\begin{align*}
    F'(z, \rho) &= -A(z)(1+\rho)^{-1/2} + D(z)(1-\rho)^{-1/2}, \\
    F''(z, \rho) &= \frac{1}{2}A(z)(1+\rho)^{-3/2} + \frac{1}{2}D(z)(1-\rho)^{-3/2},
\end{align*}
where $A(z) := \E[s(Z, z)\sqrt{1+r}]$ and $D(z) := \E[s(Z, z)\sqrt{1-r}]$. In the following, we show that $A(z)$ and $D(z)$ cannot be nonpositive under Assumption~\ref{assump:correlation-mineigen}, implying the global convexity of $F$ and verifying the assumptions.

\begin{proposition}\label{prop:bivariate-verification}
    Under Assumption~\ref{assump:correlation-mineigen}, Assumptions~\ref{assump:correlation-uniqueness} and \ref{assump:local-convexity-radius} hold in the bivariate case ($d=2$). 
\end{proposition}

\begin{proof}
    \noindent\textit{Step 1. Positivity of $A(z)$ and $D(z)$, and verification of Assumption~\ref{assump:correlation-uniqueness}.}
    As in Section~\ref{sec:supp-one-step-bivariate}, if $A(z)$ and $D(z)$ are not both positive, then $F(z,\cdot)$ is concave, monotone, or flat, and its minimizer set contains boundary points of $[-1,1]$. Assumption~\ref{assump:correlation-mineigen} implies $S_F(z)=\s(\rho_F(z))$ uniformly lies in the interior, hence $|\rho_F(z)|<1$, so boundary minimizers are impossible. Therefore $A(z)>0$ and $D(z)>0$, which in turn gives $F''(z,\rho)>0$ for all $\rho\in(-1,1)$ and uniqueness of $\rho_F(z)$---this verifies Assumption~\ref{assump:correlation-uniqueness}.

    \vspace{3mm}
    \noindent\textit{Step 2. Uniform Bounds on Coefficients.}
    Note that 
    $$A(z) = \E[\sqrt{1+r}] + \E[(Z-\E[Z])^\top \sqrt{1+r}]\cov(Z)^{-1}(z - \E[Z])$$ 
    is an affine function of $z$. Since $A(z) > 0$ for all $z$ in the compact set $\{z: \|z\| \le B\}$ (as established in Step 1), its minimum on $\{z: \|z\| \le B\}$ is attained and must be strictly positive.
    Thus, there exists $c > 0$ such that $\min_{\|z\| \le B} A(z) \ge c$ and $\min_{\|z\| \le B} D(z) \ge c$.

    \vspace{3mm}
    \noindent\textit{Step 3. Verification of Assumption~\ref{assump:local-convexity-radius} (Uniform Local Convexity).}
    At the unique minimizer $\rho_F=\rho_F(z)$, the second derivative is
    $$ F''(z, \rho_F) = \frac{1}{2}\frac{A(z)}{(1+\rho_F)^{3/2}} + \frac{1}{2}\frac{D(z)}{(1-\rho_F)^{3/2}}. $$
    Since $1 \pm \rho_F \le 2$, using the uniform lower bound $c$ for $A(z)$ and $D(z)$,
    $$ F''(z, \rho_F) \ge \frac{1}{2}\frac{c}{2^{3/2}} + \frac{1}{2}\frac{c}{2^{3/2}} = \frac{c}{2\sqrt{2}} =: \kappa/2 > 0. $$
    Translating to the tangent space: for $M \in \Tcal \subset\Scal_2$ with $\|M\|_F^2 = 2h^2$, %
    $$ H(z)[M, M] = F''(z, \rho_F) h^2 \ge \frac{\kappa}{2} h^2 = \frac{\kappa}{4} \|M\|_F^2. $$
    This verifies Assumption~\ref{assump:local-convexity-radius}, completing the proof.
\end{proof}

\subsection{Independent and weakly dependent cases}

Next, we verify Assumptions~3 and 4 for the case where the predictor $Z$ and the latent correlation $S$ are independent, then extend the argument to the weakly dependent case.

\begin{proposition}\label{prop:independence-verification} 
    Suppose that $Z$ and $S$ are independent. Under Assumption~2, Assumptions~3 and 4 hold. 
\end{proposition}

\begin{proof} Using the independence between 
$Z$ and $S$, the population objective function decouples as \begin{align*}
    F(z, \s) &=  \E[s(Z, z)] \E[\Bcal^2(S, \s)],\\
     &= \E[\Bcal^2(S, \s)]
\end{align*} 
where the second equality follows from $\E[s(Z, z)] = 1$. 
Thus, the Fr\'echet regression function $S_F(z)$ becomes the constant Fr\'echet mean $\bar{S} := \argmin_{\s \in \Ecal_d} \E[\Bcal^2(S, \s)]$. The uniqueness Assumption~\ref{assump:correlation-uniqueness} of $\bar{S}$ follows from the \textit{strong convexity} of the squared BW metric on the set of positive definite matrices \citep{bhatiaBuresWassersteinDistance2019}, guaranteed under Assumption 2.

To verify Assumption~4, consider the Hessian $H_0 = \E[\nabla^2 \Bcal^2(S, \s)]|_{\s=\bar{S}}$ of $F(z, \cdot)$ at $\bar{S}$. Since $S \succeq \lambda_0 I_d$ almost surely, each functional $\s \mapsto \Bcal^2(S, \s)$ is strongly convex on $\Ecal_d^{++}$, implying that $H_0$ is strictly positive definite on the tangent space $\Tcal \subset \Scal_d$:
$$ \min_{M \in \Tcal, \|M\|_F=1} H_0[M, M] =: \kappa > 0; $$
see the weak dependence case below for the explicit lower bound for $\kappa$.
This verifies Assumption~4.
\end{proof}

\paragraph*{Weak dependence case}
We extend the argument of the independent case to the case where $Z$ and $S$ have small dependence, with some concrete mathematical derivations. We characterize the dependence as the covariance between $Z$ and the Hessian--quadratic form $\nabla^2_{\s}\Bcal^2(S,\s)[M,M]$ of the BW-based functional $\Bcal^2(S,\cdot)$. Specifically, define the curvature at $\s$ along the direction $M\in\Tcal$ as:
$$
\kappa_{\s,M}(S):=\nabla^2_{\s}\Bcal^2(S,\s)[M,M],\qquad \s\in\Ecal_d(\lambda_0/2),\ M\in\Tcal,\ \|M\|_F=1.
$$

We fist show that this term is positive.
By equation (A.2) of \citet{kroshninStatisticalInferenceBures2021}, 
\begin{align*}
    \kappa_{\s,M}(S) = 
    \sum_{i, j=1}^d\frac{\varDelta_{ij}^2}{\sqrt{\lambda_i\lambda_j}(\sqrt{\lambda_i} + \sqrt{\lambda_j})},
\end{align*}
where $\Lambda=\diag(\lambda_1,\ldots,\lambda_d)$ are eigenvalues of $S^{1/2}\s S^{1/2}$ satisfying $U\Lambda U^\top = S^{1/2}\s S^{1/2}$ and $\varDelta=U^\top S^{1/2}MS^{1/2}U$. Since $\|S^{1/2}\s S^{1/2}\|_{op} \le \|S\|_{op} \|\s\|_{op}$, we have the rough bound 
$$\lambda_i \le \lambda_{\max}(\s) \lambda_{\max}(S) \le d\, \lambda_{\max}(\s).$$
Also, 
$$\|\varDelta\|_F = \|S^{1/2}MS^{1/2}\|_F \ge  \lambda_{\min}(S) \|M\|_F  \ge \lambda_0\|M\|_F$$ 
by applying the inequality $\|AB\|_F \ge \lambda_{\min}(A) \|B\|_F$ twice (for symmetric positive semidefinite $A, B$). These bounds yield an $S$-independent lower bound on the curvature:
$$\kappa_{\s, M}(S) \ge \frac{\lambda_0^2 \|M\|_F^2}{2d^{3/2} \lambda^{3/2}_{\max}(\s)} \ge \frac{\lambda_0^2}{2d^3}\|M\|_F^2,$$
and hence the worst-case expected curvature satisfies
$$\kappa_0:= \inf_{\s\in \Ecal_d} \inf_{M\in\Tcal,\,\|M\|_F=1} \E[\kappa_{\s, M}(S)] \ge \frac{\lambda_0^2}{2d^3} > 0. $$

By standardizing the predictor $Z$, we assume $\E[Z] = 0$ and $\cov(Z) = I_p$ so that $s(Z, z) = 1 + Z^\top z$.
Then,
$$ H(z, \s)[M,M]:= \nabla^2_{\s}F(z,\s)[M,M]=\E[\kappa_{\s,M}(S)] + z^\top \E[Z\,\kappa_{\s,M}(S)]. $$
The first term is the baseline positive curvature. Thus, the convexity of $F(z, \cdot)$ is determined by the dependence-driven perturbation $\E[Z\,\kappa_{\s,M}(S)]$, which is exactly the covariance between $Z$ and the scalar curvature $\kappa_{\s, M}(S)$.
Define
$$
\varepsilon_{\mathrm{dep}}:=\sup_{\s\in\Ecal_d(\lambda_0/2)}\sup_{M\in\Tcal,\|M\|_F=1}\|\E[Z\,\kappa_{\s,M}(S)]\|,
$$
and $C:=\sup_{\|z\|\le B}\| z\|$. Here, $\varepsilon_{\mathrm{dep}}$ measures the maximal predictor--curvature covariance, which vanishes when $Z$ and $S$ are independent. When $C\varepsilon_{\mathrm{dep}} < \kappa_0$, we have
\begin{align*}
    \inf_{\|z\|\le B} \inf_{\s\in\Ecal_d(\lambda_0/2)} \inf_{M\in\Tcal,\,\|M\|_F=1} H(z, \s)[M, M] & \\
    \ge \inf_{\|z\|\le B} \inf_{\s\in\Ecal_d(\lambda_0/2)}\inf_{M\in\Tcal,\,\|M\|_F=1} &\E[\kappa_{\s, M}(S)] - \|z\|\, \|\E[Z\,\kappa_{\s, M}(S)]\| &\\
    \ge \kappa_0 - C \varepsilon_{\mathrm{dep}} > 0.& & 
\end{align*}
Hence Assumption~4 holds, because Assumption~2 places minimizers $S_F(z)$ in $\Ecal_d(\lambda_0)\subset \Ecal_d(\lambda_0/2)$ uniformly over $\|z\|\le B$. Moreover, the same lower bound implies that $F(z,\cdot)$ is strongly convex on the convex set $\Ecal_d(\lambda_0/2)$. Since $S_F(z)\subset\Ecal_d(\lambda_0)$ from Assumption~\ref{assump:correlation-mineigen}, any other minimizer does not lie outside $\Ecal_d(\lambda_0/2)$. We thus have the uniqueness Assumption~\ref{assump:correlation-uniqueness}. 
Note that independence is the special case $\varepsilon_{\mathrm{dep}}=0$. The following proposition summarizes the result:
\begin{proposition}\label{prop:weak-dependence-verification}
    Suppose that $\E[Z] = 0$ and $\cov(Z) = I_p$ by standardization. If $C\, \varepsilon_{\mathrm{dep}} < \kappa_0$ where $C := \sup_{\|z\|\le B} \|z\|$, then under Assumption~2, Assumptions~3 and 4 hold.
\end{proposition}

This result also suggests that, by setting $B$ small enough so that the domain of $z$ is near the mean $\E[Z]=0$, we can similarly verify Assumptions~\ref{assump:correlation-uniqueness} and \ref{assump:local-convexity-radius}.

\section{Auxiliary simulation results}\label{sec:supp-auxiliary-simulation}

Following the same settings as in Section~\ref{sec:simulations}, we evaluate the performance of the regression fits from NPT-FR, Marginal-FR, and Gaussian-FR using the Wasserstein distance in the MSPE on an independent test set. The corresponding MSPE is defined as
$$\text{MSPE}_{\text{wass}} = \frac{1}{n_{\text{te}}} \sum_{l=1}^{n_{\text{te}}} d_W^2(\hat  \omega_F(Z_l^{\text{te}}), \omega_l^{\text{te}}),$$
which is not decomposable into marginal and correlation components. Since the Wasserstein distance has no closed form computation, we approximate each $d_W^2(\hat  \omega_F(Z_l^{\text{te}}), \omega_l^{\text{te}})$ by generating $N_{err} = 2000$ i.i.d. samples from each $\omega_F(Z_l^{\text{te}})$ and $\omega_l^{\text{te}}$, followed by applying the \texttt{ot.emd2} function in the Python Optimal Transport library \citep{flamaryPOTPythonOptimal2021}.

\begin{figure}
    \centering
    \includegraphics[width=0.98\linewidth]{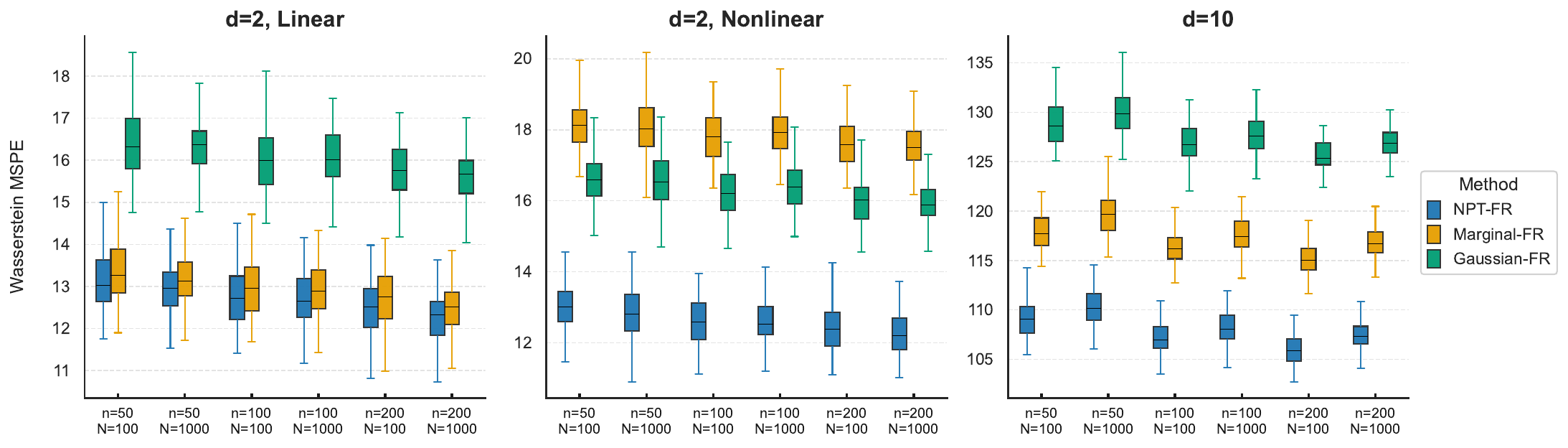}
    \caption{Out-of-sample mean squared prediction errors in $d_W$ ($\text{MSPE}_{\text{wass}}$) on $n^{\text{te}}=500$ test samples over $n_{\text{rep}}=100$ Monte Carlo replicates. Columns correspond to $d=2$ with linear correlation, $d=2$ with nonlinear correlation, and $d=10$. 
    }
    \label{fig:simulation-wasserstein-boxplots}
\end{figure}

Figure~\ref{fig:simulation-wasserstein-boxplots} presents boxplots of $\text{MSPE}_{\text{wass}}$. Overall, NPT-FR continues to outperform the other methods, with Marginal-FR performing comparably in the linear bivariate case due to a small correlation error introduced in that setting. Performance generally improves as $n$ increases across all scenarios. Unexpectedly, we observe a slight deterioration in performance as $N$ increases when $d=10$, despite the improvement in $\text{MSPE}_{\text{corr}}$ for NPT. This inconsistent deterioration suggests less numerical stability and inflated sampling error in the $d=10$ setting, where the Wasserstein distance suffers from the curse of dimensionality and is thus less reliable.

\putbib
\end{bibunit}

\end{document}